\documentclass[prd,aps,amsfonts,eqsecnum,superscriptaddress,nofootinbib,notitlepage,longbibliography,11pt]{revtex4-1}
\usepackage[letterpaper, left=1in, right=1in, top=0.9in, bottom=0.8in]{geometry}

\usepackage{packages}
\newcommand{\AC}[1]{ 
}
\newcommand{\FB}[1]{
}
\newcommand{\AB}[1]{ 
}
\newcommand{\PMH}[1]{
}
\newcommand{\jd}[1]{
}

\begin{document}
\vspace{-1.3cm}
\title{
Efficient unitary designs and pseudorandom unitaries from permutations
}


\begin{abstract}

In this work we give an efficient construction of unitary $k$-designs using $\tilde{O}(k\cdot\poly(n))$ quantum gates, as well as an efficient construction of a parallel-secure pseudorandom unitary (PRU). Both results are obtained by giving an efficient quantum algorithm that lifts random permutations over $\vS(N)$ to random unitaries over $\vU(N)$ for $N=2^n$. In particular, we show that products of exponentiated sums of $\vS(N)$ permutations with random phases approximately match the first $2^{\Omega(n)}$ moments of the Haar measure. By substituting either $\tilde{O}(k)$-wise independent permutations, or quantum-secure pseudorandom permutations (PRPs) in place of the random permutations, we obtain the above results. The heart of our proof is a conceptual connection between the large dimension (large-$N$) expansion in random matrix theory and the polynomial method, which allows us to prove query lower bounds at finite-$N$ by interpolating from the much simpler large-$N$ limit. The key technical step is to exhibit an orthonormal basis for irreducible representations of the partition algebra that has a low-degree large-$N$ expansion. This allows us to show that the distinguishing probability is a low-degree rational polynomial of the dimension $N$.
\end{abstract}

\author{Chi-Fang Chen}
\email{achifchen@gmail.com}
\affiliation{Institute for Quantum Information and Matter,
California Institute of Technology, Pasadena, CA, USA}

\author{Adam Bouland}
\affiliation{Department of Computer Science, Stanford University}

\author{Fernando G.S.L. Brand\~ao}
\affiliation{Institute for Quantum Information and Matter, California Institute of Technology, Pasadena, CA, USA}
\affiliation{AWS Center for Quantum Computing, Pasadena, CA}
\author{Jordan Docter}
\affiliation{Department of Computer Science, Stanford University} 
\author{Patrick Hayden}
\affiliation{Stanford Institute for Theoretical Physics}
\affiliation{Department of Physics, Stanford University}
\author{Michelle Xu}
\affiliation{Stanford Institute for Theoretical Physics}
\affiliation{Department of Physics, Stanford University}
\maketitle

\vspace{-1.1cm}
\tableofcontents\newpage

\setcounter{page}{1}
\section{Introduction}

Pseudorandom states and unitaries are fundamental tools in
quantum information.
By efficiently creating ensembles of states/unitaries that mimic the Haar measure, we can access Haar random properties without paying a cost exponential in the number of qubits.
Broadly speaking, two types of quantum pseudorandomness have been previously considered.
The first is \textit{information-theoretic} pseudorandomness, where the goal is to generate unitary or state $k$-designs, i.e., ensembles that match the first $k$ moments of the Haar measure. 
These are the natural quantum analogs of $k$-wise independent functions or permutations.
This yields information-theoretic $k$-copy security, and suffices for many applications, such as randomized benchmarking \cite{emerson2005scalable,knill2008randomized,dankert2009exact}, cryptography~\cite{divincenzo2002quantum,ambainis2009nonmalleable}, shadow tomography \cite{huang2020predicting}, communication~\cite{hayden2008decoupling,szehr2013decoupling}, and phase retrieval~\cite{kimmel2017phase}. The second is \textit{computational} pseudorandomness, where the goal is to create pseudorandom states and unitaries (PRSs and PRUs) that are computationally indistinguishable from Haar \cite{ji2018pseudorandom}. These are the quantum analogs of pseudorandom functions (PRFs) or permutations (PRPs), and have found significant applications both in quantum cryptography (e.g., \cite{kretschmer2021quantum,kretschmer2023quantum}) and in the complexity of physical systems, e.g., \cite{bouland2019computational,kim2020ghost,aaronson2024quantum,yang2023complexity}.
This relaxed security notion allows one to obtain properties impossible in information-theoretic settings, such as small ensemble (key) sizes, many copy security, and low entanglement \cite{ji2018pseudorandom,aaronson2024quantum}.

Many questions remain open in both branches of the quantum pseudorandomness family tree.
On the $k$-design side, a significant open problem has been to efficiently generate unitary $k$-designs on $n$ qubits. It is known that this task requires at least $\Omega(nk)$ quantum gates \cite{brandao2016local} but so far existing constructions have not achieved the bound, despite much work in the area ~\cite{webb2015clifford,zhu2016clifford,zhu2017multiqubit,mitsuhashi2023clifford, harrow2009random,harrow2009efficient, brandao2016local,harrow2023approximate,haferkamp2022random,hunter2019unitary, dankert2009exact,cleve2015near,onorati2017mixing,nakata2017efficient,bannai2019explicit, nakata2021quantum, bannai2022explicit,o2023explicit,kaposi2023constructing,haah2024efficient,chen2024efficient}.
The first efficient and systematic constructions of approximate unitary $k$-designs which work for large values of $k$ were given by \cite{brandao2016local}, who showed that random local circuits achieved the goal in $\CO(k^{10.5} n^2)$ time. This was followed by an improvement to $\tCO(k^5 n^2)$ by \cite{haferkamp2022random} using an improved analysis.
Very recently, this was improved to $\tilde{O}(k^2n^2)$ quantum gates by two independent works with different constructions \cite{haah2024efficient,chen2024efficient}, which is still quadratically off from the lower bound.
This stands in contrast to related $k$-wise independent objects, as optimal constructions with linear scaling in $k$ are known for quantum state designs (e.g., \cite{brakerski2019pseudo}) as well as for analogous classical objects, namely $k$-wise independent functions \cite{Joffe,wegman1981new,ALON1986fastsimple} or approximate $k$-wise independent permutations \cite{Kassabov2007Sym,caprace2023tame}. We note that linear scaling\footnote{We emphasize that this refers to the quantum gate count -- which is typically the bottleneck in applications -- rather than the scaling in the number of bits of classical randomness which was recently achieved by \cite{o2023explicit}. } in $k$ has previously been shown\footnote{We note there have also been plausible arguments that certain continuous time Brownian motions should attain linear scaling \cite{jian2023linear} but the cost of simulating them on a quantum computer remains to be analyzed.} only in restricted cases, such as the limit of large local dimension \cite{haferkamp2021improved}, or if the number of moments matched is small ($k\ll\CO(\sqrt{n})$) \cite{nakata2017efficient}. This problem is not only a fundamental and natural question in theoretical computer science, but has also gathered the attention of the quantum gravity community, because the linear growth in circuit complexity associated with an efficient $t$-design ensemble (see e.g., \cite{roberts2017chaos}) is believed to play a role in resolving certain paradoxes in quantum gravity \cite{brown2018second}.

On the \textit{computational} pseudorandomness side, a significant open problem is to construct pseudorandom unitaries (PRUs) from standard cryptographic assumptions, such as the existence of quantum-secure one-way functions (OWFs).
While several efficient constructions of pseudorandom state ensembles are known \cite{ji2018pseudorandom,brakerski2019pseudo,aaronson2024quantum,giurgica2023pseudorandomness,jeronimo2023subset,macommunication}, progress has been much more difficult in the unitary case.
While several PRU candidates have been proposed \cite{ji2018pseudorandom}, none of them have proven security.
However, a number of objects intermediate between a PRS and a PRU have been rigorously constructed.
For example, Ananth et al. constructed pseudorandom function-like states, a generalization of a PRS that allows one to create polynomially many independent PRSs \cite{ananth2022cryptography,ananth2022pseudorandom}.
Subsequently, Lu et al. constructed what they call pseudorandom state scramblers, which are ensembles of unitaries that produce a PRS from any \textit{fixed} input state \cite{lu2023quantum}.
There has also been a recent construction of parallel-secure pseudorandom isometries between spaces of differing dimensions \cite{ananth2023pseudorandom}, as well as another variant of a pseudorandom state scrambler for real (orthogonal) states~\cite{brakerski2024real}.
However, the existence of efficient pseudorandom unitaries with general query security remains open.

In this work, we make progress on both quantum pseudorandomness open problems simultaneously. First, we give a construction of a unitary $k$-design using quantum gates scaling near-linearly with $k$, which matches the lower bound on $k$-dependence up to logarithmic factors.
\begin{thm}[Efficient quantum $k$-designs]\label{thm:intro_kdesign_from_kwise}
    There exists an efficient quantum algorithm to generate an $\epsilon$-approximate unitary $k$-design (in diamond norm) using $\tCO(k\,\poly(n))$ one and two-qubit quantum gates. The $\tCO(\cdot)$ absorbs polylogarithmic dependence on $n, k, \epsilon^{-1}.$
\end{thm}

Second, we also give a construction of a pseudorandom unitary ensemble with nonadaptive security.
\begin{thm}[Parallel PRU from quantum-secure OWF]\label{thm:intro_PRU_from_OWF}
    The existence of a one-way function secure against quantum attack implies an efficient quantum algorithm to generate parallel-secure pseudorandom unitaries.
\end{thm} 

To achieve the above results, we demonstrate that the following meta-result:
\begin{align}
    \text{\textit{ (Pseudo)random permutations can be efficiently lifted to (pseudo)random unitaries}.}
\end{align}
More formally, we show there is an efficient quantum algorithm which, given black-box access to random permutations (and their inverses) in $\vS(N)$ for $N=2^n$, creates an ensemble of $n$-qubit unitaries that is close to the Haar measure in an appropriate sense.
In particular, we show that the ensemble is an approximate unitary $k$-design for a large value of $k$, namely $k=2^{\Omega(n)}$.
Our main results are then obtained by substituting different forms of pseudorandomness for the black box permutations from $\vS(N)$. To obtain our $k$-design result, we show that it suffices to substitute in approximate $\tilde{O}(k)$-wise independent permutations~\cite{Kassabov2007Sym,caprace2023tame}. (See~\autoref{thm:efficient_k_wise}.) Here, the key point is that high-precision $k$-wise independent permutations are known to be implementable in $\CO(\poly(n)k)$-time, so this efficiency then directly lifts to our unitary construction. To obtain our parallel PRU, we substitute in quantum secure pseudorandom permutations (PRPs, inverse allowed), which can be derived from quantum-secure one-way functions \cite{zhandry2016note}.

\subsection{Main results}

To describe our main construction, the basic building blocks are \textit{random phased permutations},
random permutation with random \textit{complex} signs:
\begin{align}
    \vZ &:= \vD_{z} \cdot \vS,\quad \text{where}\notag \tag{random phased permutations}\\
    \vS&\stackrel{unif}{\sim} \vS(N) \tag{uniformly random permutations on $N$ elements}\\
    \vD_{z}&: (\vD_{z})_{ij} = \delta_{ij} z_i\quad\text{where} \quad z_i\stackrel{i.i.d.}{\sim} \vU(1).\tag{random diagonal complex phases}
\end{align}
While we will eventually pursue pseudorandomness, in most of our discussions, it will suffice to treat the $\vZ$ as truly random. Eventually, they will be replaced as needed by suitable classical pseudorandom counterparts.

The central object we build from random phased permutations $\vZ$ is the matrix exponential
\begin{align}
    \e^{\ri \theta_m \vA_{m}} \quad \text{for}\quad \theta_m = \CO(1),
\end{align}
where these $2m$-sparse Hermitian matrices $\vA_m$ are made of i.i.d. copies of $\vZ$ and their adjoints
\begin{align}
        \vA_m &:= \frac{1}{\sqrt{2m}}\sum_{a=1}^m (\vZ_a+\vZ^{\dagger}_a) \quad \text{where}\quad \vZ_a\stackrel{i.i.d.}{\sim} \vZ,
\end{align}
which can be thought of as the adjacency matrix for a random graph weighted by random phases. 

The main result of this work is that the product of a few independent copies of these exponentiated sparse random matrices $\e^{\ri \theta_m \vA_{m}}$, left and right multiplied by independent phased permutations $\vZ_L$ and $\vZ_R$, approximates a Haar random unitary for quantum computers making nonadaptive (i.e., parallel) queries. Formally, this nonadaptive query access model considers the quantum channel associated with the $k$-fold tensor product of random unitary $\vU$
\begin{align}
    \CN_{2k,\vU}[\vrho]:= \BE_{\vU} [ \vU^{\otimes k}\vrho\vU^{\dagger \otimes k} ] \quad \text{with the diamond norm}\quad     \normp{\CN_1-\CN_2}{\diamond}: = \normp{ (\CN_1-\CN_2) \otimes \CI}{1-1},
\end{align}
to quantify the probability of distinguishing our construction from Haar using \textit{arbitrary} entangled inputs and measurements. We say that an ensemble of unitaries is an $\epsilon$-approximate $k$-design (see, e.g.,~\cite{low2010pseudo}) if the induced channel is close to the $k$-fold Haar channel in the diamond norm. That is, no possible quantum experiment (with arbitrary computation and ancillas) can distinguish the ensemble from the Haar measure given parallel access to $k$ copies of the unitary.

\begin{defn}[Approximate unitary $k$-designs]\label{defn:approx_kdesign}
    An ensemble of unitaries $\vU$ is an $\epsilon$-approximate $k$-design if $\normp{\CN_{2k,\vU}-\CN_{2k,\vU_{\text{Haar}}}}{\diamond} \leq \epsilon,$ where $\vU_{\text{Haar}}$ is drawn from Haar random unitaries $\vU(N)$.
\end{defn}

We can now state the main technical result: the product of exponentials of sums of random phased permutations is an $\epsilon$-approximate $k$-design for large values of $k$ and small values of $\epsilon$.

\begin{restatable}[Unitary designs from product of sparse exponentials]{thm}{parallelPRU}\label{thm:parallelPRU}
    For each $m,\ell \in \mathbb{Z}^{+}$, let $\vV$ be the $N$-dimensional random unitary formed by a product of the $\ell$ independent exponentials of the random matrices $\vA^{(1)}_m, \vA^{(2)}_m,\cdots, \vA^{(\ell)}_m$, left and right multiplied by random permutations $\vZ_L, \vZ_R$, i.e.
    \begin{align}
        \vV:=\vZ_L \L( \prod_{j=1}^{\ell} \e^{\ri \theta_m \vA^{(j)}_m} \R)\vZ_R,
    \end{align}
    for some $m$-dependent angle $\theta_m= \CO(1)$. Then, $\vV$ is an approximate $k$-design in diamond distance. More precisely, for $N = 2^n$,
    \begin{align}
        \normp{\CN_{2k,\vV}-  \CN_{2k,\vU_{Haar}}}{\diamond} \le \CO\L(\frac{k^8\ell^8(m^4n^8+ n^{8})}{N}+\frac{k^4}{m^{\ell}}\R).
    \end{align}
\end{restatable}
The point is that the ensemble is an approximate $k$-design for a \textit{superpolynomial} value of $k$, using only very few, specifically $\tCO(m\ell)$, random permutations from $\vS(N)$. Since sampling uniformly random permutations requires exponentially many bits of randomness, we are not evading known lower bounds on the cardinality of a unitary design \cite{brandao2016local,roberts2017chaos}.
Crucially, however, the number of iterations $\ell$ can be small. For example, setting $m=2$ and $\ell=n$ yields an approximate  $k$-design for very large value of $k=2^{\Omega(n)}$, but
even $m=2$ and $\ell = \log^2(n)$ suffices for a superpolynomial design; on the contrary, spectral gap approaches (e.g.,~\cite{brandao2016local}) to unitary designs often cost $\Omega(nk)$ rounds of products. 

While truly random permutations are still costly to sample from, the key point is that by substituting \emph{pseudorandom} permutations in their place, we establish the link between classical pseudorandom permutations and quantum pseudorandom unitaries, as seen from the following two corollaries, establishing the aforementioned~\autoref{thm:intro_kdesign_from_kwise} and~\autoref{thm:intro_PRU_from_OWF}. First, we consider substituting information-theoretically secure pseudorandom permutations (i.e., $k$-wise independent permutations) in place of the truly random permutations to obtain an efficient unitary $k$-design. This requires choosing suitable parameters $m$ and $\ell$ (see~\autoref{cor:k_design_from_k_wise_full} for completeness).
\begin{restatable}[Unitary $k$-designs from $k'$-wise independence (informal)]{cor}{wisetodesign}\label{cor:k_design_from_k_wise}
There is a constant $c$ such that for $k \le 2^{cn}$, an $\epsilon$-approximate quantum unitary $k$-design can be efficiently implemented by applying our construction with
 $k'$-wise independent (discrete) phases and $k'$-wise independent permutations in place of the truly random phases/permutations, where
 \begin{align}
    k' = \CO(k\log(k/\epsilon)).
\end{align}
\end{restatable}

Here, we are making key use of the fact that our construction only relies on $k'=\tilde{\CO}(k)$-th moments of the permutations via efficient Hamiltonian simulation algorithms (e.g. \cite{berry2015simulating,gilyen2019quantum}).
Thus our result ``lifts''
the efficient construction of $k'$-wise independent permutations and functions to the construction of unitary $k$-designs. This is in a similar spirit to recent results of Brakerski and Shmueli that lift $2k$-wise independent functions to approximate quantum state $k$-designs \cite{brakerski2019pseudo}.

To achieve an efficient unitary design, we now need to leverage the fact that
$k'$-wise independent permutations and functions can be implemented in only $\CO(\poly(n)k')$ time.
 However, one caveat is that the best known classical $k'$-wise independent permutations (\autoref{defn:classical_kwise}) are only approximately $k'$-wise independent. Fortunately, explicit, high-accuracy classical $k$-wise independent permutations do exist, which are of sufficiently low error to not affect the construction.
 This gives us unitary $k$-designs that are nearly algorithmically and information-theoretically optimal (in terms of $k$ dependence and in the diamond distance). See \autoref{sec:optimal_k_design} for further discussion.
\begin{cor}[Efficient quantum $k$-designs with almost linear $k$-scaling] \label{cor:efficient-k-designs}
    There is a constant $c$ such that for $k \le 2^{cn}$, an $\epsilon$-approximate unitary $k$-design can be implemented using 
\begin{align}
\tCO(\poly(n) k )&\quad \text{one and two-qubit gates}\quad \text{and}\quad \CO(nk) \quad \text{bits of classical randomness}
\end{align}
where the notation $\tCO$ absorbs factors polylogarithmic in $n,k$ and $1/\epsilon.$ Thus, the $k$-dependence in the gate complexity and the seed length are both (nearly) optimal. 
\end{cor}
Previously, the best known systematic construction of unitary $k$-designs for superpolynomially large $k$ has gate complexity $\tCO(n^2k^2)$ ~\cite{haah2024efficient,chen2024efficient}. 

Moving on to computational pseudorandomness, we substitute cryptographically secure pseudorandom permutations (i.e., PRPs) in place of the truly random permutations to obtain a parallel-secure PRU:

\begin{cor}[Quantum secure-PRP implies parallel-PRU] Suppose quantum-secure pseudorandom permutations exist. Then our construction $\vV$ gives a pseudorandom unitary under nonadaptive queries.
\end{cor}
\begin{proof}
Implement our unitary $\vV$ using any efficient quantum algorithm given query access to polynomially many $\vZ_i$ to super-polynomial precision. We obtain computational pseudorandomness from information-theoretic pseudorandomness by replacing the random phased permutations with the pseudorandom alternatives and applying a simple hybrid argument. Note that the controlled permutations can be efficiently implemented given function queries of permutations (and their inverses).
\end{proof}

Note that the existence of quantum-secure PRPs (with inverses) only requires assuming the existence of quantum-secure pseudorandom functions~\cite{zhandry2016note}, which is a standard cryptographic assumption. Furthermore, if there exists a low-depth implementation of quantum-secure PRPs, then our PRU can also be low-depth with suitable parameters $m,\ell$ and standard quantum algorithmic implementations. 

\subsection{Proof idea}

In a nutshell, multiplying sparse matrices is an efficient way to get a dense matrix, but controlling the diamond distance from Haar requires careful analysis. Prior approaches to unitary $k$-designs have often been based on the spectral gaps of random walks (e.g.,~\cite{brandao2016local}), which bootstrap the statistical distance from a comparatively more tractable spectral gap. However, this approach necessarily requires a log-dimensional factor $\CO( \log(N^{k})) = \CO(nk)$ multiplicative blowup in the gate complexity due to the conversion from 2-norm to 1-norm. For cryptographic applications, the attacker may perform an arbitrary polynomial number of queries, so security requires a \textit{fixed} poly-time construction of a superpolynomial $k$-design, posing a fundamental barrier for spectral gap approaches.

The essence of our alternative argument can be captured in the following observation:
\begin{align} \label{eq:freeness}
    \vG \stackrel{(1)}{\approx} \frac{1}{\sqrt{m^{\ell}}} \sum^{m^{\ell}}_{i=1} \vZ_i\stackrel{(2)}{\approx}   \L(\frac{1}{\sqrt{m}} \sum^m_{i=1} \vZ^{(\ell)}_i\R)\cdots \L(\frac{1}{\sqrt{m}} \sum^m_{i=1} \vZ^{(1)}_i\R),
\end{align}
where the $\vZ_i$ are i.i.d. copies of a random phased permutation. The first approximation (1) is nothing more than a central limit theorem (CLT) established using a matrix Lindeberg argument similar to those in~\cite{chen2023sparse,chen2024efficient}, but it allows us to obtain a very nice random matrix: $\vG$ is drawn from the Ginibre ensemble~\cite{Ginibre1965}, which is a complex Gaussian matrix that is both left and right unitarily invariant. However, the CLT-type convergence rate is too slow (polynomial in the number of summands) and will incur a large $\Omega(\poly(k))$-cost scaling with the number of queries $k$.

The crux of our argument for circumventing the large number of i.i.d. copies is the second approximation (2): a product of sums reproduces the statistics of an \textit{independent} sum but using many fewer $(\tCO(m\ell))$ copies of $\vZ_i$, as stated in following loosely stated principle:
\begin{center}
\textit{In the large-$N$ limit, distinct words of permutations are effectively independent of each other.}
\end{center}

As an instructive example, when the dimension $N$ is large, the following correlated words of $\vZ_1,\vZ_2$ acting on $\ket{i}$ are almost independent of each other:
\begin{align}
    \L(\vZ_1\ket{i}, \vZ_2\ket{i}, \vZ_1\vZ_2\ket{i}, \vZ_2\vZ_1\ket{i} \R) \stackrel{dist}{\approx} \L(\vZ_1\ket{i}, \vZ_2\ket{i}, \vZ_3\ket{i}, \vZ_4\ket{i}\R).
\end{align}
Indeed, knowing $\vZ_1\ket{i}$ and $\vZ_2\ket{i}$ tells us nothing about $\vZ_2\vZ_1\ket{i}$, unless the very unlikely collision $\vZ_2\ket{i}\propto \ket{i}$ occurs. Applying this intuition to the product of sums, we can get $m^{\ell}$ independent $\vZ$s from merely $\tCO(m\ell)$ many truly independent $\vZ$s!

At first glance, the above analysis may appear strange, as if we get ``more randomness for free'' from a much smaller ($m \ell$) number of independent elements. Careful thought reveals the approximation in Eq.~\eqref{eq:freeness} is possible because we
only consider \emph{low-moment} properties of the permutations, i.e., we are fixing $k\ll N$ and then taking a large-$N$ limit.
In the $N\rightarrow\infty$ limit, the large amount of randomness in the permutations themselves effectively ``decouples'' the different terms, as collisions become vanishingly improbable. Strictly speaking, this precise statement only holds in a \emph{non-adaptive} setting, i.e., when the inputs are fixed in advance. Of course, sequential/adaptive queries would reveal correlations between these words -- for example, if one were able to query $\vZ_2$ \textit{after} knowing the result of $\vZ_1\ket{i}$, it would coincide with $\vZ_2\vZ_1\ket{i}$.
However, the key point is that under non-adaptive queries\footnote{We conjecture our ensemble also gives adaptive security, but we note this would require further proof ideas, such as defining a more refined notion of independence of different words.}, such attacks are not possible, and the words effectively decouple.

Unfortunately, proving the validity of the argument above in finite dimensions $N$ is nontrivial. There are nonasymptotic correlations, however tiny, between the distinct words because those words are made of only $\tCO(m\ell)$-many independent random phased permutations. To complete the proof, the remaining and most substantial technical argument is a framework to control those finite-$N$ corrections effectively.

In principle, since permutations are reasonably nice objects, we could brute force through the combinatorics using diagrammatic calculations. With more work, it is possible to go beyond the limit and perform a \textit{large-$N$ expansion}: for ``nice'' random matrices ensembles, the Weingarten calculus and its variants frequently yield sums over diagrams with coefficients being rational functions of $N$. To prove our results, however, it is necessary to control the total contribution coming from all orders in the $1/N$ expansion. Moreover, it is very challenging to systematically capture the fine-grained combinatorics and \textit{cancellations} required to deliver strong enough nonasymptotic bounds.

Our main strategy is the following large-$N$-interpolation principle that guides the path forward:

\begin{center}
\textit{``Nice'' random matrix properties at finite $N$ are controlled by the large-$N$ limit.}\\
\end{center}

Conceptually, we are looking for an \textit{interpolation} argument: instead of directly calculating a complicated random matrix quantity at finite $N$, we start with the much simpler large-$N$ limit, and control the finite-$N$ corrections by arguing that the function ``changes slowly'' as a function of $\frac{1}{N}$.

The crucial mathematical tool is a basic fact in polynomial approximation (see \autoref{sec:Markovs}): Markov's ``other'' inequality states that a real-valued, \textit{low-degree} polynomial $f(x)$ whose values are \textit{bounded} on an interval will have bounded derivatives. This seemingly innocent property of polynomials has found profound implications in establishing lower bounds in classical circuit~\cite{beigel1993polynomial} and quantum query complexity, a strategy referred to as the \textit{polynomial method}, e.g.,~\cite{beals2001quantum,aaronson2004quantum,barnum2001quantum,nayak1999quantum,kutin2005quantum,razborov2003quantum}. By identifying an often nonobvious interpolation parameter $x$ such that $f(x_1)$ and $f(x_2)$ correspond to the acceptance probabilities of the quantum algorithm on the two cases to distinguish, one can prove that the difference of the distinguishing probabilities $f(x_1)-f(x_2)$ is extremely small. While explicitly understanding \textit{all possible} adversarial quantum algorithms is essentially impossible, we often do have sufficient structural restrictions to guarantee that $f(x_1)-f(x_2)$ is a bounded degree polynomial which, remarkably, proves to be sufficient.

Our key insight is to draw a direct connection between the large-$N$ expansion in random matrix theory and polynomial approximation by setting the interpolation parameter to be (\autoref{fig:largeN_poly})
\begin{align}
    x = \frac{1}{N}\quad \text{to control the finite-$N$ corrections}\quad \labs{f\left(\frac{1}{\infty}\right)-f\left(\frac{1}{N}\right)},
\end{align}
which requires constructing a suitable function that fulfills the requirements of the polynomial methods, as the most substantial part of our proof.
\begin{figure}[t]
    \begin{center}
    \includegraphics[width=0.6\textwidth]{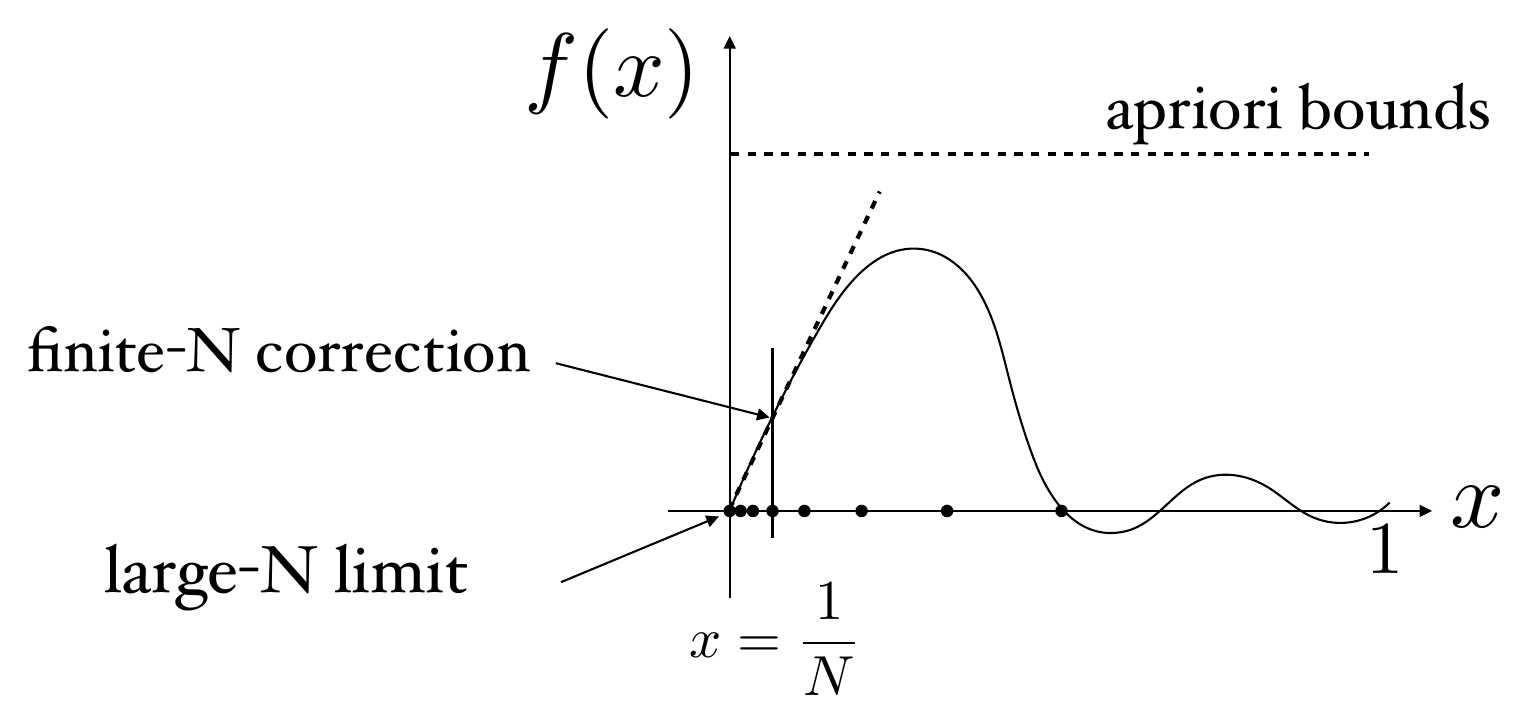}
    \end{center}
    \caption{Interpolating the $\frac{1}{N}$ expansion. For suitable quantities dependent on the dimension $N$, we control the finite-$N$ behavior by interpolating from the large-$N$ limit. This argument requires that (1) the quantity of interest be a low-degree rational function of $N$, (2) the function obeys a small apriori bound, and (3) the large-$N$ limit can be effectively calculated. 
    }
    \label{fig:largeN_poly}
\end{figure}

We further explain the above high-level intuition in the following sections.  The sum of permutations is not unitary, so we need to consider an exponential of sums, introducing an infinite series in~\autoref{sec:single_exp}; in~\autoref{sec:LargeN_main}, we argue that the distinct words $\vZ_1\vZ_3\cdots$ are independent in the large-$N$ limit; in~\autoref{sec:CLT_main}, we spell out the central limit theorem for independent sums; in~\autoref{sec:requirements_main}, we explain how to control the finite-$N$ corrections through interpolation, which is the most involved step.

\subsubsection{A single exponential as a sum of words}
\label{sec:single_exp}
To begin with, let us study one exponential of our sparse random matrices $\e^{\ri \theta\vA_m}$. The Taylor series expansion can be formally expressed as a weighted sum over \textit{words}:\footnote{This is also known as the \textit{free group} when there are no nontrivial relations between the generators.} any product of symbols $\vZ_1,\cdots ,\vZ_m$ and their inverses $\vZ_a^{-1} = \vZ_a^{\dagger}$ subject only to the trivial cancellations $\vZ_a\vZ^{\dagger}_a=\vZ^{\dagger}_a\vZ_a= \vI$, such as
\begin{align}
\vZ_1,\vZ_1^{\dagger},\vZ^{\dagger}_2\vZ_1, \vZ_1^3 \vZ_2^{\dagger} \vZ_1, \ldots
\end{align}
We denote the set of such words by $\CW(\{\vZ_a\}_{a=1}^m)$.
\begin{align}
    \e^{\ri \theta\vA_m} &= \vI + \frac{\ri\theta}{\sqrt{2m}}\sum_{a=1}^m \L(\vZ_a+\vZ^{\dagger}_a\R) - \frac{\theta^2}{2}\frac{1}{2m} \L(\sum_{a=1}^m \vZ_a+\vZ^{\dagger}_a\R)^2 + \cdots \\
    & = \sum_{i} w_i\vW_i \quad \text{where}\quad \vW_i \in \CW(\{\vZ_a\}_{a=1}^m).
\end{align}
The expression above is basically a sum over $m$ independent $\vZ_i$ and their adjoints but with higher-order corrections due to unitarity. Now, we take iterated \textit{products} of independent copies of these sparse exponentials 
\begin{align}
        \prod_{j=1}^{\ell} \e^{\ri \theta\vA^{(j)}_m} = \sum_{\vec{i}=(i_1,\cdots,i_{\ell})} w_{\vec{i}}\vW_{\vec{i}}\quad \text{where}\quad w_{\vec{i}} = (w_{i_1},\dots,w_{i_\ell}) \quad \text{and}\quad \vW_{\vec{i}}= \vW_{i_1}\cdots \vW_{i_{\ell}}.\tag*{(distinct words: $\vW_{\vec{i}}\ne \vW_{\vec{j}}$ for each $\vec{i}\ne \vec{j}$)}
\end{align}
We expect the sparsity (and the number of distinct words) to increase exponentially with $\ell$, so that merely $\ell \sim \log(N)$ iteration may already produce a very dense matrix. To see that the above product is anywhere close to Haar, the hints come from first analyzing the \textit{large-$N$ limit} where the structure drastically simplifies. 

\subsubsection{\texorpdfstring{Large-$N$ limit}{Large N limit}}
\label{sec:LargeN_main}

\autoref{lem:freeness_random_P} will establish that in the limit of infinite dimension $N$, fixing the length of the words and the number of queries $k$, all distinct nontrivial (i.e., non-identity) words can be simultaneously \textit{replaced} with \textit{independent} phased permutations. This means that we may replace our complicated, correlated sum with an independent sum in the diamond norm
\begin{align}
    \vZ_L \L(\sum_{\vec{i}} w_{\vec{i}}\vW_{\vec{i}}\R)\vZ_R&\stackrel{N\rightarrow \infty}{=} \vZ_L \L(\sum_{\vec{i}} w_{\vec{i}}\vZ_{\vec{i}}\R) \vZ_R.\tag*{(distinct words are independent in the large-$N$ limit: \autoref{lem:freeness_random_P})}
\end{align}
The additional left and right symmetrizations help prove independence in the parallel query model while also being algorithmically affordable.\footnote{We expect the result to hold even without such symmetrization.} The (phased) permutation symmetry reduces the effective dimension so much that the input state can be restricted to a space of dimension depending only on $k$ and independent of $N$, which allows us to control the statistical distance (diamond norm) without paying an $N$-dependent conversion loss, which could have been problematic when taking the large-$N$ limit. 
\subsubsection{ Central limit theorem}
\label{sec:CLT_main}
The nice feature of independent sums is that we can invoke the central limit theorem 
\begin{align}
    \sum_{\vec{i}} w_{\vec{i}}(\theta)\vZ_{\vec{i}} \stackrel{m^{\ell}\gg 1}{\approx} \vG\quad \text{since}\quad \BE[ \vZ^{(\dagger)}_i \otimes \vZ^{(\dagger)}_i] = \BE[ \vG^{(\dagger)}_i \otimes \vG^{(\dagger)}_i] \tag{in the sense of $k$-fold CP channels}
\end{align}
where $\vG$ is the \textit{Ginibre ensemble}~\cite{Ginibre1965}, a random nonhermitian matrix with i.i.d. Gaussian entries. (See \autoref{sec:notations} for the definition.) Indeed, since the second (mixed) moments match, in the limit of many summands,\footnote{ and with some requirement that weights are all reasonably small.} the fine-grained details of the weights $w_{\vec{i}}$ are washed away, and the sum converges to a Gaussian random matrix; see~\autoref{lem:lindeberg_largeN} for nonasymptotic bounds demonstrated using a Lindeberg principle~\cite{chen2023sparse}. While the Ginibre ensemble is not unitary, its effect in a $k$-fold Completely positive channel is \textit{equal} to Haar in the large-$N$ limit (\autoref{lem:Ginibre_is_Haar})
\begin{align}
    \vG &\stackrel{N\rightarrow \infty}{=} \vU_{Haar}. \tag{in the sense of $k$-fold CP channels}
\end{align}
Intuitively, the Ginibre ensemble $\vG$ has the same symmetry as a Haar random unitary (left and right unitary invariance) and also behaves like a unitary for fixed input states, as seen by the expectation $\BE[\vG^{\dagger}\vG]=\vI$. 

\subsubsection{Verifying the requirements for interpolation}
\label{sec:requirements_main}
The last and most significant step in our argument is to control the finite-$N$ corrections through an interpolation argument from the large-$N$ limit.
In particular, we choose the function of interest $f$ to be the distinguishing probability between the unitaries $\vV=\vZ_L\L(\prod_{j=1}^{\ell} \e^{\ri \theta_m\vA^{(j)}_m} \R)\vZ_R$ and $\vU_{Haar}$
by defining the key quantity
\begin{align}
    f_{\vrho,\vO}\left(\frac{1}{N}\right):=\tr\L[\vO\CN_1\vrho\R] - \tr\L[\vO\CN_2\vrho\R]\quad \text{for each}\quad \vrho\quad \text{and}\quad \vO
\end{align}
where $\CN_1 = \CN_{2k,\vV}$ and $\CN_2 = \CN_{2k,\vU_{Haar}}$ are the corresponding channels.

We aim to control the finite-$N$ distinguishing probability $f_{\vrho,\vO}(\frac{1}{N})$ from the large-$N$ limit $f_{\vrho,\vO}(\frac{1}{\infty})$ \textit{for arbitrary fixed $\vrho$ and $\vO$} that the quantum attacks may use. We spell out the structural requirements for $f_{\vrho,\vO}(\frac{1}{N})$ for interpolation:

\textbf{Low-degree rational functions of $N$} (\autoref{lem:extensions}). The random permutations are very nice objects that are defined consistently across dimensions $N'\ne N$, inducing channels $\CN_1^{(N')}$ and $\CN_2^{(N')}$. In particular, the average over permutations can be calculated using a well-defined diagrammatic expansion (\autoref{sec:diagram_Partition_algebra}), with coefficients that are low-degree rational functions of $N$.\footnote{Qualitatively, this can be regarded as the analog of Weingarten expansion in the unitary case.} However, the problem is that the test quantum state $\vrho$ and the test operator $\vO$, which are selected adversarially, do not obviously apply in other dimensions. Finding the suitable extension requires a few arguments: first, the permutation symmetry substantially reduces the number of parameters. Specifically, the Schur-Weyl duality for the commutant of permutations restricts the effective input state $\vrho$ to a much smaller object, the \textit{partition algebra} $\vP_k(N)$ (\autoref{sec:diagram_Partition_algebra}), whose \textit{algebraic structure} only depends on the number of copies $k$, in particular being \textit{independent} of $N$ as long as the dimension is large $N \ge 2k$~\cite{halverson2005partition,halverson2020set}.\footnote{In fact, the partition algebras stabilize and are all \textit{isomorphic} for large enough dimension $N\ge 2k$ (\autoref{thm:stablize}).}  
    While the partition algebra has been studied in the algebraic combinatorics literature, our argument necessitates explicitly describing the embedding with respect to the computational basis in order to verify that the embedding ``does not change too quickly'' as $N$ increases. This requires a detailed foray into the structure of the partition algebra and writing down an explicit orthogonal basis (\autoref{sec:explicit_basis}) as a diagrammatic sum to establish that the basis coefficients are rational polynomials in $\frac{1}{N}$.

    The above understanding of the algebraic structure allows us to handcraft the suitable extension $f_{\vrho,\vO}$ for other dimensions $N' \ne N$
    \begin{align}
        f_{\vrho,\vO}\left(\frac{1}{N'}\right) &=  \tr\L[\vO^{(N')}\CN^{(N')}_1\vrho^{(N')}\R]-\tr\L[\vO^{(N')}\CN^{(N')}_2\vrho^{(N')}\R]
    \end{align}
    by defining a family of test operators
    \begin{align}
        &\L(\vO^{(N')}, \vrho^{(N')} \R) \quad \text{for each }\quad 2k \le  N' \le \infty\\ \text{such that} \quad &\L(\vO^{(N)}, \vrho^{(N)} \R)= \L(\vO, \vrho\R).
    \end{align}
    The operators $\vrho^{(N)},\vO^{(N)}$ for other dimensions $N'$ relate to the original $\vrho, \vO$ in dimension $N$ by substituting the basis we found for the partition algebra. In the end, the function $f_{\vrho,\vO}(\frac{1}{N})$ is (approximately\footnote{We need to truncate the rapidly converging Taylor expansion for the exponential function.}) a rational polynomial of $N$ with degree $\poly(k)$ and poles at small integers $1,\ldots, \tilde{\CO}(2k\ell)$. 
    
    \textbf{A priori bounds.} The expression is a difference between probabilities such that
\begin{align}
    \labs{f_{\vrho,\vO}(\frac{1}{N'})} \le 2 \quad \text{for each integer} \quad N'.
\end{align}

    \textbf{Large-$N$ limits} (\autoref{lem:freeness_random_P}). The large-$N$ limits of each channel coincide a pair of nicer ones: $\CN_1$ to the sum over independent permutations $\sum w_{\vec{i}}\vW_{\vec{i}} \rightarrow \sum w_{\vec{i}}\vZ_{\vec{i}}$, and $\CN_2$ to the ``Gaussian'' model: $\vU_{Haar}\rightarrow \vG_{Ginibre}$.
\begin{align}
    \tr\L[\vO_{N'}\CN_1\vrho_{N'}\R] &\stackrel{N'\rightarrow \infty}{=}\tr\L[\vO_{N'}\CN^{(free)}_1\vrho_{N'}\R]\notag\\
    \tr\L[\vO_{N'}\CN_2\vrho_{N'}\R] &\stackrel{N'\rightarrow \infty}{=}\tr\L[\vO_{N'}\CN^{(Ginibre)}_2\vrho_{N'}\R].
\end{align}
    The independent sums $\CN^{(free)}_1$ and Gaussian $\CN^{(Ginibre)}_1$ can be compared by the Lindeberg principle~\cite{chen2023sparse,chen2024efficient}, with an error suppressed by $\poly(1/m^{\ell})$ typical in central limit theorems (\autoref{lem:lindeberg_largeN}). While the most general Lindeberg argument works in any finite dimension $N$, the large-$N$ limits simplify the calculations.

\subsection{Road map}
The remaining sections begin with a summary of the notation and the various random matrix ensembles we will encounter (\autoref{sec:notations}), followed by preliminary background (\autoref{sec:prelim}): Markov's other inequality (\autoref{sec:Markovs}), basic notions in algebra (\autoref{sec:minimal_algebra}), the representation theory of the symmetric group (\autoref{sec:rep_symmetry_group}), and lastly a self-contained introduction to the partition algebra (\autoref{sec:Parition_algebra}).  

We then derive the main lemmas: how distinct words of permutations can be replaced by independent permutations in the large-$N$ limit (\autoref{sec:asym_freeness_perm}), deriving an explicit basis for the partition algebra that is low degree in $N$ (\autoref{sec:partitionAlgebra_largeN}), then extending the test state $\vrho$ and observable $\vO$ across dimensions (\autoref{sec:extending_test_ops}). The complete proof of the main result is assembled in \autoref{sec:proof_main}.

\subsection{Related work}
We recently became aware of the independent work of Metger, Poremba, Sinha, and Yuen \cite{metger2024pseudorandom}, which obtain similar results for parallel-secure PRUs and $k$-designs via a completely different construction and analysis. 
At early stages of this work, the conceptual connection between large-$N$ limit and Markov inequality was passed between the concurrent works~\cite{Jorge_work_in_progress,chen2024efficient} through the first author. Nevertheless, in each scenario, the technical arguments needed to realize this idea appear rather different, and we include self-contained expositions.

\subsection{Acknowledgments}
We thank Chris Bowman, Dmitry Grinko, Jeongwan Haah, Aram Harrow, Jonas Haferkamp, Tom Halverson, William He, Hsin-Yuan Huang, Martin Kassabov, Yunchao Liu, Fermi Ma, Tony Metger, Quynh Nguyen, Ryan O'Donnell, Alexander Poremba, Arun Ram,  Makrand Sinha, Norah Tan, Joel Tropp, Jorge Garza Vargas, and Henry Yuen for stimulating discussions. We thank the Simons Institute for the Theory of Computing. A.B. was supported in part by the DOE QuantISED grant DE-SC0020360, the AFOSR under grant FA9550-21-1-0392, and by the U.S. DOE Office of Science under Award Number DE-SC0020266. P.H. acknowledges support from AFOSR (award FA9550-19-1-0369), DOE (Q-NEXT), CIFAR and the Simons Foundation.

\section{Notation}\label{sec:notations}
This section summarizes our notation.
We write constants and integration elements in roman font ($\e, \pi,\ri$ and $\rd t$, $\rd x$ ), scalar variables in lowercase ($a,b$), vectors in bras and kets $\ket{\psi}$, matrices in bold uppercase ($\vG, \vU$), and vectorized matrices in curly bras and kets $|\vO)$.  We use curly font for several objects: super-operators ($\CN$), sets $(\CS)$, and algorithms $(\CA)$.

\begin{align*}
      n & & \text{number of qubits} \\
  N &(=2^n) & \text{local dimension}\\
  k & &  \text{number of copies}\\
  z\in\BC & & \text{complex numbers}\\
  \BE[\cdot] & & \text{Expectation}
\end{align*}
Linear algebra:
\begin{align*}	
	\norm{\ket{\psi}}&: \quad &\text{the Euclidean norm of a vector $\ket{\psi}$}\\
 |\vO)&: \quad &\text{vectorized operator $\vO$}\\
	\norm{\vA}&:= \sup_{\ket{\psi},\ket{\phi}} \frac{\bra{\phi} \vA \ket{\psi}}{\norm{\ket{\psi}}\cdot \norm{\ket{\phi}}} \quad &\text{the operator norm of a matrix $\vA$}\\
	\vA^*&: \quad & \text{the entry-wise complex conjugate of a matrix $\vA$}\\
 	\vA^\dagger&: \quad & \text{the Hermitian conjugate of a matrix $\vA$}\\
  \vA^T&: \quad & \text{the transpose of a matrix $\vA$}\\
  	\norm{\vA}_p&:= (\tr \labs{\vA}^p)^{1/p}\quad&\text{the Schatten p-norm of a matrix $\vA$}\\
  \norm{\CN}_{p-p} &:= \sup_{\vA} \frac{\normp{\CN[\vA]}{p}}{\normp{\vA}{p}}\quad&\text{the induced $p-p$ norm of a superoperator $\CN$}\\
  \normp{\CN}{\diamond}&:=\normp{\CN\otimes \CI}{1-1}&\text{the diamond distance to capture inputs entangled with ancillas}\\
  \CN_{2k,\vO}&:=\BE \vO^{\otimes k}[\cdot] \vO^{\dagger\otimes k} &\text{the CP map associated with $k$-fold tensor product of $\vO$}
\end{align*}
Representation theory:
\begin{align*}
  \vec{\lambda}&\vdash m & \text{integer partitions $(\lambda_1,\lambda_2,\cdots)$ of $m$ such that $\sum_i \lambda_i = m$, $\lambda_{1}\ge \lambda_2\ge\cdots$}\\
  \vec{\lambda}^*& = (\lambda_2, \cdots) & \text{integer partitions with the first entry removed}\\
  \Pi &\vdash [2k]  &\text{set-partitions for $2k$ elements}\\
  \vO_{\Pi} & & \text{diagram corresponding to set-partition $\Pi$ acting on $(\BC^{N})^{\otimes k}$}\\
  \vO'_{\Pi} & & \text{like $\vO_{\Pi}$, but distinct blocks have distinct indices}\\
  \vP_k(N)&&\text{the partition algebra for $(\BC^{N})^{\otimes k}$}\\
  \Lambda_{k,N}& & \text{set of integer partitions corresponding to irreps of $\vP_{k}(N)$}\\
  \vJ_m && \text{span of diagrams with at most $m$ propagating blocks}\\
  \vsigma\in \vS_m && \text{Symmetry group acting on $m$ elements}\\
  \vec{p}_{\mu} && \text{Young symmetrizer corresponding to irreps labeled by integer partition $\vec{\mu}$}\\
  t \in SYT(\vec{\mu}) && \text{the set of standard Young tableaux associated with shape $\vec{\mu} \vdash m$.}
\end{align*}

\subsection{Random matrix ensembles}
Our technical argument relies on specific properties of several random matrix ensembles. In dimension $N$ (which is often $2^n$ for $n$ qubits), we denote 
\begin{align}
    \vU&\stackrel{Haar}{\sim} \vU(N) \tag{Haar random unitaries}\\
    \vS&\stackrel{unif}{\sim} \vS(N) \tag{uniformly random permutations}\\
    \vD_{z}&: (\vD_{z})_{ij} = \delta_{ij} z_i\quad\text{where} \quad z_i\stackrel{Haar}{\sim}\vU(1)  \tag{random diagonal phases}\\
    \vG&: (\vG)_{ij} \sim \frac{g +\ri g'}{\sqrt{2N}}. \tag{Ginibre ensemble}
\end{align}
We will use different notations for the symmetric group to highlight the two distinct roles that the group plays in the paper. $\vS(N)$ acts on individual $N$-dimensional systems while $\vS_m$ acts by permuting copies of those systems for $m\le k$. In the above, all Gaussians are i.i.d. centered with unit variance $\BE[g]=0, \BE[g^2] =1$. The Ginibre ensemble consists of non-Hermitian random matrices, which we normalize such that 
\begin{align}
    \BE [\vG^{\dagger}\vG] = 1\cdot \vI,
\end{align}
making it behave like a unitary in some regards. The random diagonal matrices are not interesting on their own and are meant to multiply with the permutation by
\begin{align}
    \quad \vZ := \vD_z \vS \tag{random phased permutations}.
\end{align}
These \textit{1-sparse} matrices are basic building blocks enabled by standard pseudorandomness assumptions (Q-PRPs, \autoref{defn:PRP}; inverses allowed). The central object we consider is the exponential $\e^{\ri \theta\vA_m }$ of the sparse random matrices
\begin{align}
      \vA_m &:= \frac{1}{\sqrt{2m}}\sum_{a=1}^m (\vZ_a+\vZ^{\dagger}_a) \quad \text{where}\quad \vZ_a\stackrel{i.i.d.}{\sim} \vZ. \tag{adjacency matrix for phased random graph}
\end{align}
The normalization ensures that $\BE[\vA_m^{\dagger}\vA_m] = \vI.$ The exponentials of the above sparse random matrices can be efficiently simulated on quantum computers using standard algorithms.

\begin{prop}[Efficient implementation{~\cite{gilyen2019quantum}}]\label{prop:cost_exponential}
For each $m, \ell\in \mathbb{Z}^+$ and $\theta \in \mathbb{R}$, the unitary for $\e^{\ri \theta_m \vA_{m}}$ can be implemented at precision $\epsilon$ at cost
\begin{align}
\CO\bigg(\log(m)m\L(\theta\sqrt{m}+\log(1/\epsilon)\R)\bigg) \quad &\text{one and two-qubit gates},\\
\quad \CO(\log(m))\quad&\text{ancillas}, \\
\quad \text{and}\quad \CO\bigg(\theta\sqrt{m}+\log(1/\epsilon)\bigg)\quad&\text{controlled-queries}
\end{align}
 to the SELECT operator (or its adjoint)
 \begin{align}
        \sum_{a=1}^m \ket{a}\bra{a}\otimes \vZ_a.
    \end{align}
\end{prop}
\begin{proof}
Use LCU with uniform weights $(1/\sqrt{2m},\cdots,1/\sqrt{2m})$ to create a block-encoding for $\frac{1}{\sqrt{2m}}\vA_{m}$. Then, apply QSVT.
\end{proof}

\subsection{\texorpdfstring{Pseudorandomness and $k$-designs}{Pseudorandomness and k-designs}}
We briefly recall the notions of both computational pseudorandomness and $k$-designs.
\begin{defn}[Quantum-secure pseudorandom permutations \cite{zhandry2016note}]\label{defn:PRP}
We say a keyed family of permutations $\vS_{w} \in \vS(2^n)$ is quantum-secure pseudorandom permutation (PRP) if each $\vS_{w}$ can efficiently generated, and for any poly-time quantum algorithm $\CA$, 
\begin{align}
    \labs{\BE_{w}\CA(\vS_w,\vS^{-1}_w) -\BE_{\vS\in \vS(N)}\CA(\vS,\vS^{-1})} \le \frac{1}{\text{superpoly}(n)}.  
\end{align}
  That is, the difference in the algorithms' acceptance probability on a random element of the PRP ensemble and on a truly random permutation is super-polynomially small.
\end{defn}

On the other hand, a quantum unitary $k$-design (\autoref{defn:approx_kdesign}) is an information-theoretic notion. When the error is zero, this simply amounts to saying that the following $2k$-fold tensors are equal:
\begin{align}
    \BE_{\vV}[\vV^{\otimes k}\otimes \vV^{*\otimes k} ] = \BE_{\vU_{Haar}}[\vU_{Haar}^{\otimes k}\otimes \vU_{Haar}^{*\otimes k}].
\end{align}
In the presence of error, the diamond distance (\autoref{defn:approx_kdesign}) quantifies the best a quantum computationally unbounded distinguisher can do given parallel access to the unitary ensemble. Various other norms yield qualitatively different operational meanings in the presence of error, but the unconditional interconversion costs can easily scale with the dimension.

\section{Preliminaries}\label{sec:prelim}
In this section, we collect several technical facts that form the raw ingredients of the argument. In particular, we include a largely self-contained summary of the relevant facts from abstract algebra.

\subsection{Markov's ``other'' inequality for polynomials}\label{sec:Markovs}

\begin{figure}[t]
    \centering
    \includegraphics[width=0.7\textwidth]{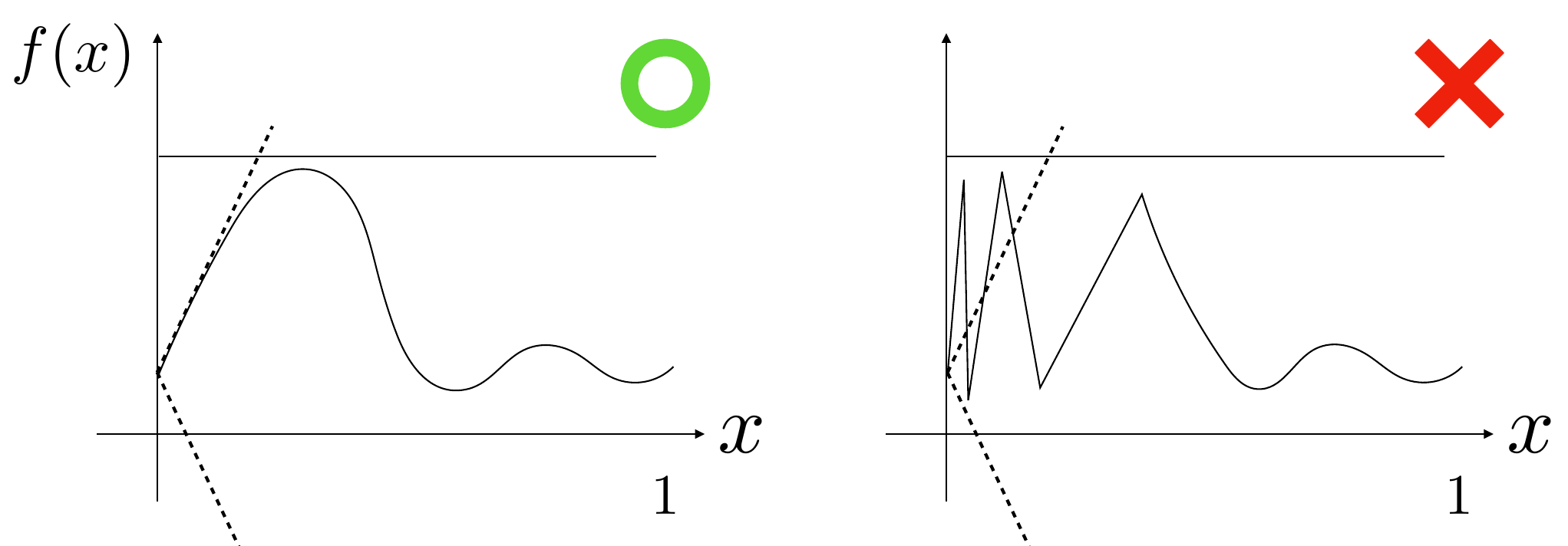}
    \caption{Markov's other inequality states that a bounded, low-degree polynomial cannot change too quickly. This simple fact can prove nontrivial query complexity lower bounds through careful choice of the interpolation parameter $x$ and the low-degree quantity $f(x)$.}
    \label{fig:polymethod}
\end{figure}

One of the very few techniques available for proving query complexity lower bounds is the \textit{polynomial method}. The starting point is a rather simple classical result in low-degree polynomials. 

\begin{lem}[Markov's other inequality~\cite{markov1889, markov1916, cheneyrivlin1966, ehlich1964}]
        Let $f(x):\mathbb{R}\rightarrow \mathbb{R}$ be a real polynomial of degree $d$.
    $$\sup_{x\in [0,1]}|f'(x)| \leq 2d^2 \max_{x\in [0,1]}\labs{f(x)}.$$
\end{lem}
The standard recipe for applying the above to lower-bounding the number of queries required for distinguishing tasks goes roughly as follows: cook up a quantity $f$ that depends on an interpolation parameter $x$ such that 
\begin{itemize}
    \item The values $f(x_1)$ and $f(x_2)$ correspond to the acceptance probabilities for the two cases to be distinguished. 
    \item The function $f(x)$ is a low-degree polynomial with a degree $d$ roughly the number of queries.
    \item The function $f(x)$ can be \textit{extended} to a larger interval $[x_1,x_3]$ with $x_3 - x_1 \gg x_2-x_1$ such that its value remains bounded.
\end{itemize}
Then, the distinguishing probability must not change too quickly between the two cases provided the degree $d$ is small: 
\begin{align}
    |f(x_1)-f(x_2)| \le 2d^2\frac{|x_1-x_2|}{\labs{x_1-x_3}} \max_{x\in [x_1,x_3]} \labs{f(x)}.
\end{align}

Applying the above requires a case-by-case adaptation, and creativity is often required to satisfy the combination of constraints. In our setting, the natural yet powerful approach is to select the interpolation parameter to be the inverse-dimension
\begin{align}
     x= \frac{1}{N}\quad \text{with interpolation range}\quad 0 \le x  \le 1,
\end{align}
which connects to the idea of the $\frac{1}{N}$ expansion in physics and in random matrix theory.

We will use the following specialized versions as a black-box.
\begin{lem}[Large-$N$ interpolation]\label{lem:large_N_interpolate}
    Let $f(x):\mathbb{R}\rightarrow \mathbb{R}$ be a real polynomial of degree $d$ and let $N, N_0$ be any nonnegative integers. Suppose that $d^2 \leq \frac{N_0+1}{4}$,
    \begin{align}
        \sup_{N \ge N_0} N \L|f\L(\frac{1}{N}\R)-f(0)\R| 
        \le \sup_{x\in [0, \frac{1}{N_0}]} |f'(x)| &\leq 4 d^2 \cdot N_0  \sup_{N \ge N_0} \L|f\L(\frac{1}{N}\R)\R|.
    \end{align}
\end{lem}
See, e.g.,~\cite[Lemma 4.1, Lemma 4.2]{chen2024new} for a standard argument. In our case, the interpolation range will need to effectively shrink
\begin{align}
0\le x\le 1 \rightarrow 0 \le x \le \frac{1}{N_0}    
\end{align}
to handle the presence of gaps in the integer reciprocals $\{\frac{1}{1}, \frac{1}{2},\cdots, \frac{1}{N},\cdots \}$; this will further worsen the dependence on the degree by roughly $\CO(d^2)\rightarrow \CO(d^4)$ (since $N_0 \ge 4d^2-1$). 
\begin{lem}[Markov inequality for rational polynomials with clustered poles]
Consider a real rational polynomial of integer $N$
\begin{align}
    f\L(\frac{1}{N}\R) = \frac{a(N)}{b(N)} \quad \text{where }\quad b(N) = \prod_i (N - b_i)^{m_i}. 
\end{align}
Suppose the poles $b_i$ are located within a disc of radius $\labs{b_i} \le B$, and the total degree is $d = \sum_i m_i$. Then, suppose $N_0\ge 8dB +d^2 -1$, and $f(\frac{1}{\infty}) < \infty$, we have that
    \begin{align}
        \sup_{N \ge N_0} N \L|f\L(\frac{1}{N}\R)-f\L(\frac{1}{\infty}\R)\R| 
        \le 4d^2(N_0+10dB) \sup_{N\ge N_0}\labs{f\L(\frac{1}{N}\R)}.
    \end{align}
\end{lem}

\begin{proof}
Since the limit is bounded $f(\frac{1}{\infty}) = \CO(N^0)$, the total number of poles also bounds the numerator degree by $d = \sum_i m_i$. It is instructive to rewrite in terms of $\frac{1}{N}$ by dividing with $N^{d}$
\begin{align}
    \frac{a(N)}{b(N)} = \frac{\tilde{a}(\frac{1}{N})}{\tilde{b}(\frac{1}{N})}.
\end{align}
Markov's other inequality (\autoref{lem:large_N_interpolate}) will be able to handle $\tilde{a}(\frac{1}{N})$. To handle the denominator, observe that
\begin{align}
    \labs{\frac{\tilde{b}(\frac{1}{\infty})}{\tilde{b}(\frac{1}{N})}-1} &= \labs{\prod_{i} \frac{1}{(1-\frac{b_i}{N})^{m_i}} -1}\\
    &\le \labs{\prod_{i} {\L(1+\frac{2B}{N}\R)^{m_i}} -1}\tag{assuming $2B\le N$ and using that $\labs{(1-z)^{-1} -1}\le 2\labs{z}$ for $\labs{x}\le \frac{1}{2}$}\\
    &\le \labs{\e^{2dB/N} -1}\tag{since $1+x \le \e^{x}$}\\
    &\le \frac{4dB}{N} \tag{if $ 2dB/N\le 1$}.
\end{align}
Therefore,
\begin{align}
    \labs{\frac{\tilde{a}(\frac{1}{N})}{\tilde{b}(\frac{1}{N})} - \frac{\tilde{a}(\frac{1}{\infty})}{\tilde{b}(\frac{1}{\infty})}} &\le \labs{\frac{\tilde{a}(\frac{1}{N})}{\tilde{b}(\frac{1}{N})} - \frac{\tilde{a}(\frac{1}{\infty})}{\tilde{b}(\frac{1}{N})}} + \labs{\frac{\tilde{a}(\frac{1}{\infty})}{\tilde{b}(\frac{1}{N})} - \frac{\tilde{a}(\frac{1}{\infty})}{\tilde{b}(\frac{1}{\infty})}}\\
    &\le \frac{4d^2N_0 \sup_{N\ge N_0}\labs{\tilde{a}(\frac{1}{N})}}{N\tilde{b}(\frac{1}{N})}+ \labs{\frac{\tilde{b}(\frac{1}{\infty})}{\tilde{b}(\frac{1}{N})}-1}\cdot\labs{\frac{\tilde{a}(\frac{1}{\infty})}{\tilde{b}(\frac{1}{\infty})}}\tag{by Markov's \autoref{lem:large_N_interpolate}}\\
    &\le \frac{4d^2N_0 \sup_{N\ge N_0}\labs{f(\frac{1}{N})}}{N}\cdot \frac{\sup_{N\ge N_0}\labs{\tilde{b}(\frac{1}{N})}}{\tilde{b}(\frac{1}{N})}+ \labs{\frac{\tilde{b}(\frac{1}{\infty})}{\tilde{b}(\frac{1}{N})}-1}\cdot\labs{f\L(\frac{1}{\infty}\R)}\tag{rewrite in terms of $f$}\\
    &\le  \frac{4d^2N_0 \sup_{N\ge N_0}\labs{f(\frac{1}{N})}}{N}\cdot \L(1+ 8\frac{dB}{N_0}\R) + \frac{4dB}{N}\cdot \labs{f\L(\frac{1}{\infty}\R)}\tag{if $8dB/N_0 \le 1$}\\
    &\le  \frac{4d^2(N_0+10dB)}{N}\sup_{N\ge N_0}\labs{f\L(\frac{1}{N}\R)} 
\end{align}
The assumptions we used were that $N_0 \ge d^2-1, N_0 \ge 8 dB$, and  $N \ge 2 dB$, so it suffices to assume $N_0 \ge 8dB + d^2 -1$ to complete the proof.
\end{proof}

The above allows us to only keep track of the locations and multiplicity of the pole without worrying about the polynomial degree on the numerator.

\subsection{Basic algebraic notions}\label{sec:minimal_algebra}

The partition algebra represented on the Hilbert space $\vB(\BC(\CH)^{\otimes k})$ is an \textit{algebra} (matrix addition and multiplication) with \textit{involution} (matrix conjugate transpose). This abstract algebraic structure already imposes structure that will be helpful. We introduce the minimal algebraic assumption as follows.  
\begin{defn}[Associative algebras over the complex numbers (abbreviated as algebras in this work)]
    A $\BC$-algebra $\vA$ is a set of elements such that
    \begin{itemize}
    \item (Closed under scalar multiplication) For any $z\in \BC$ and $\vA_1 \in \vA$,
        \begin{align}
            z\vA_1 \in \vA.
        \end{align}
        \item (Closed under addition and multiplication) For any $\vA_1 , \vA_2 \in \vA$,
        \begin{align}
            \vA_1+\vA_2 &\in \vA\\
            \vA_1 \cdot \vA_2 &\in \vA.
        \end{align}
        \item (Distributive and associative) For any $\vA_1, \vA_2, \vA_3 \in \vA$,
        \begin{align}
            (\vA_1 \cdot \vA_2 )\cdot \vA_3 &= \vA_1 \cdot (\vA_2 \cdot \vA_3)\\
            (\vA_1 + \vA_2 ) + \vA_3 &= \vA_1 + (\vA_2 + \vA_3)
        \end{align}
    \end{itemize}
\end{defn}

\begin{defn}[$\dagger$-algebra]
    A $\dagger$-algebra $\vA$ is an $\BC$-algebra with involution $(\cdot)^{\dagger}$ such that
    \begin{itemize}
        \item For any $\vA_1 \in \vA$, 
        \begin{align}
            (\vA_1^{\dagger})^{\dagger} = \vA_1.
        \end{align}
        \item For any $\vA_1, \vA_2 \in \vA$,
        \begin{align}
            (\vA_1 + \vA_2)^{\dagger} &= \vA_1^{\dagger} + \vA_2^{\dagger}\\
            (\vA_1\vA_2)^{\dagger} &= \vA_2^{\dagger}\vA_1^{\dagger}.
        \end{align}
        \item For any $z\in \BC$,
        \begin{align}
        (z \vA_1)^{\dagger} = z^* \vA_1^{\dagger}.    
        \end{align}
    \end{itemize}
\end{defn}
\begin{defn}[$\dagger$-homomorphism]\label{defn:dagger_homo}
    A $\dagger$-homomorphism is a map between $\dagger$-algebras $f:\vA \rightarrow \vA'$ such that 
    \begin{align}
        f(\vA_1+\vA_2) &= f(\vA_1)+f(\vA_2),\\
        f(\vA_1\cdot\vA_2) &= f(\vA_1)\cdot f(\vA_2),\\
        f(\vA_1^\dagger) &= f(\vA_1)^\dagger.
    \end{align}
\end{defn}
\begin{defn}[Ideals]\label{defn:ideals}
    A left-ideal $\vJ$ of a ring $\vR$ is a subring that is closed under multiplying the ring on the left
    \begin{align}
        \vR \vJ \subset \vJ.
    \end{align}
    Vice versa for the right ideals. A \textit{two-sided ideal} is both a left and right ideal.
\end{defn}

\begin{defn}[Factor]
    An algebra whose center consists only of multiples of the identity is called a factor.
\end{defn}

\begin{defn}[Matrix algebra]
    For any finite-dimensional complex Hilbert space $\CH$, the \textit{matrix algebra} $\vB(\CH)$ is the set of linear transformations $\CH \rightarrow \CH$. Moreover, the adjoint or, equivalently, the conjugate transpose w.r.t. any orthogonal basis defines a canonical \textit{involution} of the algebra.
\end{defn}

\begin{thm}
[$\dagger$-Subalgebras of matrix algebras]
\label{thm:structure_finite_dim}
    Any $\dagger$-subalgebra $\vA$ of a finite dimensional matrix algebra $\vB(\BC^{N})$ is unitarily conjugate to a direct sum of factors (labeled by $\lambda$)
    \begin{align}
    \vU \vB  \vU^{\dagger} = \bigoplus_{\lambda}\vM_{\lambda}\otimes \vI \quad \text{where}\quad \vM_{\lambda} = \vB(\BC^{d_{\lambda}}).
    \end{align}
    Further, any two-sided ideal $\vJ$ of $\vA$ must have the structure
    \begin{align}
    \vU \vJ  \vU^{\dagger} = \bigoplus_{\lambda}\vM'_{\lambda}\otimes \vI \quad \text{where}\quad \vM'_{\lambda} = 
    \vM_{\lambda} \quad \text{or}\quad \vM'_{\lambda} = 0,
    \end{align}
    and the quotient is isomorphic to another ideal 
   \begin{align}
       \frac{\vA}{\vJ} \simeq \vJ' \simeq \bigoplus_{\lambda}\vM''_{\lambda}\otimes \vI \quad \text{where}\quad \vM''_{\lambda} = 
    0 \quad \text{or}\quad \vM''_{\lambda} = \vM_{\lambda}.
   \end{align}
    The quotient can be naturally understood as a projector $\CQ_{\vJ}: \vA \rightarrow \vJ' \subset \vA $ by $\CQ_{\vJ}[\vM_{\lambda}\otimes \vI]= \vM_{\lambda''}\otimes \vI$, which is also a ${\dagger}$-homomorphism.
\end{thm}
\begin{proof}
This is a consequence of the Artin-Wedderburn theorem. See~\cite{halverson2005partition} for a self-contained exposition.
\end{proof}
\begin{defn}[Group algebra with complex coefficients]\label{defn:groupAlgebra}
    For a group $\vG$, the group algebra $\BC\vG$ consists of formal weighted linear combination of elements $\vG_i \in \vG$
    \begin{align}
        \sum_i z_i\vG_i \in \BC\vG, 
    \end{align}
    with multiplication given by the group multiplication $\vG_1\cdot \vG_2$.
\end{defn}

\subsection{Representations of the symmetric group}
\label{sec:rep_symmetry_group}
This section includes a minimal review of the representation theory of the symmetric group that will be helpful for working with the partition algebra; see~\cite{sagan2013symmetric} for a textbook introduction.
\subsubsection{Integer partitions}\label{sec:integerPartitions}

\begin{figure}[t]
    \centering
    \includegraphics[width=0.5\textwidth]{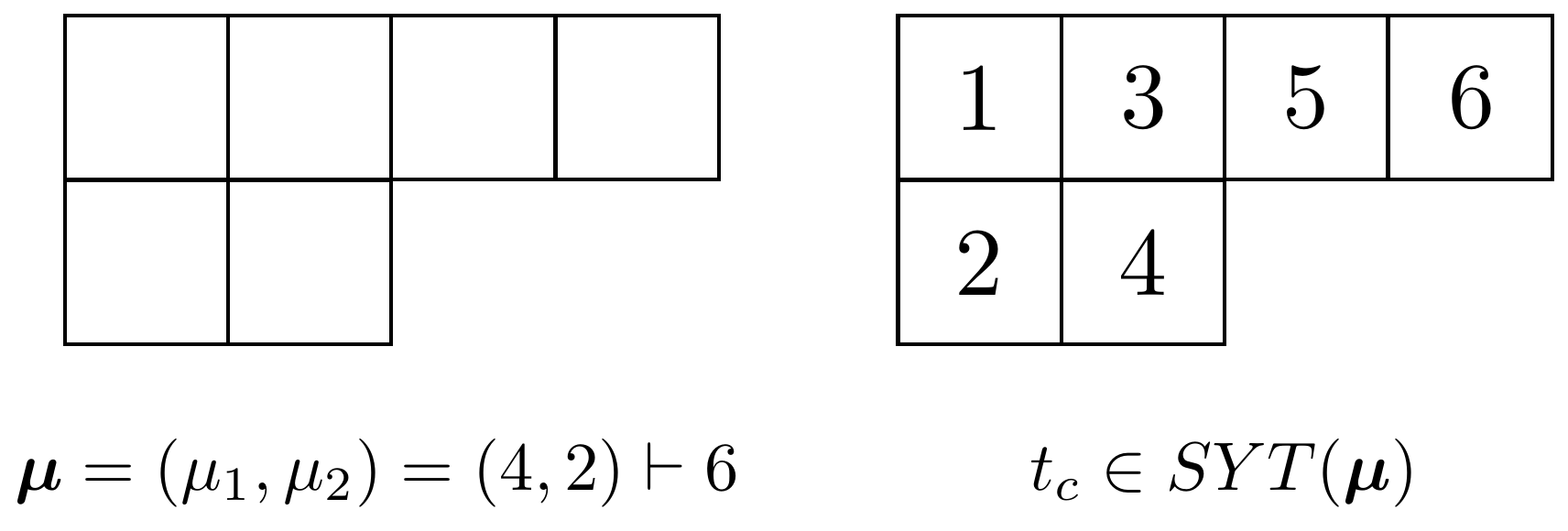}
    \caption{(Left) A Young diagram corresponding to the integer partition $\vec{\mu}\vdash 6$. (Right) The column-reading Young tableau $t_c \in SYT(\vec{\mu})$ as a special case of standard Young tableaux.}
    \label{fig:young}
\end{figure}
For any nonnegative integer $N\in \mathbb{Z}$ and vector of nonnegative integers $\vec{\lambda}$, we say that $\vec{\lambda}$ is an \textit{integer partition} of $N$, or 
\begin{align}
    \vec{\lambda} = (\lambda_1,\cdots,\lambda_{\ell}) \vdash N \quad \text{whenever} \quad \lambda_1+\cdots+\lambda_{\ell}=N \quad\text{and}\quad \lambda_1\ge \lambda_2\cdots \ge \lambda_{\ell}.
\end{align}
We will also denote the related notions by
\begin{align}
    \labs{\vec{\lambda}}:= \lambda_1+\cdots+\lambda_{\ell}\quad\text{and}\quad \vec{\lambda}^* := (\lambda_2,\cdots, \lambda_{\ell}).
\end{align}
In our usage for labeling the irreducible representations of the partition algebra $\vP_k(N)$ by $\vec{\lambda}\in \Lambda_{k,N}$, we will often consider a partition of the local dimension $N$ such that the first row is $ \labs{\vec{\lambda}}-\labs{\vec{\lambda}^*} \ge N-k$. Since the length of the first row is fixed by the remaining boxes $\lambda_1 = N- \labs{\vec{\lambda}^*}$, we will also only focus on the partition of the remaining boxes, often denoted by $\vec{\lambda^*}=\vec{\mu}\vdash m$.

\subsubsection{Young diagrams and Young tableaux}
Any integer partition $\vec{\mu}\vdash m$ corresponds to a \textit{Young diagram} (Fig.~\ref{fig:young}). In particular, one can fill in numbers $1,\ldots, m$ in arbitrary order, giving a \textit{Young tableau}. For a given shape $\vec{\mu}\vdash m$, a \textit{standard Young tableau} $t \in SYT(\vec{\mu})$ has its rows increase from left to right, and the columns increase from top to bottom. For each integer partition $\vec{\mu}$, the \textit{column-reading tableau} $t_c \in SYT(\vec{\mu})$ fills in the Young diagram by going down the columns from left to right from $1$ to $m$. 

\subsubsection{Irreducible modules}

This section describes the irreducible modules of the symmetric group algebra, and the notions presented here will prove useful when studying the partition algebra. For each Young tableau $t$ of shape $\vec{\mu}$, define the particular permutation
\begin{align}
    \vsigma_t\quad \text{such that}\quad \vsigma_t: t_c \rightarrow t \in SYT(\vec{\mu}).
\end{align}
That is, $\vsigma_t$ is the permutation that sends the column-reading Young tableau $t_c$ to the Young tableau $t$.

For each integer partition $\vec{\mu}\vdash m$, consider the \textit{Young symmetrizer}
\begin{align}
    \vec{p}_{\vec{\mu}} := \sum_{\vsigma\in C(t_c)} \sum_{\vsigma'\in R(t_c)} \text{sgn}(\vsigma)\vsigma\vsigma' \in \BC\vS_m
\end{align}
where $\BC\vS_m$ is the group algebra (\autoref{defn:groupAlgebra}) for the symmetric group $\vS_m$, $C(t_c)$ is the subgroup of permutations that preserves the rows of $t_c$, $R(t_c)$ the one for the columns, and sgn$(\vsigma)$ is the sign of the permutation $\vsigma$.

\begin{lem}[Irreducible module of {$\BC\vS_m$~\cite{sagan2013symmetric}}]\label{lem:irrep_Sm}
For each integer partition $\vec{\mu}\vdash m $, the space
\begin{align}
\BC\text{-span}\{\vsigma_t \vec{p}_{\vec{\mu}}\}_{t\in SYT(\vec{\mu})} \subset \BC\vS_{m}    
\end{align}
is a copy of the irreducible module $\vS^{\vec{\mu}}_m$ labeled by $\vec{\mu}$ in the left regular representation of $\BC\vS_m$.
\end{lem}

\begin{lem}[An orthonormal basis]\label{lem:ONbasis_Sm}
    There exist $\vec{u}_t\in \BC\text{-span}\{\vsigma_{t'}\}_{t'\in SYT(\vec{\mu})}$ such that 
\begin{align}
    \BC\text{-span}\{\vsigma_t \vec{p}_{\vec{\mu}}\}_{t\in SYT(\vec{\mu})} = \BC\text{-span}\{\vu_t \vec{p}_{\vec{\mu}}\}_{t\in SYT(\vec{\mu})}
\end{align}
and the map 
    \begin{align}
        \vec{u}_t \vec{p}_{\vec{\mu}}\vec{p}^{\dagger}_{\vec{\mu}}\vec{u}^{\dagger}_{t'} \rightarrow \ket{t}\bra{t'}\quad \text{is a $\dagger$-isomorphism}.
    \end{align}
 Therefore, the span
    \begin{align}
    \BC\text{-span}\{\vec{u}_t \vec{p}_{\vec{\mu}}\vec{p}^{\dagger}_{\vec{\mu}}\vec{u}^{\dagger}_{t'} \}_{t,t'\in SYT(\vec{\mu})} \subset \BC\vS_{m}
\end{align}
is a factor as a $\dagger$-subalgebra.
\end{lem}
\begin{proof}
    For each irrep labeled by $\vec{\mu}$, run any orthogonalization procedure and fix it. The point is that this solely depends on the irrep label $\mu$ for the symmetric group (and not on the dimension of the computation basis).
\end{proof}
Note that the above notions can be discussed abstractly without considering the computational basis (or we could represent the permutation by acting on $m$ many copies of $N$-dimensional Hilbert space $\vS_m$). 
In contrast, the multiplication rule for the partition algebra $\vP_k(N)$ depends on $N$, and we expect the basis to also explicitly depend on $N$. We will see that the inclusion $\BC[\vS_m] \otimes \vI_{k-m} \subset \vP_k(N)$ will be very helpful in analyzing the structure of the partition algebra, and the objects $\vec{u}_t \vec{p}_{\vec{\mu}} \in \vP_{k}(N)$ defined above will play a crucial role.

\subsubsection{The computational basis}
Consider $k$ copies of $N$-dimensional Hilbert spaces $(\mathbb{C}^{N})^{\otimes k}$ and its matrix algebra $\vB((\mathbb{C}^{N})^{\otimes k}).$ A natural choice of orthonormal basis is the computational (orthonormal) basis $\ket{m}$,
which naturally extends by the tensor product structure
\begin{align}
\CH =\mathbb{C}^{N} &= \text{Span}\{ \ket{m}\}_{m \in [N]}\\
\text{inducing}\quad     (\mathbb{C}^{N})^{\otimes k} &= \text{Span}\{ \ket{m_1}\otimes \cdots\otimes\ket{m_k}\}_{m_i \in [N]},\\
    \quad \text{inducing}\quad \vB((\mathbb{C}^{N})^{\otimes k}) & = \text{Span}\{ \underset{=:\vO_{\vec{m}}}{\underbrace{\ket{m_1}\otimes \cdots\otimes\ket{m_k}\cdot \bra{m_{k+1}}\otimes \cdots\otimes\bra{m_{2k}}}} \}_{m_i \in [N]},\label{eq:Om}\\
    &\stackrel{vec.}{\equiv} \text{Span}\{ \underset{=:|\vO_{\vec{m}})}{\underbrace{\ket{m_1}\otimes \cdots\otimes\ket{m_k}\cdot \ket{m_{k+1}}\otimes \cdots\otimes\ket{m_{2k}}}} \}_{m_i \in [N]}\label{eq:vecm}.
\end{align}

The algebra $\vB((\mathbb{C}^{N})^{\otimes k})$ as a $\BC$-vector space is naturally endowed with the Hilbert-Schmidt inner product $\tr[\vO^{\dagger}_1\vO_2]$. When the multiplication structure is unimportant, we may \textit{vectorize} (\textit{vec.} as above) the matrices into vectors of doubled dimension to focus on the Hilbert space structure. The inner produce remains the same by $(\vO_1|\vO_2):=\tr[\vO^{\dagger}_1\vO_2]$.

Formally, for a matrix $\vA$, let us denote its vectorization (or purification) by 
\begin{align}
|\vA):= (I \otimes T)\vA 
\end{align}
using the ``transpose'' map $T\bra{i}=\ket{i}$ to turn bras into kets. 
Consequently, the definition extends to the \emph{vectorization} of a superoperator by
\begin{align*}
\mathcal{N}[\cdot]=\sum_j  \alpha_j \vA_j[\cdot]\vB_j \rightarrow \vec{\mathcal{N}}=\sum_j \alpha_j \vA_j\otimes\vB^T_j
\end{align*}
where $\vB^T_j$ denotes the transpose of the matrix $\vB_j$ in the computational basis $\ket{i}$. We use curly fonts $\CN$ for superoperators and bold fonts $\vec{\CN}$ for the vectorized superoperators, which are themselves matrices. 

\subsection{The partition algebra}\label{sec:Parition_algebra}

In quantum information, we often consider the commutant of tensor copies of unitaries
\begin{align}
     \L(\{\vU^{\otimes k}\}_{\vU \in \vU(N)}\R)' \subset \vB((\mathbb{C}^N)^{\otimes k}).
\end{align}
This commutant algebra is completely characterized by Schur-Weyl duality and the representation theory of the symmetric group. In this article, we are interested in the commutant of tensored \textit{permutations} $\vS^{\otimes k}$, called the \textit{partition algebra} $\vP_{k}(N)$, which depends on the number of copies $k$ and the dimension $N$:
\begin{align}
     \vP_{k}(N) := (\{\vS^{\otimes k}\}_{\vS \in \vS(N)})' \subset \vB((\mathbb{C}^N)^{\otimes k})\quad \text{where}\quad \vS \ket{i} = \ket{S(i)} \quad \text{for each}\quad i = 1,\ldots, N.
\end{align}
In the case $N \ge 2k$, Schur-Weyl duality also holds as follows:
\begin{align}
    \vP_k(N) &\equiv \bigoplus_{\vec{\lambda} \in \Lambda_{k,N}} \vP_{\vec{\lambda}} \otimes \vI_{\vS_{\vec{\lambda}}}\quad \text{and}\quad \BC\text{-span}\{\vS^{\otimes k}\}_{\vS\in \vS(N)} = \bigoplus_{\vec{\lambda} \in \Lambda_{k,N}} \vI_{\vP_{\vec{\lambda}}}\otimes \vS_{\vec{\lambda}}\\
    \text{acting on}\quad (\BC^N )^{\otimes k} &\equiv \bigoplus_{\vec{\lambda} \in \Lambda_{k,N}} \CH_{\vP_{\vec{\lambda}}} \otimes \CH_{\vS_{\vec{\lambda}}}
\end{align}
where the irreps are labeled by integer partitions 
\begin{align}
    \vec{\lambda} \in \Lambda_{k,N} := \{ \vec{\lambda} \vdash N \mid 0 \le \labs{\vec{\lambda}^*}\le k\}.
\end{align}
Indeed, since the permutations acting naturally on the computational basis are a special subset of the unitaries, $\vS(N) \subset \vU(N)$, the commutant is larger $ (\{\vU^{\otimes k}\}_{\vU \in \vU(N)})' \subset \vP_{k}(N)$. Crucially for our later argument, the algebra ``stops changing'' after $N \ge 2k$\footnote{This is a different regime from some other applications where $N \ll 2k,$ e.g., \cite{fei2023efficient,nguyen2023mixed,grinko2023efficient,grinko2023gelfand}.}.
\begin{thm}
[{\cite{halverson2020set}}]\label{thm:stablize}
    If $N \ge 2k$, the dimension of the partition algebra (sum of dimensions of the irreps) is $dim(\vP_{k}(N)) = Bell_{2k}$, which is independent of $N$. In fact, for each fixed $k$, the partition algebra stabilizes in terms of $N$. That is, the irreducible algebra representations are isomorphic (as algebras)
    \begin{align}
        \vP_{\vec{\lambda}} \simeq \vP_{\vec{\lambda'}} \quad \text{whenever}\quad \vec{\lambda^*} = \vec{\lambda^{'*}}\quad \text{and}\quad \labs{\vec{\lambda}},\labs{\vec{\lambda'}} \ge 2k.
    \end{align}
\end{thm}

\begin{prop}[Bounds on the Bell number, e.g.,{~\cite{low2010pseudo}}] 
The Bell numbers are bounded by
    \begin{align}
        Bell_{n} \le n!\quad \text{for each}\quad n \in \mathbb{Z}.
    \end{align}
\end{prop}
\begin{lem}[Hook formula{~\cite[Theorem 3.10.2]{sagan2013symmetric}}] \label{lem:dimension_S_lambda}
    For each integer partition $\vec{\lambda}\in \Lambda_{k,N}$, the dimension of the corresponding irreducible representation of the symmetric group is a polynomial of $N$
    \begin{align} \label{eq:Sn_irrep_dims}
        \dim \vS_{\vec{\lambda}} = \frac{N!}{\prod_{(i,j)\in \vec{\lambda}}h(i,j)} = \frac{ N^{b_0(\lambda)}(N - 1)^{b_1(\lambda)}\cdot (N - 2k -1)^{b_{2k-1}(\lambda)} }{\prod_{(i,j)\in \vec{\lambda}^*} h(i,j)} = s_{\vec{\lambda}}(N) 
    \end{align}
    where $h(i,j)$ is the hook length of the Young diagram associated with shape $\vec{\lambda}$
    \begin{align}
    b_j  = 0, 1 \quad \text{and}\quad deg (s_{\vec{\lambda}}) = \sum b_j = \labs{\vec{\lambda}^*} \le k. 
    \end{align}
\end{lem}
From \autoref{eq:Sn_irrep_dims}, viewing the dimension as a function of $N$, we see that the zeros are located at integer values in a bounded range and have multiplicity at most one. Strictly speaking, we can restrict to a stronger symmetry as we actually use phased permutations (see, e.g., ~\autoref{lem:freeness_random_P}). Nevertheless, most of the discussions will consider the permutation $\vS(N)$ as it is more general and well-studied.
\subsubsection{Set-partitions}\label{sec:set_partitions}
Following the notation of~\cite{low2010pseudo}, consider partitions $\Pi$ of a set with $\ell$ elements. (It will be convenient to distinguish notationally from the integer partitions discussed in~\autoref{sec:integerPartitions}.)
We write
\begin{itemize}
    \item $
    \Pi \vdash \{1,\cdots, \ell\}$ \, iff $\Pi$ is a set-partition of $\{1,\cdots,\ell\}$.
\item $(i,j) \in \Pi$ \, iff $(i, j)$ are in the same \textit{block} (also known as a \textit{cycle}) of $\Pi$.
\item $\labs{\Pi}:= (\text{number of blocks of }\Pi).$
\item $\Pi_1 \le \Pi_2$ iff $ (i,j) \in \Pi_1 \implies (i,j) \in \Pi_2$, which reads ``$\Pi_1$ is a refinement of $\Pi_2$.''
\end{itemize}
For example, let $\Pi = \{\{1,3\},\{2\}\}$, then
\begin{align}
    \Pi \vdash \{1,2,3\}, \quad (1,3) \in \Pi, \quad \text{and}\quad \labs{\Pi} =2.
\end{align}
The partial order structure between integer partitions will become important, and we will make use of generalized inclusion-exclusion principles. This can be understood in linear algebraic terms.

\begin{lem}[Linear extension of partial order~\cite{davey2002introduction}]\label{lem:extensionPartialOrder}
Any finite set with partial order can be extended as a subset of a totally ordered set. Therefore, the partial order can be (nonuniquely) represented as an upper triangular $0/1$-matrix $\vM$ satisfying $M_{ij} = 1$ if and only if $i\geq j$. 
\end{lem}
\begin{lem}[Restricting partial order]\label{lem:restricting_partial_order}
    For any finite set with partial order induced a partial order for any subset. Restriction of an upper triangular matrix on a subset of indices remains upper triangular.
\end{lem}
\begin{lem}[Inverting triangular matrices]\label{lem:inv_upper_triangular}
    Any upper triangular matrix with nonzero diagonals is invertible by another upper triangular matrix.
\end{lem}
\begin{proof}
This follows immediately from the expression for the inverse of a matrix in terms of its adjugate~\cite{horn2012matrix}.
\end{proof}

\subsubsection{Explicit presentation}
\begin{figure}[t]
    \centering
    \includegraphics[width=0.8\textwidth]{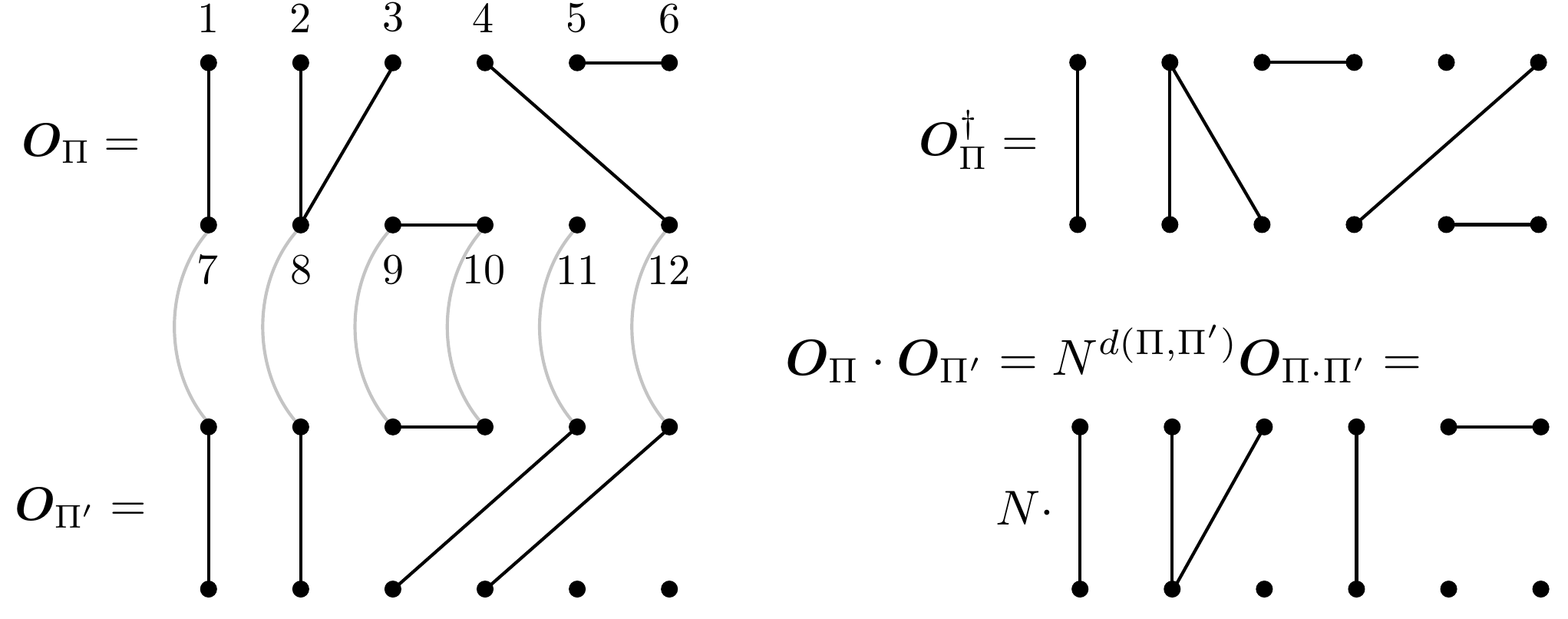}
    \caption{(Left) The diagrams $\vO_{\Pi}$ and $\vO_{\Pi'}$ corresponding to set-partitions $\Pi= \L((1,7), (2,3,8), (9,10),(4,12),(11),(5,6)\R)$ and $\Pi' = \L((1,7),(2,8),(3,4),(5,9), (10,6),(11),(12)\R)$ for $k=6$. (Top right) The conjugation $\vO_\Pi^{\dagger}$ is simply flipping the diagram upside down. (Bottom right) Multiplying two diagrams amounts to identifying the adjacent rows (indicated by gray lines) and simplifying the diagram to remove any connected components lying completely within the identified rows. A factor of $N$ appears for each component removed during diagram simplification, accounted for by $d(\Pi,\Pi')$. Thus, although the diagrams and set-partitions do not explicitly refer to the local dimension $N$, the partition algebra is a priori dependent on $N$.}
    \label{fig:O_dagger_multiply}
\end{figure}

\label{sec:diagram_Partition_algebra}
The partition algebra $\vP_{k}(N)$ can be defined abstractly but, for concreteness, we will explicitly present it in the computational basis throughout this work. We need to introduce new notation to capture this high-dimensional object. In particular, for our purposes, the notation in this section assumes we are in the stable range where $N \ge 2k >0$. Let
\begin{align}
    \vec{m} &= (m_1,\cdots, m_{2k}) \in \{1,\cdots, N\}^{2k}\tag{tuple of indices}\\
    M'_{\Pi} &:= \{ \vec{m}: m_i = m_j \quad\text{iff}\quad (i,j)\in \Pi\}. \tag{set of $\vec{m}$ whose partition is exactly $\Pi$}\\
    M_{\Pi} &:= \{ \vec{m}: m_i = m_j \quad\text{if}\quad (i,j)\in \Pi\}. \tag{set of $\vec{m}$ whose partition is refined by $\Pi$}
\end{align}
The two types of sets $M_{\Pi}$, $M'_{\Pi}$ are related by the union
\begin{align}
    M_{\Pi} = \bigcup_{\Pi' \ge \Pi} M'_{\Pi'},  
\end{align}
where each $M'_{\Pi'}$ are disjoint from each other.

The cardinality of sets $M_{\Pi}, M'_{\Pi}$ is given by an elementary calculation:
\begin{align} \label{eq:size_M_sets}
\labs{M'_{\Pi}} = N (N-1)\cdots (N- \labs{\Pi}+1)\quad \text{and}\quad \labs{M_{\Pi}} = N^{\labs{\Pi}}.
\end{align}
In particular, observe that both quantities are polynomials in $N$, and the zeros are located at integer values -- this seemingly innocent structure will play a crucial role in our interpolation argument. Using the above notation, the partition algebra $\vP_{k}(N)\subset \vB((\mathbb{C}^N)^{\otimes k})$ is a $\dagger$-algebra presented by a $\mathbb{C}$-linear combination of the following operators (using the notation of \eqref{eq:Om}): 
\begin{align}
    \vP_{k}(N) = \text{Span}_{\mathbb{C}}\{\vO_{\Pi}\}_{\Pi}\quad \text{where}\quad
    \vO_{\Pi} &:=  \sum_{\vec{m} \in M_{\Pi}} \vO_{\vec{m}}\quad 
    \text{for each} \quad \Pi \vdash \{1, \cdots, 2k\}.
\end{align}
These can be described diagrammatically, as illustrated in~\autoref{fig:O_dagger_multiply}.
The \textit{involution} (i.e., adjoint) acts as 
\begin{align}
    ( \vO_{\Pi} )^{\dagger} = \vO_{\Pi^{\dagger}} \quad \text{where}\quad \Pi^{\dagger} \vdash \{1, \cdots, 2k\}.
\end{align}
Graphically, the involuted partition $ \Pi^{\dagger}$ amounts to flipping it ``upside down,'' which is consistent with the matrix adjoint. The multiplication operation (realized as matrix multiplication) acts nicely on these operators by 
\begin{align}
    \vO_{\Pi_2}\vO_{\Pi_1} = \vO_{\Pi_2\cdot \Pi_1} \cdot N^{d(\Pi_2,\Pi_1)}\label{eq:O2O1=O21},
\end{align}
where the multiplication rule for set-partitions $\Pi_2\cdot \Pi_1$ is summarized in~\autoref{fig:O_dagger_multiply} and given in more detail in~\cite{halverson2005partition,halverson2020set}. The integer $d$ is a function of both $\Pi_1$ and $\Pi_2$ but independent of $N$. Graphically, 
\begin{align}
    d = \# ( \text{components without free legs after multiplying $\Pi_1$ and $\Pi_2$} ).
\end{align}
While $\vP_k(N)$ depends on both $k$ and $N$, the algebraic structure is mainly determined by $k$ since $\Pi \vdash \{1, \cdots, 2k\}$; the dependence on $N$ only comes in as a multiplicative weight $N^{d(\Pi_2,\Pi_1)}$. Alternatively, we may describe the algebra in the orthogonal basis
\begin{align}
    \vO'_\Pi &:=\sum_{\vec{m} \in M'_{\Pi}} \vO_{m} \quad \text{such that}\quad \tr[ \vO'^{\dagger}_{\Pi'}\vO'_\Pi] = \delta_{\Pi,\Pi'}\cdot \norm{\vO'_\Pi}^2_2 \quad\text{for each}\quad \Pi,\Pi' \vdash \{1, \cdots, 2k\}.\label{eq:O'pi},
\end{align}
as drawn in \autoref{fig:Oprime_Ovectorize}.
This new set of operators is related to the original ones by a \textit{M\"obius} inversion (a generalized inclusion-exclusion principle for partial orders,~\autoref{lem:extensionPartialOrder}, \autoref{lem:inv_upper_triangular}) 
\begin{align}
    \vO_{\Pi} &=  \sum_{\Pi' \ge \Pi} \vO'_{\Pi'} \\
    &= \sum_{\Pi' \vdash \{1, \cdots, 2k\} } K_{\Pi\Pi'} \vO'_{\Pi'}\quad \text{where}\quad K_{\Pi\Pi'} := \indicator(\Pi' \ge \Pi)\label{eq:K}\\
    \vO'_{\Pi} &=  \sum_{\Pi' \vdash \{1, \cdots, 2k\} } (K^{-1})_{\Pi\Pi'} \vO_{\Pi'}\label{eq:Kinv} 
\end{align}
where $\vK^{-1}$ is the M\"obius inversion of $\vK$; indeed, the matrix $\vK$ is upper triangular with ones on the diagonals and thus invertible (\autoref{lem:inv_upper_triangular}). Importantly, the entries of matrix $\vK$ only depend on the combinatorial structure of $\Pi\vdash 2k$ and are independent of $N$.

For simplicity, we can ignore the multiplicative structure and vectorize the operators (by turning bras to kets~\eqref{eq:vecm}). In this picture, it is convenient to think about the orthogonal basis $|\vO'_\Pi)$ which, as we've seen, can be expressed by a linear combination of $|\vO_{\Pi})$.
\begin{align}
    |\vO'_\Pi) &:=\sum_{\vec{m} \in M'_{\Pi}} |\vO_{\vec{m}}) \quad \text{where}\quad (\vO'_{\Pi}|\vO'_\Pi) = \norm{\vO'_\Pi}^2_2=\labs{M'_{\Pi}}  \quad\text{for each}\quad \Pi\vdash \{1, \cdots, 2k\}\label{eq:norm_O'pi}.
\end{align}

Then, the action of tensored permutations can be written as a superoperator projecting $\vB((\mathbb{C}^N)^{\otimes k}) \rightarrow \vP_{k}(N) \subset \vB((\mathbb{C}^N)^{\otimes k})$
\begin{align}
    \CN_{2k,\vS}:=\BE [\vS^{\otimes k}( \cdot )\vS^{\dagger \otimes k} ]&= \sum_{\Pi \vdash 2k} \frac{\vO'_{\Pi} \tr[\vO'_{\Pi} (\cdot) ]}{\norm{\vO'_{\Pi}}^2_2},\label{eq:average_projector}
\end{align}
or as an operator after vectorization
\begin{align}
     \vec{\CN}_{2k,\vS} := \BE [\vS^{\otimes k}\otimes \vS^{*\otimes k} ]&= \sum_{\Pi \vdash 2k} \frac{|\vO'_\Pi)(\vO'_\Pi|}{\norm{\vO'_{\Pi}}^2_2}, 
\end{align}
where the Hilbert-Schmidt norm reads $\norm{\vO}_2 = \sqrt{\tr[\vO^{\dagger}\vO]} = \sqrt{(\vO|\vO)}$. Importantly, for each $\Pi$, the zeros of the normalization factors $\labs{M'_{\Pi}}$ are contained in the integer set $\{0,1,\cdots, 2k\}$ by \autoref{eq:size_M_sets}. The operators $\vO_{\Pi}$ are also nice in the sense that
\begin{align}
     \tr[\vO^{\dagger}_{\Pi_1}\vO_{\Pi_2}] = N^{c(\Pi_1,\Pi_2)} \quad \text{for each partition}\quad \Pi \vdash \{1, \cdots, 2k\}\label{eq:trO1O2}
\end{align}
for an integer $c(\Pi_1,\Pi_2)$ depending only on $\Pi_1,\Pi_2$ (and independent of $N$); graphically, 
\begin{align}
    c(\Pi_1,\Pi_2) = \# ( \text{connected components after fully contracting $\Pi_1^\dagger$ and $\Pi_2$})\label{eq:cPi1Pi2}.
\end{align}
This also gives the inner product 
\begin{align}
    \frac{\tr[\vO^{\dagger}_{\Pi_1}\vO_{\Pi_2}]}{\norm{\vO_{\Pi_1}}_2\norm{\vO_{\Pi_2}}_2} = \frac{1}{\sqrt{N}^{\delta c(\Pi_1,\Pi_2)}} \quad \text{where}\quad \delta c(\Pi_1,\Pi_2) = c(\Pi_1,\Pi_1) + c(\Pi_2,\Pi_2) - 2c(\Pi_1,\Pi_2).
\end{align}
\begin{figure}[t]
    \centering
    \includegraphics[width=0.7\textwidth]{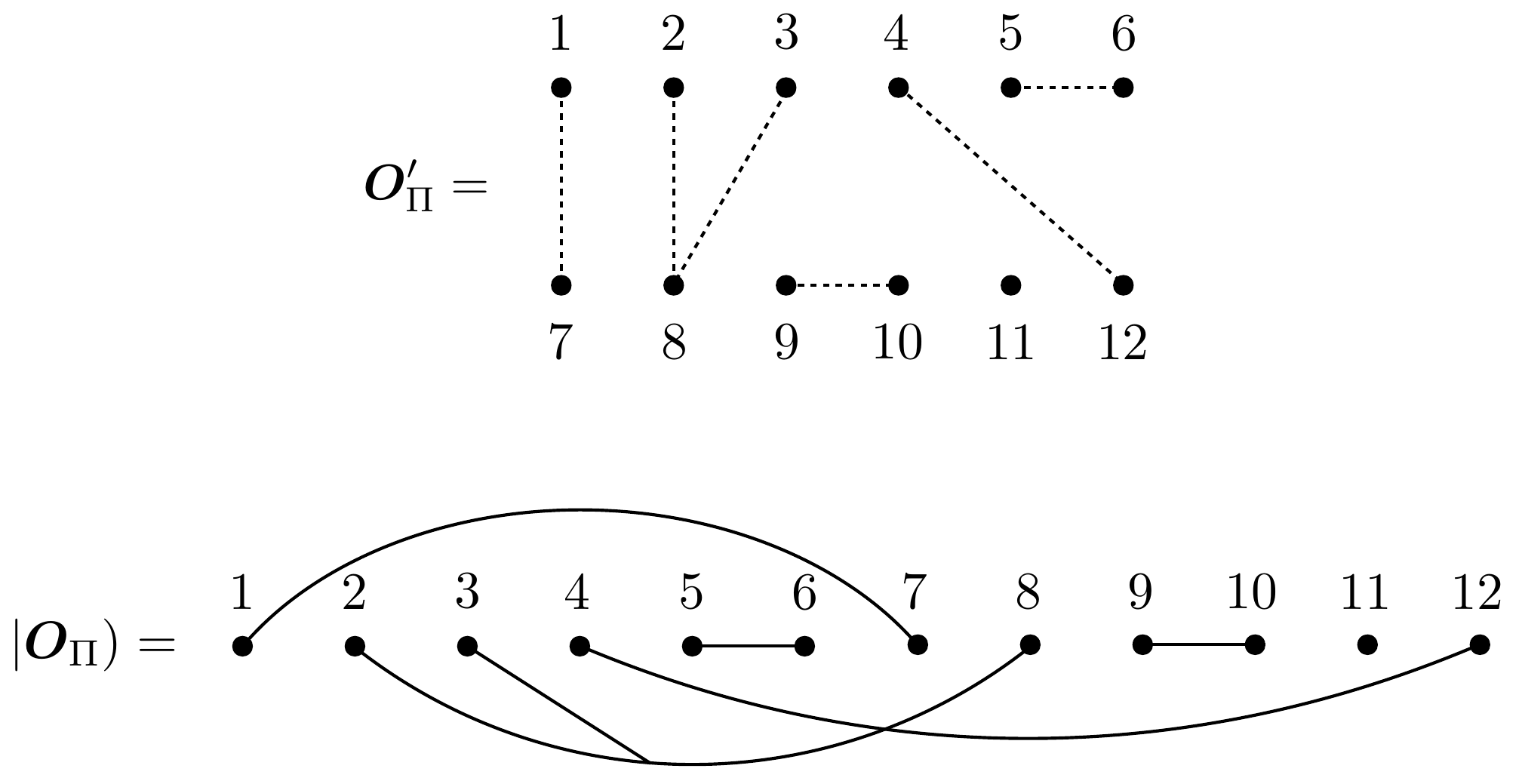}
    \caption{Related notions. The operator $\vO'_{\Pi}$ is modified from $\vO_{\Pi}$ such that distinct blocks have distinct labels (denoted by dotted lines). The vectorization $|\vO_{\Pi})$ treats the bras and kets equally.}
    \label{fig:Oprime_Ovectorize}
\end{figure}

\begin{figure}[t]
    \centering
    \includegraphics[width=0.3\textwidth]{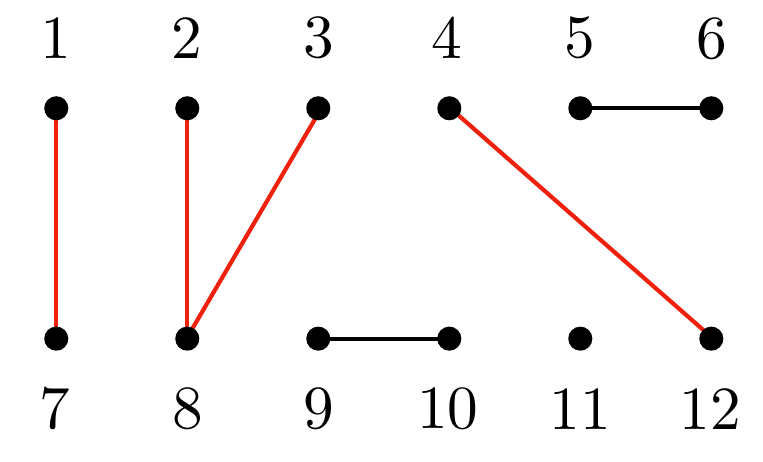}
    \caption{An illustration for propagating blocks, highlighted in red. The lines between $(5,6)$ and $(9,10)$ do not propagate since they are completely within the top and the bottom. The number of propagating blocks is $3$, i.e., the above diagram is in $\vJ_3(N)$ but not in $\vJ_2(N).$}
    \label{fig:propagating}
\end{figure}

\subsubsection{Algebraic structure}
The two-sided ideals $\vJ_m(N)$ for $0 \leq m \leq k$ defined by
\begin{align}
    \vJ_{m} (N):= \BC\text{-span}\{\vO_{\Pi} \,|\, \Pi\vdash 2k \, \text{has at most $m$ propagating blocks}\}
\end{align}
will play a central role in understanding the algebraic structure.
A block in $\Pi$ is propagating if it contains both nodes at the bottom and the top (\autoref{fig:propagating}).
Indeed, since diagram multiplication, from either the left or the right,  cannot increase the number of propagating blocks, we have that
\begin{align}
    \vP_k(N)\vJ_m(N) = \vJ_m(N) \vP_k(N) \subset \vJ_m(N).
\end{align}
These ideals are naturally related by
\begin{align}
    \vJ_0\subset \vJ_1 \subset \cdots \vJ_k = \vP_{k}(N). 
\end{align}
\begin{figure}[t]
    \centering
    \includegraphics[width=0.3\textwidth]{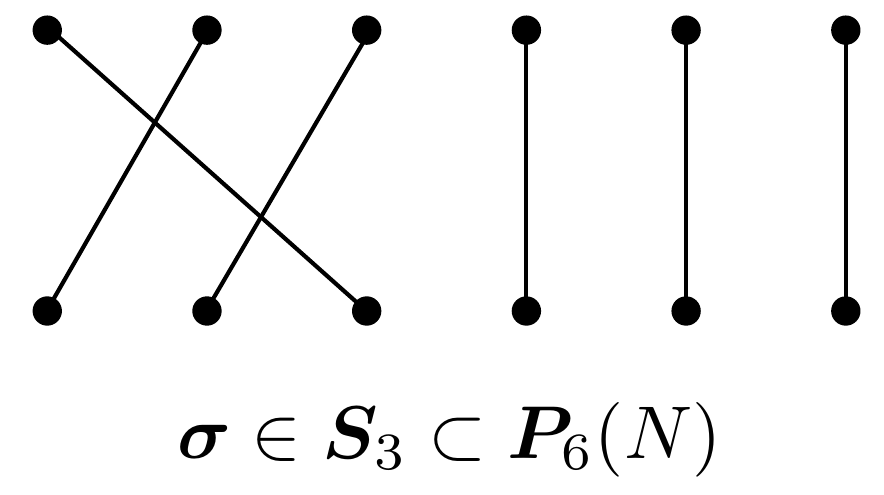}
    \caption{The symmetric group $\vS_m$ embedded into the partition algebra $\vP_k(N)$ by acting on the leftmost nodes.}
    \label{fig:SinP}
\end{figure}
\begin{lem}[Projection onto subalgebra]
The quotient projector acting on the partition algebra $\vP_k(N)$ can be written as
\begin{align}
    \CQ_{\vJ_{m}}[\cdot] = \sum_{\labs{\vec{\lambda^*}}> m } \vec{\Pi}_{\vec{\lambda}}\cdot\vec{\Pi}_{\vec{\lambda}},
\end{align}  
where the projectors $\vec{\Pi}_{\vec{\lambda}} \in \vP_{k}(N)$ correspond to the irreps of the partition algebra $\vP_{k}(N)$ labeled by $\vec{\lambda}\in \Lambda_{k,N}$.
\end{lem}
\begin{proof}
    For any $\vec{\lambda}\in \Lambda_{k,N}$, consider the orthogonal basis representatives we construct later in~\autoref{lem:orthonormal}, and note that the matrix element is written as some element in $\vJ_{\labs{\vec{\lambda}^*}}$ modulo $\vJ_{\labs{\vec{\lambda}^*}-1}.$ Therefore, if $ \labs{\vec{\lambda^*}}>m$, $\CQ_{\vJ_{m}}$ acts as the identity, and otherwise, it acts as zero, which coincides with the advertised expression.
\end{proof}
It will also be crucial to consider the symmetric group algebra as a subalgebra
\begin{align}
    \BC\vS_m \subset \vP_k(N) \quad \text{for each}\quad 1 \le m \le k
\end{align}
where $\vS_m$ acts on the leftmost $m$ nodes (\autoref{fig:SinP}).

\section{New lemmas}
In this section, we prove the key new lemmas that are foundational for our main results. 

\subsection{Asymptotic freeness of random permutations}
\label{sec:asym_freeness_perm}
The power of large-$N$ limits is that a version of ``freeness'' kicks in for random permutations: different \textit{words} are effectively independent from each other. The concepts of free groups and free words will be particularly helpful.
\begin{defn}[Free words]
    For a countable set of symbols $\{s_i\}$, we denote the set of free words by 
    \begin{align}
        \CW(\{s_i\}_i) = \{e, s_1, s_2, s_1s_2, s_2s_1,s_1^{-2}, \cdots \},
    \end{align}
    which are also the elements of the free group generated by $\{s_i\}_i$.
\end{defn}
\begin{lem}[Effective asymptotic freeness of phased random permutations]\label{lem:freeness_random_P} 
    Consider the free words generated by a set of $m$ independent uniformly random phased permutations
    \begin{align}
        \CW(\{\vZ_a\}_{a=1}^m) \quad \text{where}\quad \vZ_a \stackrel{i.i.d.}{\sim} \vZ.
    \end{align}
    Fix integers $\ell,k$ and fix a set of $\ell$ distinct words $\vW_1,\ldots, \vW_{\ell} \in \CW$. Consider a corresponding set of $\ell$ independent permutations
    \begin{align}
    \vZ_{i} := \begin{cases}
        \stackrel{i.i.d.}{\sim} \vZ\quad &\text{if}\quad \vW_i \ne \vI\\
        \vI \quad &\text{if}\quad \vW_i= \vI.
    \end{cases}
    \end{align}
    Then, for each $\vec{j} \in [\ell]^{2k}$, in the large-$N$ limit,
    \begin{align}
        &\lim_{N\rightarrow \infty} \lnormp{\CN_{2k,\vZ}\circ\bigg( \CN_{ \vec{j}, \{\vW_i\}} -\CN_{\vec{j}, \{\vZ_i\}} \bigg)\circ \CN_{2k,\vZ}}{\diamond} = 0,
    \end{align}
    where
    \begin{align}
    \CN_{\vec{j}, \{\vW_i\}} &:=\BE \vW_{j_1}\otimes \cdots \vW_{j_k} [\cdot] \vW^{\dagger}_{j_{k+1}}\otimes \vW^{\dagger}_{j_{2k}},  \quad \quad \vec{j} \in [\ell]^{2k}, \label{eq:channel_from_perm_set} \\
    \CN_{2k,\vZ} &:= \BE \vZ^{\otimes k} [\cdot] \vZ^{\dagger\otimes k}
\end{align}
\end{lem}

Roughly speaking, we want to show that in the large-$N$ limit, distinct nontrivial words can be replaced by independent random permutations $\{ \vW_{i} \} \rightarrow \{\vZ_i\}$ (identity remains the identity) in the sense of tensored mixed moments (which may not be completely-positive). Here, we are showing this channel equivalence holds, when the overall channel is first symmetrized by a random phased permutation before and after our words by the channel $\CN_{2k,\vZ}$. This symmetrization allows us to restrict the test operators to a \textit{proper subalgebra} of the partition algebra with nicer properties. Let
\begin{align}
    B_{2k}:=\left\{ \Pi\vdash\{1,\cdots,2k\}| \text{each block has an equal number of sites on top and bottom.}\right\}
\end{align}
then, the commutant of phased permutations are exactly the associated diagrams. In particular, every block is propagating.
\begin{lem}[Commutant of $\vZ^{\otimes k}$]\label{lem:commZ}
\begin{align}
    (\vZ^{\otimes k})^{'} = \BC\text{-Span}\{ \vO'_{\Pi}|\Pi\in B_{2k}\} = \BC\text{-Span}\{ \vO_{\Pi}|\Pi\in B_{2k}\}.
\end{align}    
\end{lem}
\begin{proof}[Proof of~\autoref{lem:commZ}]
    Recall that the random phased permutation can be written as $\vZ := \vD_z \vS$. We already have that the commutant of $(\vS^{\otimes k})^{'}$ is the partition algebra $\vP_{k}(N)$ with orthogonal basis $\{\vO'_{\Pi}| \Pi \vdash 2k \}$. It follows that $(\vZ^{\otimes k})^{'} \subset (\vS^{\otimes k})^{'} = \vP_{k}(N)$. Any $\vO'_{\Pi}$ for $\Pi\in B_{2k}$ is contained in the commutant $(\vZ^{\otimes k})^{'}$ as the balanced signs from $\vD_z$ cancel each other out, and the unsigned permutations $\vS$ fixes the diagram
    \begin{align}
        \CN_{2k,\vZ}[\vO'_{\Pi}]&=\CN_{2k,\vS}\circ\CN_{2k,\vD_{z}}[\vO'_{\Pi}]\\
        &=\CN_{2k,\vS}[\vO'_{\Pi}]\tag{Since $\Pi\in B_{2k}$}\\
        &= \vO'_{\Pi}.\tag{Since $\vO'_{\Pi}\in \vP_{k}(N)$}
    \end{align}
    Conversely, for any other diagram $\Pi\notin B_{2k}$, there exists a non-balanced block, and thus the operator $\vO'_{\Pi}$ would vanish under the symmetrization $\CN_{2k,\vZ}$ due to the unbalanced signs on that block. Next, we show that $\vO_{\Pi}$ for $\Pi\in B_{2k}$ also form a basis for the commutant $(\vZ^{\otimes k})^{'}$. This is because each diagram $\vO_{\Pi}$ for $\Pi\in B_{2k}$ can written as $\vO'_{\Pi}$ for $\Pi\in B_{2k}$
    \begin{align}
    \vO_{\Pi} = \sum_{\Pi'\ge \Pi} \vO'_{\Pi'}
\end{align}
and that each $\Pi\in B_{2k}$ can only be refinements of $\Pi'\in B_{2k}$
\begin{align}
     \Pi'\ge \Pi, \Pi \in B_{2k} \implies \Pi' \in B_{2k}.
\end{align}
Therefore,
\begin{align}
    \BC\text{-Span}\{ \vO_{\Pi}|\Pi\in B_{2k}\} \subset \BC\text{-Span}\{ \vO'_{\Pi}|\Pi\in B_{2k}\}.
\end{align}
Since the two vector spaces have the same dimensions, they must be equal.
\end{proof}

We can now prove the asymptotic freeness of phased random permutations.
\begin{proof}[Proof of~\autoref{lem:freeness_random_P}]
We begin with controlling the 1-norm for particular ``factorized'' inputs (see~\eqref{eq:Om})
\begin{align}
    \lnormp{\bigg( \CN_{ \vec{j}, \{\vW_i\}} -\CN_{\vec{j}, \{\vZ_i\}}\bigg)[\vO_{m}]}{1} \quad \text{for any fixed}\quad m \in M'_{\Pi}.
\end{align}
This effectively becomes a classical probability problem since the input is factorized. The case for $\vZ$ is simple
\begin{align}
    \CN_{\vec{j}, \{\vZ_j\}}[\vO_{m}] \stackrel{N\rightarrow\infty}=
    \begin{cases}
        \frac{\vO_{\Pi(\vec{j},m)}}{\norm{\vO_{\Pi(\vec{j},m)}}_1}&\quad \text{if each block of $m$ is acted by ``balanced'' top-bottom pairs of $\vZ_j$ and $\vZ^{\dagger}_j$}\\
        0&\quad \text{else}
    \end{cases}
\end{align}
where the convergence is accounted in 1-norm. Basically, the signs need to cancel each other within each block of $m$ for the expression to contribute (i.e., if $\vZ_3$ appears in the top row of a block, then $\vZ_3^{\dagger}$ must appear in the bottom row.). When they are perfectly balanced and the signs all cancelled, they give a mixture of random $\vO_{m'}$ that (up to $1/N$ correction in trace distance) corresponds to a particular partition $\Pi(\vec{j},m)\ge \Pi$ which refines the partition $\Pi$ of $m$ according the types of $\vZ_{j}$ present in the same block. In particular, this newly produced partition remains in $\Pi(\vec{j},m) \in B_{2k}$, and is propagating.
The case of $\vW_{j}$ is very much the same, up to a vanishing ``collision probability'' $\sim \frac{1}{N}$ (fixing the length of words and $k$). To see this, first sample the permutations within $\vW_{j}$ without the signs. With high probability $1-\CO(\frac{1}{N})$ (``without collision''), the signs for each $\vW_j$ in different blocks will be independent. That is, after sampling over signs (conditioned on the permutations), we have a situation very similar to $\vZ_i$s:
\begin{align}
    \CN_{\vec{j}, \{\vW_i\}}[\vO_{m}] \stackrel{N\rightarrow\infty}= \begin{cases}
        \frac{\vO_{\Pi(\vec{j},m)}}{\norm{\vO_{\Pi(\vec{j},m)}}_1}&\quad \text{if each block of $m$ is acted by ``balanced'' $\vW_j$ and its conjugates $\vW^{\dagger}_j$}\\
        0&\quad \text{else}.
    \end{cases}
\end{align}
When the signs perfectly cancel out, in the large-$N$ limit, distinct words acting on distinct inputs are independent from each other, which again gives $\frac{\vO_{\Pi(\vec{j},m)}}{\norm{\vO_{\Pi(\vec{j},m)}}_1}$. That is,
\begin{align}
     \lim_{N\rightarrow \infty}\lnormp{\bigg( \CN_{ \vec{j}, \{\vW_i\}} -\CN_{\vec{j}, \{\vZ_i\}}\bigg)[\vO_{m}]}{1} = 0\label{eq:Om_vanishes}.
\end{align}
Now, we use the above to prove for general unentangled PSD $\vZ$-symmetrized inputs $\vrho = \CN_{2k,\vZ}[\vrho]$. Any such input can be written as a weighted sum over (orthogonal) diagram operators squared
\begin{align}
    \vrho = \L(\sum_{\Pi\in B_{2k}} c_{\Pi} \frac{\vO'_{\Pi}}{\norm{\vO'_{\Pi}}_2} \R)\L(\sum_{\Pi\in B_{2k}} c_{\Pi} \frac{\vO'_{\Pi}}{\norm{\vO'_{\Pi}}_2} \R)^{\dagger}\quad \text{such that}\quad \sum_{\Pi} \labs{c_{\Pi}}^2 =1,\label{eq:rho_OOdagger}
\end{align}

In particular, we observe that each cross terms has bounded 1-norms
\begin{align}
    \lnormp{ \frac{\vO'_{\Pi}}{\norm{\vO'_{\Pi}}_2}\frac{\vO^{'\dagger}_{\Pi'}}{\norm{\vO'_{\Pi'}}_2}}{1} \le 1 \tag{Cauchy-Schwarz}
\end{align}
and that
\begin{align}
    \vO'_{\Pi}\vO'_{\Pi'} \propto\vO'_{\Pi''}\quad \text{for}\quad \Pi,\Pi',\Pi^{''} \in B_{2k}\label{eq:O'O'O'}.
\end{align}
This allows us to study the inputs as individual diagrams $\vO'_{\Pi}$ (instead of arbitrary linear combinations). Remarkably, $\vO'_{\Pi}$ (when normalization by its 1-norm) behaves very much like a classical probability distribution in the following sense
\begin{align}
\frac{\vO'_{\Pi}}{\norm{\vO'_{\Pi}}_{1}} = \frac{1}{M'_{\Pi}}\sum_{m\in M'_{\Pi}} \vO_m \quad \text{where}\quad \norm{\vO_m}_1 =1.\label{eq:OpiOm}
\end{align}
This is because the ``balanced'' structure allows us to calculate its 1-norm by permuting the top row $\norm{\vO_{\Pi}}_{1} = \norm{\vO_{\sigma} \vO_{\Pi}}_{1} = M'_{\Pi}$; otherwise the 1-norm normalization may not have a good connection to $\frac{1}{M'_{\Pi}}\sum_{m\in M'_{\Pi}} \vO_m$. It is crucial here that all blocks are propagating, as a nice feature of these ``balanced'' diagrams in $B_{2k}.$

Based on the above, we can then put together the bounds by decomposing inputs into $\vO_m$s (at a multiplicative cost depending only on $k$)
\begin{align}
    \lnormp{\bigg( \CN_{ \vec{j}, \{\vW_i\}} -\CN_{\vec{j}, \{\vZ_i\}}\bigg)[\vrho]}{1} &\le \sum_{\Pi,\Pi'} \labs{c_{\Pi}}\labs{c_{\Pi'}}\lnormp{\bigg( \CN_{ \vec{j}, \{\vW_i\}} -\CN_{\vec{j}, \{\vZ_i\}}\bigg)[\frac{\vO'_{\Pi}}{\norm{\vO'_{\Pi}}_2}\frac{\vO'_{\Pi'}}{\norm{\vO'_{\Pi'}}_2}]}{1}\tag{By~\eqref{eq:rho_OOdagger}}\\
    &\le  \sum_{\Pi,\Pi'} \labs{c_{\Pi}}\labs{c_{\Pi'}} \cdot \sup_{\Pi\in B_{2k}} \frac{\lnormp{ \bigg( \CN_{ \vec{j}, \{\vW_i\}} -\CN_{\vec{j}, \{\vZ_i\}} \bigg)[\vO'_{\Pi}] }{1} }{\normp{\vO'_{\Pi}}{1}}\tag{By~\eqref{eq:O'O'O'}}\\
    &\le \sum_{\Pi,\Pi'} \labs{c_{\Pi}}\labs{c_{\Pi'}} \cdot \sup_{\Pi\in B_{2k}} \sup_{m\in M'_\Pi} \lnormp{ \bigg( \CN_{ \vec{j}, \{\vW_i\}} -\CN_{\vec{j}, \{\vZ_i\}} \bigg)[\vO'_{m}] }{1}\tag{By~\eqref{eq:OpiOm}} \\
    &\stackrel{N\rightarrow \infty}=0 \tag{Since $\sum_{\Pi,\Pi'} \labs{c_{\Pi}}\labs{c_{\Pi'}}$ is independent of $N$ and by~\eqref{eq:Om_vanishes}}.
\end{align}
Accounting for ancillas (whose dimension is independent of $N$) incurring a $N$-independent multiplicative blowup (\autoref{lem:small_ancilla}) and inputs with positive and negative parts, we deduce that the diamond norm is also zero in the large $N$ limit.
\end{proof}

\begin{lem}[Norm properties from independence]\label{lem:freeness_random_P_NI}
    Fix an integer $\ell$ and consider a set of phased permutations $\vZ_{1}, \cdots \vZ_{\ell}\stackrel{i.i.d.}{\sim} \vZ$ such that all but at most one of them are independent uniformly random (with at most one identity). 
    Then, in the large-$N$ limit, for each $\vec{j} \in [\ell]^{2k}$, 

(1) If $\vZ_i$ are paired up:  
\begin{align}
    \lim_{N\rightarrow \infty} \normp{\CN_{2k,\vZ}^{\dagger}\circ\CN^{\dagger}_{\vec{j}, \{\vZ_i\}}[\vI] - \vI}{\infty} =0.
\end{align}

(2) Else:
\begin{align}
    \lim_{N\rightarrow \infty} \normp{\CN_{2k,\vZ}^{\dagger}\circ\CN^{\dagger}_{\vec{j}, \{\vZ_i\}}[\vI]}{\infty}=0.
\end{align}
\end{lem}
\begin{proof}[Proof of~\autoref{lem:freeness_random_P_NI}] 
We begin by evaluating
\begin{align}
    \CN^{\dagger}_{\vec{j}, \{\vZ_i\}}[\vI] = \BE \vW_{i_1}\otimes \cdots\otimes \vW_{i_k}
\end{align}
where the words are each the product of the left and right $\vZ_i$'s.  Following the arguments from the proof of~\autoref{lem:freeness_random_P}, we optimize over $\vO'_{\Pi}$ for $\Pi\in B_{2k}$ due to the symmetrization $\CN_{2k,\vZ}^{\dagger}$, and it suffices to control inputs as the ``factorized'' operators $\vO_m$:
\begin{align}
    \labs{\tr[\BE \vW_{i_1}\otimes \cdots\otimes \vW_{i_k}\vO_m]}&=\labs{\BE \bra{m_{k+1}}\vW_{i_1}\ket{m_1}\otimes \cdots\otimes \bra{m_{2k}}\vW_{i_k}\ket{m_k} }\\
    &= \CO(\frac{k}{N}) \quad \text{if any $\vW_{i_j}$ are nontrivial}.
\end{align}
Thus, unless all $\vZ_i$ are paired up with their conjugates,
\begin{align}
    \normp{\CN_{2k,\vZ}^{\dagger}\circ\CN^{\dagger}_{\vec{j}, \{\vZ_i\}}[\vI]}{\infty} & = \sup_{\norm{\vrho}_1=1} \labs{\tr[ \CN_{2k,\vZ}^{\dagger}[\BE \vW_{i_1}\otimes \cdots\otimes \vW_{i_k}]\vrho]}\\
    &\le f(k) \sup_m \labs{\tr[ \BE \vW_{i_1}\otimes \cdots\otimes \vW_{i_k}\vO_m]} \tag{By~\eqref{eq:rho_OOdagger},\eqref{eq:O'O'O'}, and~\eqref{eq:OpiOm}.}\\
    &\stackrel{N\rightarrow \infty}{=} 0,
\end{align}

as advertised.
\end{proof}

\begin{cor}[A rescaled CPTP map]\label{cor:rescaled}
\begin{align}
    \lim_{N\rightarrow \infty} \lnormp{\CN_{2k,\vS}^{\dagger}\circ \CN^{\dagger}_{\vec{j}, \{\vZ_i\}}[\vI] - \vI \cdot \btr[\CN_{2k,\vS}^{\dagger}\circ\CN^{\dagger}_{\vec{j}, \{\vZ_i\}}[\vI]]}{\infty} = 0.
\end{align}
\end{cor}
\begin{cor}[Complete positivity in the large-$N$ limit]\label{cor:CP_large_N}
Suppose $\sum_i \labs{w_i}^2 =1$, then the CP map associated with $2k$-fold tensor product of $\sum_i w_i \vZ_i$ is trace-preserving in the large-$N$ limit
    \begin{align}
        \lim_{N\rightarrow \infty} \normp{\CN_{2k,\vS}^{\dagger} \circ  \CN^{\dagger}_{2k,\sum_i w_i \vZ_i}[\vI] - \vI}{\infty} =0.
    \end{align}
\end{cor}
That is, $\CN_{2k,\sum_i w_i \vZ_i}\circ \CN_{2k,\vS}$ is a rescaled trace-preserving map. They are also completely positive as the $\vZ_i$ need to pair up with each other.

\begin{proof}
    Expand 
\begin{align}
    \CN^{\dagger}_{2k,\sum_i w_i \vZ_i}[\vI] &= \BE \left(\sum_i w_i \vZ_i\right)^{\dagger}\left(\sum_i w_i \vZ_i\right) \otimes \left(\sum_i w_i \vZ_i\right)^{\dagger}\left(\sum_i w_i \vZ_i\right)\cdots \\
    & = \vI^{\otimes 2k} + (\text{cross terms})\tag*{($\vZ^{\dagger}\vZ = \vI$)}\\
    & \stackrel{N \rightarrow \infty}{=} \vI^{\otimes 2k} \tag*{(\autoref{lem:freeness_random_P_NI})},
\end{align}
where the cross terms result from multiplying out the bilinear tensor products. For example, for $k=2$, 
\begin{align}
    (\sum_i w_i \vZ_i)^{\dagger}(\sum_i w_i \vZ_i) \otimes (\sum_i w_i \vZ_i)^{\dagger}(\sum_i w_i \vZ_i) &= \L(\vI +  \sum_{i\ne j} w_iw_j \vZ_i^{\dagger}\vZ_j \R)\otimes  \L(\vI +  \sum_{i\ne j} w_iw_j \vZ_i^{\dagger}\vZ_j \R)\\
    &= \vI\otimes \vI + (\text{cross terms}).
\end{align}
Note that the number of cross-terms is independent of $N$.
\end{proof}
\begin{cor}[Combination of permutations and Gaussians]\label{cor:Z+G}
    The above holds with any linear combination $\sum_i w_i \vZ_i + \sum_i w'_i \vG_i$ provided that $\sum_i \labs{w_i}^2 +\sum_i \labs{w'_i}^2 =1$.
\end{cor}
\begin{proof}
    Recall the central limit theorem (for any fixed dimension $N$)
    \begin{align}
        \lim_{m\rightarrow \infty} \sum_{i=1}^m \frac{1}{\sqrt{m}}\vZ_i\stackrel{dist.}{\sim} \vG.
    \end{align}
The error in the central limit theorem is bounded in the operator norm by 
\begin{align}
    \frac{1}{m} f_1(w_i,k) +\frac{1}{m^2} f_2(w_i,k) + \cdots, 
\end{align}
where the leading contributions are cases when each $\vZ_i$ shows up either two or zero times (reproducing the Wick contractions up to $\CO(1/m)$ corrections), and the subleading contributions are cases when there are terms $\vZ_i$ that show up four times.
In particular, the norm bounds of the error terms are functions independent of $N$ (depending on the pattern of the occurrences among the $k$ indices and the $w_i$ coefficients), so the limits $m\rightarrow \infty$ and $N\rightarrow \infty$ commute. 
\end{proof}
\subsection{\texorpdfstring{Large $N$ structure}{Large N structure}}
\label{sec:partitionAlgebra_largeN}
In order to use the Markov inequality, we need to confirm that the distinguishing probability is indeed a rational polynomial of $N$. Nicely, expressions involving random permutations, like random unitary, have a pronounced large-$N$ expansion.

\begin{lem}[Large-$N$ expression in random phased permutations]\label{lem:ONO_polyN}
Consider $\CN_{ \vec{j}, \{\vW_i\}}$ as in~\autoref{lem:freeness_random_P}, and suppose the length of the words is bounded by $\ell$. Then, for each pair of partitions $\Pi_1,\Pi_2 \vdash 2k$,
\begin{align}
    \tr\L[ \vO_{\Pi_2}\CN_{ \vec{j}, \{\vW_i\}}[\vO_{\Pi_1}]\R]
\end{align}
is a rational function in $N$ where the poles can be located up to $2k\ell-1$, each with multiplicity at most $2k\ell$.
\end{lem}

The strategy is to expand the expression by canonical pieces of tensors (\autoref{fig:components}) and argue that any contraction must yield a polynomial of $N$. This formal dependence on $N$ is the only structural property we need; no further combinatorial calculations are required.

\begin{figure}[t]
    \centering
    \includegraphics[width=0.8\textwidth]{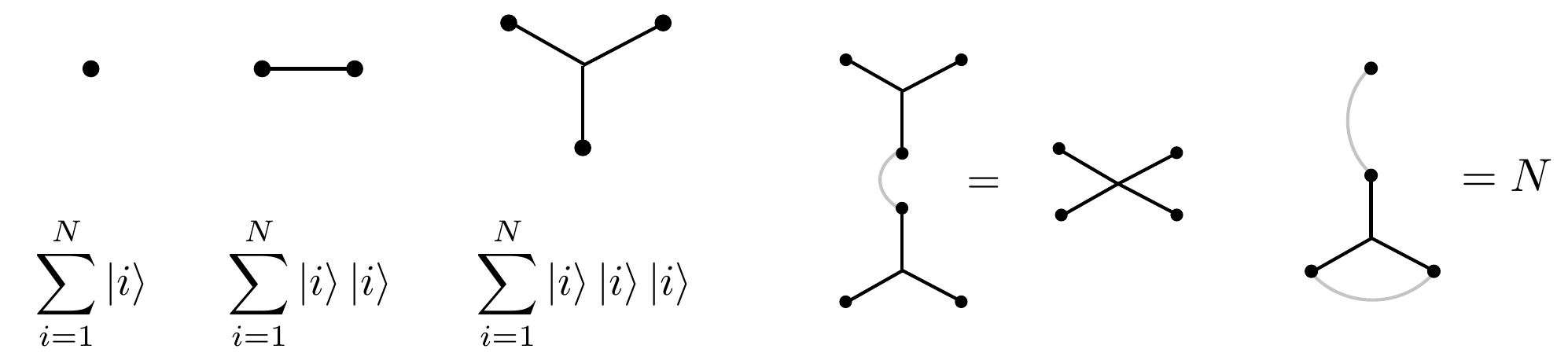}
    \caption{Diagrams for tensor components. Each dot represents a node, and each star-like component is a sum over index $i$.}
    \label{fig:components}
\end{figure}

\begin{lem}[Breaking $\vO_{\Pi}$ into components]\label{lem:breaking_into_tensors} The vectorized operator is factorized into individual tensors
\begin{align}
    |\vO_{\Pi}) = \bigotimes_i |T_{b}) \quad \text{where}\quad |T_b) = \sum_{i=1}^N \ket{i}^{\otimes t_b}
\end{align}
and $\sum_b t_b = 2k$.
\end{lem}

\begin{lem}[The power of $N$ is the number of components]
    For any tensor contraction made of components $T_i$, it evaluates exactly to 
    \begin{align}
        N^{c}\quad \text{where} \quad c = \#(\text{connected components after full contraction}).
    \end{align}
\end{lem}
In particular, this implies~\eqref{eq:trO1O2} as a special case.
\begin{lem}[Expanding diagonal phases as diagrams]
There exists coefficients $\alpha(\Pi)$ independent of $N$ such that
    \begin{align}
        |\BE \vD_{z}^{\otimes k} \otimes \vD_{z}^{\dagger \otimes k}) =\sum_{\Pi \vdash 4k} \alpha(\Pi)|\vO_{\Pi}).
    \end{align}
\end{lem}
\begin{proof}
    \begin{align}
        |\BE \vD_{z}^{\otimes k} \otimes \vD_{z}^{\dagger \otimes k})&=\BE  \L[(\sum_{i=1}^N z_i\ket{i}\ket{i})^{\otimes k} \otimes (\sum_{i=1}^N z^*_i\ket{i}\ket{i})^{\otimes k}\R]\\
        &= \sum_{i=1}^N z_i^kz_i^{*k} (\ket{i}\ket{i} \otimes \ket{i}\ket{i})^{\otimes k} + \sum_{N\ge i_2>i_1\ge 1} (\cdot) + \cdots +\sum_{N\ge i_{2k}>\cdots > i_1\ge 1}(\cdot)\\
        &= \sum_{\Pi \vdash 4k} \alpha'(\Pi)|\vO'_{\Pi})\\
        &=\sum_{\Pi \vdash 4k} \alpha(\Pi)|\vO_{\Pi})\tag*{(M\"obius inversion~\eqref{eq:Kinv})}.
    \end{align}
    The second line sums over possible patterns of distinct $z_i$s, and uses that each $z_i$ must pair with its complex conjugate since 
    \begin{align}
        \BE[ z^b_i] = \begin{cases}
            1 &\text{if integer }\quad b=0\\
            0 &\text{else}
        \end{cases}.
    \end{align}
    Thus, the nonvanishing contribution has merely unity coefficient $\labs{z_i}^2=1$. Note that the Hilbert space dimension is $N^{4k}$ instead of $N^{2k}$ as the object we consider is a (vectorized) superoperator.
\end{proof}

\begin{proof}[Proof of~\autoref{lem:ONO_polyN}]
Decompose the phased permutations by the diagonal phases and the permutation $\vZ = \vD_z \vS$. Then, the superoperator $\CN_{ \vec{j}, \{\vW_i\}}$ can be decomposed into sum and products of 
\begin{align}
    \BE \vS^{a} \otimes \vS^a, \quad   \BE \vD_z^{a} \otimes \vD_z^{*a}\quad \text{for }\quad  1\le a \le 2k\ell
\end{align}
which can be written as diagrams $\vO_{\Pi}$ and expanded as tensor components (\autoref{lem:breaking_into_tensors}). After contraction, the poles may accumulate through the occurrences of $\BE \vS^{a} \otimes \vS^a$ for different possible independent permutations $\vS_1,\vS_2...$. These expectations produce poles via~\eqref{eq:average_projector} such that after taking account of all of them, the poles can be located up to $2k\ell-1$, each with multiplicity at most $2k\ell$. \footnote{This is likely a loose bound but is qualitatively sufficient.}
\end{proof}

\subsection{Basis for the partition algebra}
\label{sec:basis_parti}

The crux of our large-$N$ interpolation argument is to extend the distinguishing probability for $N$ to other values of $N'\ne N$. In particular, we need to ``smoothly'' extend the test operators $\vrho, \vO$ to different other dimensions. Fortunately, we do know that the partition algebra $\vP_{k}(N)$ \textit{stabilizes} for large enough $N\ge 2k$ and the individual factors $\{\vP_{\vec{\lambda}}\}_{\vec{\lambda}}$ depend \textit{only} on $k$ and are \textit{independent} of $N$ if $N\ge 2k$ (\autoref{thm:stablize}). Naturally, it is tempting to keep the test operator the same in this invariant decomposition into factors.

However, our polynomial interpolation argument requires that the block-diagonalizing unitary transformation ``depends nicely'' on $N$ in the following sense. 

\begin{quest}[Writing irreps in the computation basis] 
\label{quest:basischange}
For each $k$, assume $N\ge 2k$. Does there exist a choice of orthonormal basis $\ket{v^{\vec{\lambda}}_i}$ for each factor $\vP_{\vec{\lambda}}$ of the partition algebra such that the following holds: the matrix elements $\vE_{ij}^{\vec{\lambda}} = \ket{v^{\vec{\lambda}}_i}\bra{v_j^{\vec{\lambda}}} \otimes \vI_{\vS_{\vec{\lambda}}} $ satisfy that
\begin{align}
    \tr[\vE_{ij}^{\vec{\lambda}}\vO_{\Pi}] \quad \text{is a degree $\CO(k)$ rational polynomial of $N$ for each}\quad \vec{\lambda}, i, j, \Pi?
\end{align}
\end{quest}

We do not need \textit{any} knowledge of the rational polynomial except for its \textit{degree}, and the polynomial can very well depend on $\vec{\lambda}, i, j, \Pi$. To answer the above question, we need to dive into the partition algebra and construct an orthogonal basis. Although our ensemble actually uses a stronger symmetry $\vZ$ (random phased permutations) instead of $\vS(N)$, it is sufficient to consider the weaker symmetry.

\subsubsection{An abstract construction of irreps}
For each irrep labeled by $\vec{\lambda}\in \Lambda_{k,N}$, \cite{halverson2020set} constructed the corresponding irreducible module by taking an abstract quotient
\begin{align}
    \CV_{\vec{\lambda}} 
        &:= \frac{\vP_{k}(N) \vec{e}_{\vec{\lambda}}}{\vJ_{m-1}(N) } \quad \text{as a $\BC$-vector space for $m:=\labs{\vec{\lambda}^*}$}\\ 
    \quad &\text{where}\quad \vec{v}_1( \vec{v}_2+ \vJ_{m-1}) \equiv \vec{v}_1 \vec{v}_2 + \vJ_{m-1}
    \quad \text{and} \\
    &\quad \quad \quad \vec{e}_{\vec{\lambda}}: = \vec{p}_{\vec{\lambda}^*}\vec{e}_{m}\quad \text{for}\quad     \vec{e}_{m}:= \frac{1}{N^{k-m}}\vI^{m}\otimes (\text{isolated nodes})
\end{align}
such that $\vec{e}_m^2=\vec{e}_m$. The diagrammatic form of $\vec{e}_m$ is depicted in \autoref{fig:em}.
\begin{figure}[t]
    \centering
    \includegraphics[width=0.35\textwidth]{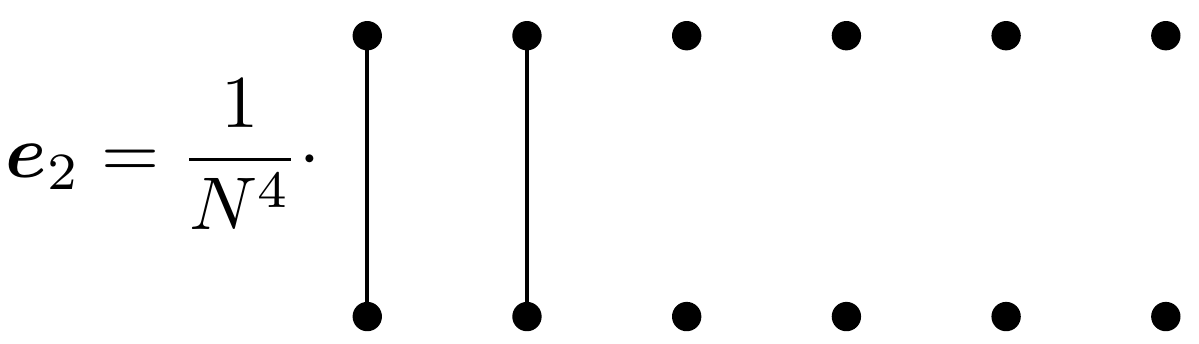}
    \caption{The diagram for $\vec{e}_m$ for $m=2$. It has identities on the leftmost nodes, and the rest are isolated nodes. The normalization is such that $\vec{e}_m^2=\vec{e}_m$.}
    \label{fig:em}
\end{figure}
\begin{thm}[\cite{halverson2020set}]\label{thm:abstract_construction_irreps}
    The quotient $\CV_{\vec{\lambda}}$ is isomorphic to the irreducible module labeled by $\vec{\lambda}\in \Lambda_{k,N}$.
\end{thm}

Since we are interested in finding an orthonormal basis, we first need an \textit{inner product} on this irreducible module $\CV_{\vec{\lambda}}$ (which was not discussed in~\cite{halverson2020set}). There is a very natural approach: the abstract quotient by the two-sided ideal $\vJ_{m-1}$ can be realized as the projector $\CQ_{\vJ_{m-1}}$ onto a subalgebra $\CQ_{\vJ_{m-1}}[\vP_{k}(N)] \subset \vP_{k}(N)$. (See \autoref{thm:structure_finite_dim}.) Therefore, the natural choice for the inner product is then to simply induce from the Hilbert-Schmidt inner product of the partition algebra by
\begin{align}
    \braket{ v^{\vec{\lambda}}_1, v^{\vec{\lambda}}_2 } := \frac{1}{\tr[\vI_{\vS_{\lambda}}]} \tr[\CQ_{\vJ}[ \vec{v}^{\vec{\lambda}}_1]^{\dagger} \CQ_{\vJ}[  \vec{v}^{\vec{\lambda}}_2]]\quad \text{for each}\quad \vec{v}^{\vec{\lambda}}_1,\vec{v}^{\vec{\lambda}}_2 \in \vP_{k}(N) \vec{e}_{\vec{\lambda}}\label{eq:innerproduct_lambda}.
\end{align}
We can verify that this inner product is \textit{consistent} (dropping scripts ${\vec{\lambda}}$ and $m-1$ for the moment)
\begin{align}
     \braket{ v_1+\vJ_1, v_2+\vJ_2 }_{\vec{\lambda}}&= \frac{1}{\tr[\vI_{\vS_{\lambda}}]} 
    \tr[\CQ_{\vJ}[ \vec{v}_1+\vJ_1] \CQ_{\vJ}[  \vec{v}_2+\vJ_2]]\\
    &=  \frac{1}{\tr[\vI_{\vS_{\lambda}}]} \tr[\CQ_{\vJ}[ \vec{v}_1]^{\dagger} \CQ_{\vJ}[  \vec{v}_2]] \\
    &= \braket{ v_1, v_2}_{\vec{\lambda}} \quad \text{for each}\quad \vec{v}_1,\vec{v}_2 \in \vP_{k}(N) \vec{e}_{\vec{\lambda}}\quad \text{and}\quad \vJ_1,\vJ_2 \in \vJ
\end{align}
and \textit{compatible} with the adjoint from the Hibert space on which $\vP_{k}(N)$ is represented
\begin{align}
    \braket{ v_1, \vO v_2 }_{\vec{\lambda}} &=  \frac{1}{\tr[\vI_{\vS_{\lambda}}]} \tr[\CQ_{\vJ}[ \vec{v}_1]^{\dagger} \CQ_{\vJ}[ \vO\vec{v}_2]]\\
    &=  \frac{1}{\tr[\vI_{\vS_{\lambda}}]} \tr[\CQ_{\vJ}[ \vec{v}_1]^{\dagger} \CQ_{\vJ}[ \vO]\CQ_{\vJ}[ \vec{v}_2]]\\
    &=  \frac{1}{\tr[\vI_{\vS_{\lambda}}]} \tr[\L(\CQ_{\vJ}[ \vO]^{\dagger} \CQ_{\vJ}[ \vec{v}_1]\R)^{\dagger} \CQ_{\vJ}[ \vec{v}_2]]\\
    &=  \frac{1}{\tr[\vI_{\vS_{\lambda}}]} \tr[\L( \CQ_{\vJ}[ \vO^{\dagger}\vec{v}_1]\R)^{\dagger} \CQ_{\vJ}[ \vec{v}_2]] = \braket{ \vO^{\dagger} v_1,  v_2 }_{\vec{\lambda}}.
\end{align}
The second, third, and fourth lines use that $\CQ_{\vJ}[\cdot]$ is a $\dagger$-homomorphism. (Recall~\autoref{defn:dagger_homo} and \autoref{thm:structure_finite_dim}.) 

In a nutshell, there is a very simple intuition for this abstract vector space: it is simply one \textit{column} in the corresponding factor of the partition algebra
\begin{align}
    \frac{\vP_{k}(N) \vec{e}_{\vec{\lambda}}}{\vJ_{m-1}(N) } \sim \BC\text{-Span}\{\ket{v^{\vec{\lambda}}_i}\bra{0}\} \in \vP_{\vec{\lambda}}.
\end{align}

\subsubsection{Explicit basis}\label{sec:explicit_basis}

For our later usage, we will also need to extract the explicit structure of the basis from the abstract construction. We will be building on the basis provided by~\cite{halverson2020set}. The construction requires some further concepts from the algebraic structure of the partition algebra (see~\autoref{fig:F_NF}).
\begin{defn}[$m$-factor]
    For any $0 \le m\le k$, the set of $m$-factors $F_m$ are diagrams where the bottom leftmost $m$ nodes propagate and the rightmost $k-m$ nodes are isolated. 
    In particular, the set of noncrossing $m$-factors $NF_m$ are those $m$-factors whose propagating edges do not cross when the diagram is drawn under the rule that propagating edges connect to the rightmost vertex of the block in the top row.
\end{defn}
Note that by permuting the bottom row (i.e., concatenating a permutation below), any $m$-factor can be made noncrossing.

\begin{lem}[A basis for $\CV_{\vec{\lambda}}${~\cite{halverson2020set}}]\label{lem:different_reps}
Let $m = \labs{\vec{\lambda}^*}$. Then,
\begin{align}
      \CV_{\vec{\lambda}} = \text{Span}\{\vO_{\omega}\vsigma_t \vec{p}_{\vec{\lambda}^*}| \omega \in NF_m , t \in SYT(\vec{\lambda}^*) \} \quad (\text{mod } \vJ_{m-1}).
\end{align}
\end{lem}
We now construct an explicit \textit{orthonormal} basis for the vector space $\CV_{\vec{\lambda}}$ for each irrep $\vec{\lambda}\in \Lambda_{k,N}$. This might be of independent interest in the representation theory of partition algebras.

\begin{thm}[Orthonormal basis for $\CV_{\vec{\lambda}}$]\label{thm:ourExplicitBasis}
For each $\vec{\lambda}\in \Lambda_{k,N}$, the set of vectors
\begin{align}
    \{ \vec{v}^{\vec{\lambda}}_i\} =  \L\{ \frac{1}{n(\omega,t)}\hat{\vO}_{\omega}\vec{u}_t \vec{p}_{\vec{\lambda}^*}| \omega \in NF_m , t \in SYT(\vec{\lambda}^*) \R\}\quad (\text{mod}\quad \vJ_{m-1})
\end{align}
forms an orthornormal basis for the irreducible representation $\CV_{\vec{\lambda}}$.
\end{thm}
Note that as compared to \autoref{lem:different_reps}, $\vsigma_t$ has been replaced by the $\vec{u}_t$ introduced in \autoref{lem:ONbasis_Sm} and $\vO_\omega$ was replaced by $\hat{\vO}_\omega$. Roughly speaking,
for a noncrossing $m$-factor $\omega$,  
\begin{align}
    \hat{\vO}_{\omega} := (\text{sum over distinct indices for the upper nodes } \{1,\cdots, k\}\text{ and lower nodes }\{1,\cdots, m\}).\label{eq:hatvO}
\end{align}
The precise definition can be found in the proof of \autoref{lem:OhatO_mfactor}. These operators are analogs of the $\vO'_{\omega}$ but focus only on a subset of vertices. (That is, when running the inclusion-exclusion principle, we keep the lower right $k-m$ nodes untouched.) This will be crucial for ensuring that the basis vectors are orthogonal and the subspace projector remains a \textit{rational function} of $1/N$.\footnote{This point is surprisingly delicate because the subspace dimension is large (superexponential in $k$), and tiny overlaps can accumulate and wiggle the span.} Later, we will choose the normalization constant $n(\omega)$ to ensure the vectors are exactly unit length (\autoref{lem:orthonormal}). 
The proof that these operators form an orthonormal basis will take place in two steps.
\begin{itemize}
    \item The new vectors span the same space (\autoref{lem:OhatO_mfactor} and \autoref{lem:same_span}).
    \item The new vectors are orthogonal (\autoref{sec:proving_ortho_basis}).
\end{itemize}

\begin{figure}[t]
    \centering
    \includegraphics[width=0.5\textwidth]{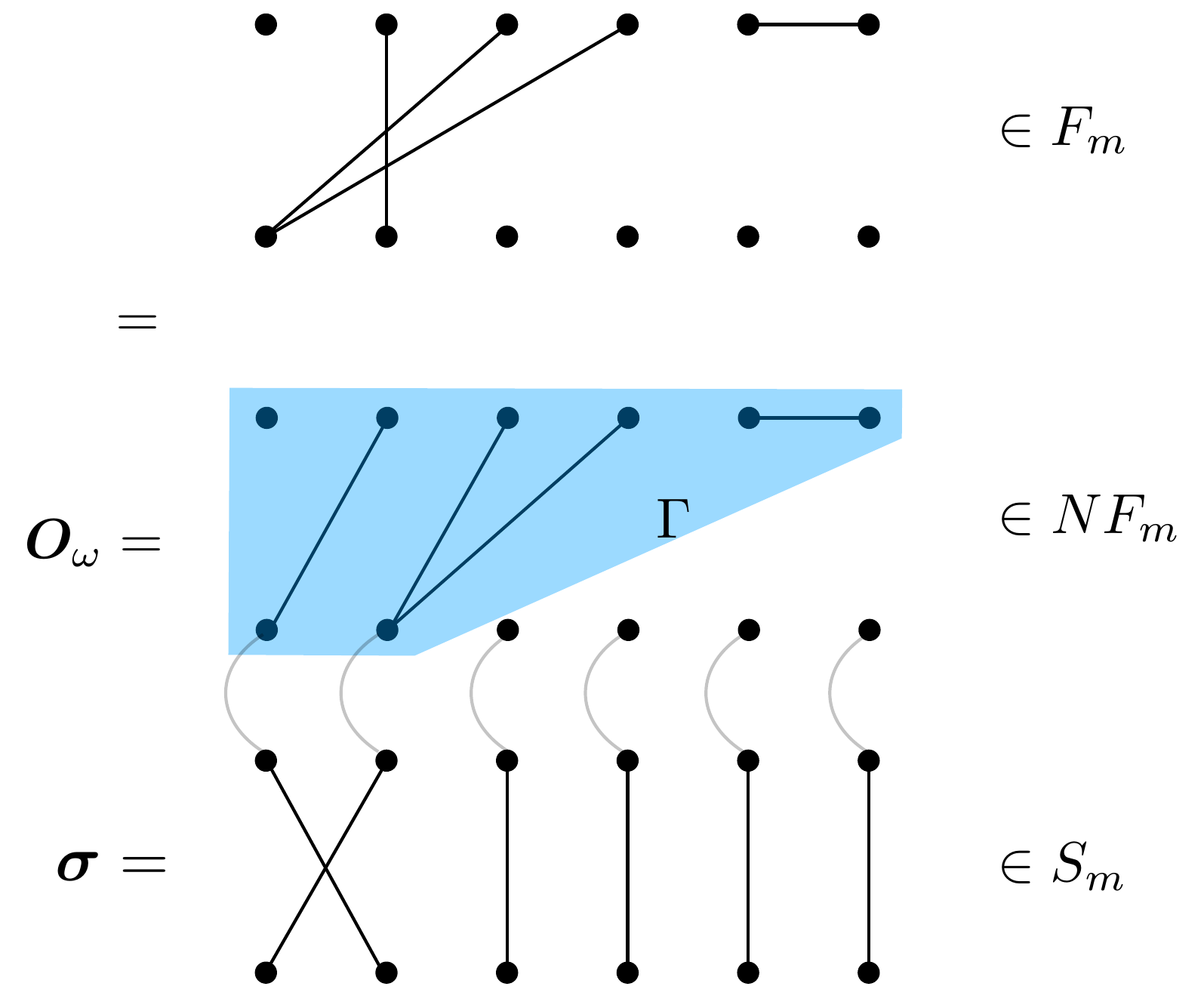}
    \caption{Illustration for $m$-factors ($F_m$) for $m=2$, which can also be written as a noncrossing $m$-factor ($\omega\in NF_m$) multiplied by a permutation $\vsigma \in \vS_m$ at the bottom (which uniquely makes it noncrossing). }
    \label{fig:F_NF}
\end{figure}

\begin{lem}[Alternate representatives]\label{lem:OhatO_mfactor}
For each integer $1\le m\le k$,
    \begin{align}
    \text{Span}\{\hat{\vO}_{\omega}| \omega \in F_m \} \equiv \, \text{Span}\{\vO_{\omega}| \omega \in F_m\} \quad (\text{mod }\vJ_{m-1} ).
    \end{align}
\end{lem}
\begin{proof}[Proof of~\autoref{lem:OhatO_mfactor}]
To study $m$-factors, it suffices to consider diagrams where the bottom right $k-m$ nodes are isolated. To that end, let
$T_1^{k-m} =(\sum_{i=1}^N\bra{i})^{\otimes k-m}$ and introduce the shorthand
\begin{align}
    \vO_{\Gamma\otimes T_1^{k-m}} =:\vO_{\Gamma}\quad \text{where}\quad \Gamma\vdash[k+m]=\{1,\cdots,k,k+1,\cdots k+m\}.
\end{align}
The induced partial order $\ge$ on $\Gamma \vdash [k+m]$ (recall~\autoref{sec:set_partitions}) allows us to run a tailored M\"obius inversion argument so that (as in~\eqref{eq:K},\eqref{eq:Kinv})
\begin{align}
    \vO_{\Gamma} &=  \sum_{\stackrel{\Gamma' \ge \Gamma,}{\Gamma'\vdash[k+m]}} \vO'_{\Gamma'} =: \sum_{\Gamma'\vdash[k+m]} K_{\Gamma\Gamma'}\vO'_{\Gamma'}\\
    \vO'_{\Gamma} &= \sum_{\Gamma'\vdash[k+m]} K^{-1}_{\Gamma\Gamma'}\vO_{\Gamma'}\label{eq:O'Gamma}
\end{align}
where $\vO'_{\Gamma'}$ are orthogonal (recall~\eqref{eq:O'pi}: distinct blocks in $\Gamma$ have distinct indices) and are related to the original diagrams $\vO_{\Gamma}$ by upper triangular matrices $K_{\Gamma\Gamma'}$ and $K^{-1}_{\Gamma\Gamma'}$ (total order extended from the partial order $\ge$ by~\autoref{lem:extensionPartialOrder}). Note that the bottom right $k-m$ nodes are untouched throughout the proof. Consider~\eqref{eq:O'Gamma} for each $m$-factor $\omega \in F_m$, 
\begin{align}
\vO'_{\omega} &= \sum_{\Gamma'\vdash[k+m]} K^{-1}_{\omega\Gamma'}\vO_{\Gamma'}\\
&=\sum_{\omega'\in F_m} K_{\omega\omega'}^{-1}\vO_{\omega'} + \sum_{\Gamma'\in \vJ_{m-1}} K_{\omega\Gamma'}^{-1}\vO_{\Gamma'} 
    \tag{Any partition $\Gamma$ is either in $F_m$ or $\vJ_{m-1}$}\\
        &\equiv \sum_{\omega'\in F_m} K_{\omega\omega'}^{-1}\vO_{\omega'}\quad (\text{mod $\vJ_{m-1}$}),    
\end{align}  
where we again use $\Gamma' \in \vJ_{m-1}$ as a shorthand for $\Gamma'\otimes T_1^{k-m} \in \vJ_{m-1}$.
The second line is the observation that any partition $\Gamma\vdash[k+m]$ corresponds to a diagram with at most $m$ propagating blocks (since there are only $m$ nodes at the bottom), and thus any diagram with (1) exactly $m$-propagating blocks (i.e., an $m$-factor) is in $F_m$ and those with (2) at most $m-1$ propagating block is by definition in $\vJ_{m-1}$. Now, since the restriction of an upper triangular matrix remains upper triangular (or, partial order remains consistent on subsets:~\autoref{lem:restricting_partial_order}), we have that
\begin{align}
    K^{-1}_{\Gamma\Gamma'}\quad&\text{for}\quad \Gamma,\Gamma' \vdash [k+m]\quad \text{is upper triangular implies the same for}\\
    K^{-1}_{\omega\omega'}\quad &\text{for}\quad \omega,\omega' \in F_m. 
\end{align}

Further, since the diagonals are nonzero (just ones), there exists a matrix $K^{(F_m)}$ defined only on $F_m\rightarrow F_m$ that inverts $K^{-1}_{\omega\omega'}$ on $F_m\rightarrow F_m$:
\begin{align}
    \sum_{\omega'\in F_m} K^{(F_m)}_{\omega\omega'} K^{-1}_{\omega'\omega''} = \delta_{\omega\omega''} \quad \text{for each}\quad \omega, \omega'' \in F_m.
\end{align}
Thus, 
\begin{align}
     \sum_{\omega'\in F_m} K^{(F_m)}_{\omega\omega'} \vO'_{\omega'} \equiv \vO_{\omega} \quad (\text{mod $\vJ_{m-1}$}), 
\end{align}
stating that the two spans are equal up to elements in $\vJ_{m-1}$.
\end{proof}

\begin{lem}[Same span for irreps]\label{lem:same_span}
    For each $\vec{\lambda}\in \Lambda_{k,N}$, let $m=\labs{\vec{\lambda}^*}$, then
    \begin{align}
\text{Span}\{\vO_{\omega}\vsigma_t \vec{p}_{\vec{\lambda}^*}| \omega \in NF_m , t \in SYT(\vec{\lambda}^*)\} &\equiv \text{Span}\{\hat{\vO}_{\omega} \vec{u}_t \vec{p}_{\vec{\lambda}^*}| \omega \in NF_m , t \in SYT(\vec{\lambda}^*)\} \quad (\text{mod } \vJ_{m-1})
    \end{align}
\end{lem}
\begin{proof}[Proof of~\autoref{lem:same_span}]
Starting with~\autoref{lem:OhatO_mfactor}, 
\begin{align}
    \text{Span}\{\hat{\vO}_{\omega}| \omega \in F_m\} &\equiv \quad \text{Span}\{\vO_{\omega}| \omega \in F_m\} \quad (\text{mod }\vJ_{m-1}),
\end{align}
we can pull out permutations to make the partitions noncrossing, giving
\begin{align}
    \text{Span}\{\hat{\vO}_{\omega}\vsigma | \omega\in NF_m, \vsigma \in S_m \} &\equiv \quad \text{Span}\{\vO_{\omega}\vsigma | \omega\in NF_m, \vsigma \in S_m \} \quad (\text{mod }\vJ_{m-1}),
\end{align}
noting that we drop the distinctness constraint on the permutation as the distinctness of $\hat{\vO}_{\omega}$ suffices to ensure $\hat{\vO}_{\omega}\vsigma = \hat{\vO}_{\omega'}$ for some $\omega \in F_m$. Then, we can right multiply the set by
\begin{align}
    \text{Span}\{\vO_{\omega}\vsigma\vec{p}_{\vec{\lambda}^*}| \omega \in NF_m, \vsigma \in S_m \} &\equiv \text{Span}\{\hat{\vO}_{\omega}\vsigma\vec{p}_{\vec{\lambda}^*}| \omega \in NF_m , \vsigma \in S_m \} \quad (\text{mod}\ \vJ_{m-1}),
\end{align}
noting that we relax the quotient $(\text{mod}\ \vJ_{m-1}\vec{p}_{\lambda^*})$ to $(\text{mod}\ \vJ_{m-1})$ by the ideal property $\vJ_{m-1}\vec{p}_{\lambda^*}\subset \vJ_{m-1}$.

By \autoref{lem:irrep_Sm}, however, $\BC\{\vsigma\vec{p}_{\vec{\lambda}^*}\}_{\vsigma\in S_m}=\BC\{\vsigma_t\vec{p}_{\vec{\lambda}^*}\}_{t \in SYT(\vec{\lambda}^*)}$, which allows us to restrict the set of permutations by
\begin{align}
    \text{Span}\{\vO_{\omega}\vsigma_t \vec{p}_{\vec{\lambda}^*}| \omega \in NF_m , t \in SYT(\vec{\lambda}^*)\} &\equiv \text{Span}\{\hat{\vO}_{\omega} \vsigma_t \vec{p}_{\vec{\lambda}^*}| \omega \in NF_m , t \in SYT(\vec{\lambda}^*)\} \quad (\text{mod } \vJ_{m-1}).
\end{align}
Finally, we can use \autoref{lem:ONbasis_Sm} and, in particular, that $\text{Span}\{\vec{u}_t \vec{p}_{\vec{\lambda}^*}\}= \text{Span}\{\vsigma_t \vec{p}_{\vec{\lambda}^*}\}$, to reach
\begin{align}
    \text{Span}\{\vO_{\omega}\vsigma_t \vec{p}_{\vec{\lambda}^*}| \omega \in NF_m , t \in SYT(\vec{\lambda}^*)\} &\equiv \text{Span}\{\hat{\vO}_{\omega} \vec{u}_t \vec{p}_{\vec{\lambda}^*}| \omega \in NF_m , t \in SYT(\vec{\lambda}^*)\} \quad (\text{mod } \vJ_{m-1}),
\end{align}
as advertised.
\end{proof}

\subsubsection{Orthogonality and matrix elements}\label{sec:proving_ortho_basis}

The very nice properties of the basis vectors we chose in \autoref{thm:ourExplicitBasis} are that they are orthogonal and immediately give the matrix elements of the factors of the partition algebra.
\begin{lem}[Orthonormal basis for each irrep of $\vP_{k}(N)$]\label{lem:orthonormal}
    For each $\vec{\lambda}\in \Lambda_{k,N}$, 
    set the normalization such that
    \begin{align}
             n(\omega,t) := \sqrt{c(\omega)} \quad \text{for each}\quad \omega \in NF_m\quad \text{and}\quad t \in SYT(\vec{\lambda}^*).
    \end{align}
    where 
   \begin{align}
    c(\omega) := \frac{N^{k-m}(N-m)!}{(N-m-d(\omega^{\dagger},\omega))!} = N^{k-m}\cdot (N-m) \cdot (N-m-1)\cdots \le N^{k-m+d(\omega^{\dagger},\omega)}.
   \end{align} 
   Then, $\ket{v_j^{\vec{\lambda}}}$ forms an orthonormal basis for the irrep labeled by $\vec{\lambda}\in \Lambda_{k,N}$, and 
    \begin{align}
        \CQ_{\vJ}[\vec{v}_i^{\dagger}\vec{v}_j] \simeq \ket{v^{\lambda}_i}\bra{v^{\lambda}_j}\otimes \vI_{\vS_{\vec{\lambda}}} \quad \text{for each}\quad i,j
    \end{align}
    gives the matrix elements for the irrep labeled by $\vec{\lambda}\in \Lambda_{k,N}$.
\end{lem}
The proof of orthogonality relies on the following multiplicative properties of $\hat{\vO}_{\omega}$.  This is precisely why we define our basis using $\hat{\vO}_{\omega}$ instead of $\vO_{\omega}$; despite having a simpler definition, the latter has a more involved algebraic structure.
\begin{lem}[Multiplicative properties of $\hat{\vO}_{\omega}$]For all $\omega, \omega' \in NF_m$, \label{lem:multiplicative_Oomega}
\begin{align}
    \hat{\vO}_{\omega'}^{\dagger}\hat{\vO}_{\omega}\equiv c(\omega) \delta_{\omega'\omega}\cdot \vec{e}_m
     \quad\text{(mod $\vJ_{m-1}$)}\label{eq:delta_omegaomega'}.
\end{align}    
\end{lem}
The above will be very helpful in the calculation as it imposes $\omega = \omega'$.

\begin{figure}[t]
    \centering
    \includegraphics[width=0.4\textwidth]{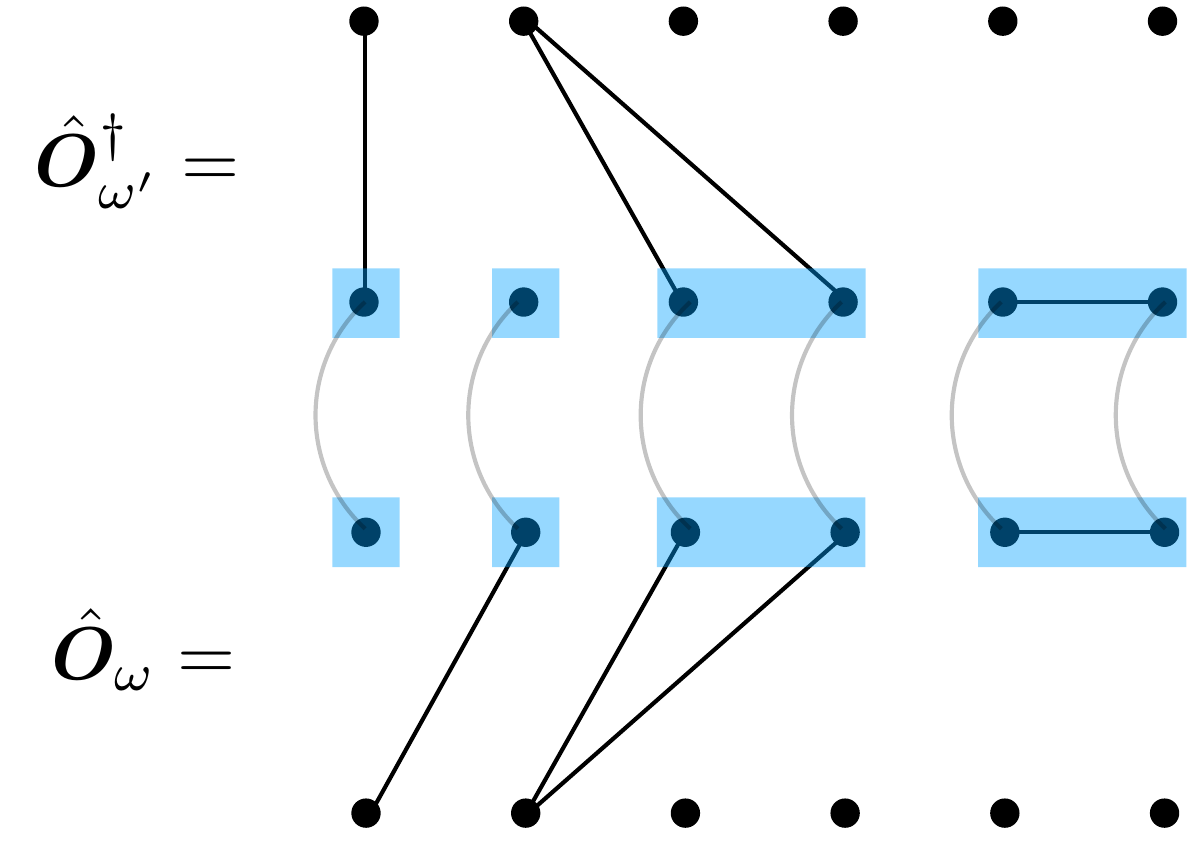}
    \caption{Multiplying two noncrossing $m$-factors $\hat{\vO}_{\omega'}^{\dagger}\cdot\hat{\vO}_{\omega}$. Due to the distinctness properties~\eqref{eq:hatvO}, the blocks in blue must match exactly, which holds in the above. However, since one propagating block of $\hat{\vO}_{\omega}$ is contracted with a nonpropagating block of $\hat{\vO}_{\omega'}$, the number of propagating blocks decreases. Thus, the above is an element of $\vJ_{1}$, and vanishes due to taking quotient.}
    \label{fig:OwOw}
\end{figure}

\begin{proof}[Proof of~\autoref{lem:multiplicative_Oomega}] 

\textbf{Step 1: Distinctness implies the set of upper blocks matches.} Recall that the $\hat{O}_\omega$ were constructed specifically so that the outgoing wires of $\hat{\vO}_\omega$ were distinct. As illustrated in \autoref{fig:OwOw}, this implies that the upper block of $\omega'$ must be exactly the same partition as that of $\omega$. Otherwise, two different blocks will be contracted to give zero. 

\textbf{Step 2: Fixing the propagating blocks.} To establish that the product is actually proportional to $\delta_{\omega'\omega}$, we need to use the properties of $\vJ_{m-1}$ and the fact that $\omega\in NF_m$ is noncrossing. Since the quotient $\vJ_{m-1}$ kills any diagram with fewer than $m$ propagating blocks, the product $\hat{\vO}_{\omega'}^{\dagger}\hat{\vO}_{\omega}$ vanishes unless the number of blocks remains $m$, which requires that the set of upper propagating blocks of $\omega'$ exactly matches the set of upper propagating blocks of $\omega$. Furthermore, since $\omega'$ is a noncrossing $m$-factor, this also uniquely fixes how the propagating blocks of $\omega'$ are wired to the lower leftmost  $m$ nodes (according to the rightmost element of each block). 

\textbf{Step 3: Calculating the constant.} Having established orthogonality, we can obtain the final form
\begin{align}
\hat{\vO}_{\omega}^{\dagger}\hat{\vO}_{\omega} &\propto \sum^N_{\text{distinct } i_1, \cdots, i_m =1 }\ket{i_1}\bra{i_1}\otimes \ket{i_2}\bra{i_2}\cdots \ket{i_m}\bra{i_m} \otimes \text{($2k-2m$ isolated nodes)} \\
    &\propto \vI_m \otimes \text{($2k-2m$ isolated nodes)} \quad \text{(mod $\vJ_{m-1}$)}\\
    &\propto \vec{e}_m \quad \text{(mod $\vJ_{m-1}$)}
\end{align}
The second line ``completes'' the identity by adding suitable elements in $\vJ_{m-1}$. Lastly, the constant is determined by contracting the upper \textit{nonpropagating} blocks, which gives the falling factorial in $c(\omega)$, and the normalization factor ($N^{k-m}$) of $\vec{e}_m$. 
\end{proof}

\begin{proof}[Proof of~\autoref{lem:orthonormal}]
We begin by verifying the algebra relation for matrix elements:
\begin{align}
    \vE_{ij}\vE_{k\ell} = \delta_{jk}\vE_{i\ell},\quad \vE^{\dagger}_{ij} = \vE_{ji}.
\end{align}
Indeed, we have that
\begin{align}
    \CQ_{\vJ}[\vec{v}_i\vec{v}^{\dagger}_j] \CQ_{\vJ}[\vec{v}_k\vec{v}^{\dagger}_{\ell}] &= \CQ_{\vJ}[\vec{v}_i\vec{v}^{\dagger}_j \vec{v}_k\vec{v}^{\dagger}_{\ell}]\tag*{(Since $\CQ_{\vJ}$ is a $\dagger$-homomorphism)}\\
    & = \frac{1}{n(\omega_1,t_1)n(\omega_2,t_2)n(\omega_3,t_3)n(\omega_4,t_4)}\CQ_{\vJ}[\hat{\vO}_{\omega_1}\vu_{t_1}\vec{p}_{\vec{\lambda^*}}\vec{p}_{\vec{\lambda^*}}^{\dagger}\vu^{\dagger}_{t_2}\hat{\vO}_{\omega_2}^{\dagger}\hat{\vO}_{\omega_3}\vu_{t_3}\vec{p}_{\vec{\lambda^*}}\vec{p}_{\vec{\lambda^*}}^{\dagger}\vu^{\dagger}_{t_4}\hat{\vO}_{\omega_4}^{\dagger}]\tag*{($i,j,k,\ell\rightarrow \omega_1,t_1,\ldots,\omega_4, t_4$)}\\
    & = \frac{c(\omega_2)}{n(\omega_1,t_1)n(\omega_2,t_2)n(\omega_3,t_3)n(\omega_4,t_4)} \CQ_{\vJ}[\hat{\vO}_{\omega_1}\vec{e}_m\vu_{t_1}\vec{p}_{\vec{\lambda^*}}\vec{p}_{\vec{\lambda^*}}^{\dagger}\vu^{\dagger}_{t_2}\vu_{t_3}\vec{p}_{\vec{\lambda^*}}\vec{p}_{\vec{\lambda^*}}^{\dagger}\vu^{\dagger}_{t_4}\hat{\vO}_{\omega_4}^{\dagger}]\tag*{(By~\autoref{lem:multiplicative_Oomega})}\\
    & = \frac{c(\omega_2)}{n(\omega_1,t_1)n(\omega_2,t_2)n(\omega_3,t_3)n(\omega_4,t_4)} \CQ_{\vJ}[\hat{\vO}_{\omega_1}\vec{e}_m\vu_{t_1}\vec{p}_{\vec{\lambda^*}}\vec{p}_{\vec{\lambda^*}}^{\dagger}\vu^{\dagger}_{t_4}\hat{\vO}_{\omega_4}^{\dagger}]\tag*{(By~\autoref{lem:ONbasis_Sm})}\\
    & = \frac{c(\omega_2)}{n(\omega_1,t_1)n(\omega_2,t_2)n(\omega_3,t_3)n(\omega_4,t_4)} \CQ_{\vJ}[\hat{\vO}_{\omega_1}\vu_{t_1}\vec{p}_{\vec{\lambda^*}}\vec{p}_{\vec{\lambda^*}}^{\dagger}\vu^{\dagger}_{t_4}\hat{\vO}_{\omega_4}^{\dagger}]\tag*{(By definition of $\vec{e}_m$)}\\
    & = \delta_{jk} \CQ_{\vJ}[\vec{v}_i\vec{v}^{\dagger}_{\ell}], 
\end{align}
and that
\begin{align}
     \L( \CQ_{\vJ}[\vec{v}_i\vec{v}^{\dagger}_j] \R)^{\dagger} = \CQ_{\vJ}[\vec{v}_j\vec{v}^{\dagger}_i] \tag*{(Since $\CQ_{\vJ}$ is an $\dagger$-homomorphism)},
\end{align}
as required.

\end{proof}

\subsubsection{Projector onto irreducible representations as a sum over diagrams}  
We have constructed an explicit orthonormal basis but the abstract quotient projector $\CQ_{\vJ_m}$ that is required to ensure orthogonality remains somewhat mysterious.
\begin{align}
    \CQ_{\vJ_{m-1}}[\cdot] = \sum_{\labs{\vec{\lambda^*}}\ge m } \vec{\Pi}_{\vec{\lambda}}\cdot\vec{\Pi}_{\vec{\lambda}}\label{eq:Q_jm_writing}.
\end{align}  
This section further gives a \textit{low-degree} expansion of $\vec{\Pi}_{\vec{\lambda}}$ in terms of the diagram operator $\vO_{\Pi}$.

\begin{lem}[A recursive formula for the projector onto irreducible representations]\label{lem:degree_projector}
For each $\vec{\lambda}\in \Lambda_{k,N}$, 
\begin{align}
    \vec{\Pi}_{\vec{\lambda}} = \sum_{i} \L( \vec{v^{\lambda}}_i\vec{v}^{\vec{\lambda}\dagger}_i - \sum_{\labs{\vec{\mu}^*}\le \labs{\vec{\lambda}^*}-1 } \vec{\Pi}_{\vec{\mu}}\vec{v^{\lambda}}_i\vec{v}^{\vec{\lambda}\dagger}_i \R).
\end{align}
Thus, solving the recursion yields
\begin{align}
    \vec{\Pi}_{\vec{\lambda}} = \sum_{\Pi} c(\Pi,\vec{\lambda}) \vO_{\Pi}
\end{align}
where $c(\Pi,\vec{\lambda})$ is a rational function of $N$ with poles at integers $0,\ldots, k$, each with multiplicity at most $k$.
\end{lem}

\begin{proof}[Proof of~\autoref{lem:degree_projector}]
\begin{align}
    \vec{\Pi}_{\vec{\lambda}}
    & = \sum_{i} \CQ_{\vJ_{m-1}}[ \vec{v^{\lambda}}_i\vec{v}^{\vec{\lambda}\dagger}_i ]\tag*{(Resolution of identity)}\\
    & = \sum_{i} \L( \vec{v^{\lambda}}_i\vec{v}^{\vec{\lambda}\dagger}_i - \sum_{\labs{\vec{\mu}^*}\le \labs{\vec{\lambda}^*}-1 }   \vec{\Pi}_{\vec{\mu}}\vec{v^{\lambda}}_i\vec{v}^{\vec{\lambda}\dagger}_i\vec{\Pi}_{\vec{\mu}} \R)\tag*{(Rewriting~\eqref{eq:Q_jm_writing})}\\
    & = \sum_{i} \L( \vec{v^{\lambda}}_i\vec{v}^{\vec{\lambda}\dagger}_i - \sum_{\labs{\vec{\mu}^*}\le \labs{\vec{\lambda}^*}-1 }   \vec{\Pi}_{\vec{\mu}}\vec{v^{\lambda}}_i\vec{v}^{\vec{\lambda}\dagger}_i \R)\tag*{(Since $\vec{v^{\lambda}}_i\vec{v}^{\vec{\lambda}\dagger}_i\in\vP_k(N)$)}.
\end{align} 
The third line drops one of the projectors since $\vec{\Pi}_{\vec{\mu}}$ are projectors for a factor. Solving the recursion implies the projector can be expressed as a formal noncommutative polynomial of the basis representatives
\begin{align}
    \vec{\Pi}_{\vec{\lambda}} = \poly\L(\{\sum_i \vec{v}^{\vec{\lambda}}_i\vec{v}^{\vec{\lambda}\dagger}_i\}_{\vec{\lambda}\in \Lambda_{k,N}}\R) \quad \text{of degree}\quad \labs{\vec{\lambda}}.
\end{align}
Further rewriting the basis representatives 
\begin{align}
    \vec{v^{\lambda}}_i\vec{v}^{\vec{\lambda}\dagger}_i &= \frac{1}{n(\omega)^2} \hat{\vO}_{\omega}\vu_{t}\vec{p}_{\vec{\lambda^*}}\vec{p}_{\vec{\lambda^*}}^{\dagger}\vu^{\dagger}_{t}\hat{\vO}_{\omega}^{\dagger}\\
    &= \sum_{\Pi} c(\Pi,\vec{\lambda}) \vO_{\Pi}
\end{align}
where each $c(\Pi,\vec{\lambda})$ is a rational function of $N$ with pole at integers $0,\ldots, k$ with multiplicity at most 1 by \autoref{lem:orthonormal}. Thus,
\begin{align}
     \vec{\Pi}_{\vec{\lambda}} = \sum_{\Pi} c(\Pi,\vec{\lambda}) \vO_{\Pi}
\end{align}
with poles at integers $0,\cdots, k$ each with multiplicity at most $k$. We will not ever need to worry about the numerator degree, as it will already be controlled by the denominator degree for the ultimate quantity of interest (bounded by $\CO( N^0) $). 
\end{proof}

\subsection{Extending test operators}
\label{sec:extending_test_ops}
However, we are still faced with the challenge that the distinguishing probability also depends on the quantum algorithm,  which amounts to all possible input states and observables in the parallel query model. To make the distinguishing probability a polynomial, we must specify how we \textit{extend} the quantum algorithm to other dimensions $N'$. In most circumstances, that would be very unnatural but here, the fact that the partition algebra is \textit{independent} of $N$ for $N \ge 2k$ allows us to define a family of ``equivalent'' quantum algorithms across larger and smaller dimensions $N' \ne N$. Doing so amounts to choosing the appropriate input state and observable. 

First, we need to reduce the size of the ancilla to control the diamond distances. 
\begin{lem}[Variational expression for diamond distance for (left and right) permutation invariant channels]\label{lem:small_ancilla}
    For any two channels $\CN_1, \CN_2$, 
    \begin{align}
        \lnormp{\CN_{2k,\vS}\circ\bigg( \CN_{1} -\CN_{2} \bigg) \circ \CN_{2k,\vS}}{\diamond}  = \sup_{\vrho,\vO} \tr\L[\vO (\CN_1-\CN_2 )\L[\vrho\R]\R]
    \end{align}
    where the normalized state $\vrho$ and operator $\vO$ are supported on bipartite Hilbert space with dimension $N^k\times \max_{\vec{\lambda}\in \Lambda_{k,N}} d_{\lambda}$, where $d_{\lambda}$ are the dimensions of the factors in the partition algebra.
\end{lem}
In particular, assuming the local dimension is large, specifically $2k \le N$, the irrep dimension $d_{\vec{\lambda}}$ will only depend on $k$ and is independent of $N$. 
\begin{proof}
For any input $\vrho$ that may be entangled with an $N^k$-dimensional ancilla, the permutation average enforces the structure
\begin{align}
    \CN_{2k,\vS}[\vrho] &= \bigoplus_{\vec{\lambda}} \vrho'_{\vec{\lambda}}\otimes \vec{\tau}_{\vec{\lambda}}.
\end{align}
where $\vrho'_{\vec{\lambda}}$ acts on the factors of the partition algebra and the ancilla. Further, since the trace distance is convex, the maximum will be attained at a pure input on one of the factors $\vrho_{\vec{\lambda}}$. For such a state, there is a further unitary $\vU_{anc}$ acting on the ancilla that ``compresses'' the entanglement (i.e., by Schur decomposition)
\begin{align}
\vrho'_{\vec{\lambda}} = \vU_{anc}\L(\ket{0}\bra{0} \otimes \vrho_{\vec{\lambda}}\R)\vU_{anc}^{\dagger}
\end{align}
where $\vrho_{\vec{\lambda}}$ is a pure state entangled between 
a $d_{\vec{\lambda}}$-dimensional subspace $\CH_{anc}$ and the $\vec{\lambda}$ factor in the partition algebra. Further, the distinguishing operator $\vO$ achieving maximal distinguishing probability will also be supported on the subspace $\CH_{anc} \otimes (\BC^{N})^{\otimes k}.$ Thus, the maximal distinguishing probability can be attained using an ancilla dimension as little as the largest dimension of the factors of the partition algebra $\max_{\vec{\lambda}\in \Lambda_{k,N}} d_{\lambda}$. 
\end{proof}
Second, since the input is ``fixed,'' and is independent of the dimension, we can define a family of ``equivalent'' states across dimensions.

\begin{lem}[Extending given inputs and observable to other dimensions]\label{lem:extensions}
Consider $\CN_{ \vec{j}, \{\vW_i\}}$ as in~\autoref{lem:freeness_random_P}. Then, assuming $N \ge 2k$, for any observable and state $\vA_{N},\vrho_{N} \in \vB((\mathbb{C}^N)^{\otimes k})$, there exist matrices $\vA_{N'},\vrho_{N'} \in \vB((\mathbb{C}^{N'})^{\otimes k})$ such that function 
\begin{align}
   f \L(\frac{1}{N'}\R) := \tr\L[ \vA_{N'}\CN_{ \vec{j}, \{\vW_i\}}[\vrho_{N'}]\R] \quad \text{for each integer} \quad N' \ge 2k
\end{align}
is rational polynomial of $N$ with poles at integers $0,\ldots, 2k\ell-1$ each with multiplicity at most $2k\ell+3k+2$.
\end{lem}
\begin{proof}[Proof of~\autoref{lem:extensions}]
For ease of notation, we begin with the case without ancillas, and then argue that the same argument works in the presence of ancillas. Indeed, what is changing is the basis for the partition algebra and the ancilla is fixed. Given any operator and state $\vA_{N},\vrho_{N} \in \vP_k(N) \subset \vB((\mathbb{C}^N)^{\otimes k})$, they can be written as
\begin{align}
    \vrho_{N} &=\bigoplus_{\vec{\lambda} \in \Lambda_{k,N}} \sum_{i,j} \rho_{ij}^{\vec{\lambda}}\ket{v_i}\bra{v_j}\otimes \frac{\vI_{\vS_{\vec{\lambda}}}}{\tr[\vI_{\vS_{\vec{\lambda}}}]}\\
    &=\sum_{\vec{\lambda} \in \Lambda_{k,N}} \frac{1}{\tr[\vI_{\vS_{\vec{\lambda}}}]} \sum_{i,j} \CQ_{\vec{\lambda}}[\vec{v}_i\vec{v}^{\dagger}_j] \rho_{ij}^{\vec{\lambda}} \\
    \vA_{N}  &=\bigoplus_{\vec{\lambda} \in \Lambda_{k,N}} \sum_{i,j} A_{ij}^{\vec{\lambda}}\ket{v_i}\bra{v_j}\otimes \vI_{\vS_{\vec{\lambda}}}\\
    &=\sum_{\vec{\lambda} \in \Lambda_{k,N}} \sum_{i,j} \CQ_{\vec{\lambda}}[\vec{v}_i\vec{v}^{\dagger}_j] A_{ij}^{\vec{\lambda}}
\end{align}
such that
\begin{align}
    \tr[\vrho_N] = \sum_{\lambda\in \Lambda_{k,N}}\sum_{i} \vrho_{ii}^{\vec{\lambda}} = 1 \quad \text{and}\quad \norm{\vA_N}_{\infty} = \max_{\lambda} \norm{\vA^{\vec{\lambda}}}_{\infty}
\end{align}
and for each $\vO_{\Pi}\in \vP_k(N)$, 
\begin{align}
    \tr[\vO_{\Pi} \vrho_{N}]&= \sum_{\vec{\lambda}\in \Lambda_{k,N}}\sum_{i,j} \frac{\tr[\vO_{\Pi}\CQ_{\vec{\lambda}}[\vec{v}_i\vec{v}^{\dagger}_j]]}{\tr[\vI_{\vS_{\vec{\lambda}}}]} \rho_{ji}^{\vec{\lambda}}\\
    \tr[\vO_{\Pi} \vA_{N}]&= \sum_{\vec{\lambda}\in \Lambda_{k,N}} \sum_{i,j} \tr[\vO_{\Pi}\CQ_{\vec{\lambda}}[\vec{v}_i\vec{v}^{\dagger}_j]]A_{ji}^{\vec{\lambda}}.
\end{align}
The above can be extended to other dimensions $N' \ge 2k$ as $\vA_{N'},\vrho_{N'} \in \vB((\mathbb{C}^{N'})^{\otimes k})$, which is
\begin{align}
    \vrho_{N'} &=\sum_{\vec{\lambda}' \in \Lambda_{k,N'}} \frac{1}{\tr[\vI_{\vS_{\vec{\lambda}'}}]} \sum_{i,j} \CQ_{\vec{\lambda}'}[\vec{v}_i\vec{v}^{\dagger}_j] \rho_{ij}^{\vec{\lambda}} \\
    \vA_{N'} &=\sum_{\vec{\lambda}' \in \Lambda_{k,N'}} \sum_{i,j} \CQ_{\vec{\lambda}'}[\vec{v}_i\vec{v}^{\dagger}_j] A_{ij}^{\vec{\lambda}'}.
\end{align}
Note that
\begin{align}
    \tr[\vrho_{N'}] = \sum_{\lambda' \in \Lambda_{k,N'} }\sum_{i} \vrho_{ii}^{\vec{\lambda}'} = 1 \quad \text{and}\quad \norm{\vA_{N'}}_{\infty} = \max_{\lambda'\in \Lambda_{k,N'}} \norm{\vA^{\vec{\lambda}'}}_{\infty}.
\end{align}
Also, we relate the extension $ N'\ge 2k $ from $N$ by setting
\begin{align}
    \rho_{ij}^{\vec{\lambda'}} = \rho_{ij}^{\vec{\lambda}}\quad \text{and}\quad A_{ij}^{\vec{\lambda'}} = A_{ij}^{\vec{\lambda}} \quad \text{for each}\quad i,j,\quad \text{if}\quad \labs{\vec{\lambda}^{'*}}=\labs{\vec{\lambda}^*}
\end{align}
since the irreducible representations $\vec{\lambda}$ of $\vP_k(N)$ can be identified with $\vec{\lambda}'$ of $\vP_k(N')$, and similarly for $A_{ij}^{\vec{\lambda}}$. We evaluate the expression by inserting resolutions of identity
\begin{align}
    \tr\L[ \vA_{N'}\CN_{ \vec{j}, \{\vW_i\}}[\vrho_{N'}]\R] &= (\vA_{N'}|\vec{\CN}_{ \vec{j}, \{\vW_i\}}|\vrho_{N'})\\
    &= \sum_{\Pi_1,\Pi_2\vdash 2k } \frac{(\vA_{N'}|\vO'_{\Pi_2})(\vO'_{\Pi_2}|\vec{\CN}_{ \vec{j}, \{\vW_i\}}|\vO'_{\Pi_1})(\vO'_{\Pi_1}|\vrho_{N'})}{\norm{\vO'_{\Pi_2}}^2_2\norm{\vO'_{\Pi_1}}^2_2}\\
    & = \sum_{\Pi_1,\Pi_2\vdash 2k } \frac{(\vA_{N'}|\vO'_{\Pi_2})(\vO'_{\Pi_2}|\vec{\CN}_{ \vec{j}, \{\vW_i\}}|\vO'_{\Pi_1})(\vO'_{\Pi_1}|\vrho_{N'})}{\norm{\vO'_{\Pi_2}}^2_2\norm{\vO'_{\Pi_1}}^2_2}.
\end{align}
Let's count the poles for each:

(i) By~\eqref{eq:norm_O'pi},
\begin{align}
    \frac{1}{\norm{\vO'_{\Pi_2}}^2_2}, \frac{1}{\norm{\vO'_{\Pi_2}}^2_2} = \frac{1}{N(N-1)\cdots}
\end{align}
which individually has poles at integers $0,\cdots, k$ each with multiplicity at most $1$.

(ii) Rewrite in terms of the $\vO_{\Pi}$'s by~\eqref{eq:Kinv}
\begin{align}
    (\vO'_{\Pi_2}|\vec{\CN}|\vO'_{\Pi_1})&= \sum_{\Pi'_1,\Pi'_2\vdash 2k  } K^{-1*}_{\Pi_2\Pi'_2} K^{-1}_{\Pi_1\Pi'_1}(\vO_{\Pi'_2}|\vec{\CN}|\vO_{\Pi'_1}),
\end{align}
which, by~\autoref{lem:ONO_polyN}, which has poles at integers $0,\ldots, 2k\ell-1$ each with multiplicity at most $2k\ell$.

(iii)
\begin{align}
    (\vO'_{\Pi_1}|\vrho_{N'}) &= \sum_{\Pi'_1\vdash 2k } K^{-1*}_{\Pi_1\Pi'_1} (\vO_{\Pi'_1}|\vrho_{N'}) \tag{By~\eqref{eq:Kinv}}\\
    & = \sum_{\Pi'_1\vdash 2k }\sum_{i,j,\vec{\lambda}}  K^{-1*}_{\Pi_1\Pi'_1} \frac{\tr[\vO_{\Pi'_1}\CQ_{\vec{\lambda}}[\vec{v}_i\vec{v}^{\dagger}_j]]}{\tr[\vI_{\vS_{\vec{\lambda}}}]}\rho_{ji}^{\vec{\lambda}}. 
\end{align}

The poles came from $\CQ_{\vec{\lambda}} = \sum_{\vec{\lambda}} \Pi_{\vec{\lambda}}$ and $\vec{v}_i\vec{v}_j^{\dagger}$ (\autoref{lem:degree_projector}) and $\tr[\vI_{\vS_{\vec{\lambda}}}]$ (\autoref{lem:dimension_S_lambda}), which has poles at integers $0,\ldots, 2k-1$ each with multiplicity at most $2k$.

(iv)
\begin{align}
    (\vO'_{\Pi_1}|\vA_{N'}) &= \sum_{\Pi'_1\vdash 2k } K^{-1*}_{\Pi_1\Pi'_1} (\vO_{\Pi'_1}|\vA_{N'}) \tag{By~\eqref{eq:Kinv}}\\
    &= \sum_{\Pi'_1\vdash 2k }\sum_{i,j,\vec{\lambda}}  K^{-1*}_{\Pi_1\Pi'_1} \tr[\vO_{\Pi'_1}\CQ_{\vec{\lambda}}[\vec{v}_i\vec{v}^{\dagger}_j]] A_{ji}^{\vec{\lambda}}
\end{align}
which like the case of (iii), has poles at integers $0,\ldots, 2k\ell-1$ each with multiplicity at most $2k\ell+k$.

Altogether, 
\begin{align}
f \L(\frac{1}{N'}\R) := \tr\L[ \vA_{N'}\CN_{ \vec{j}, \{\vW_i\}}[\vrho_{N'}]\R] \quad \text{for each integer} \quad N' \ge 2k, 
\end{align}
which has poles at integers $0,\ldots, 2k\ell-1$ each with multiplicity at most $2k\ell+3k+2$.

To extend to the case with ancillas, we run the above argument for each irrep $\vec{\lambda} \in \Lambda_{k,N}$
\begin{align}
    \vrho_{\vec{\lambda}} &= \sum_{i,j} \rho_{ij}^{\vec{\lambda}}\ket{v_i}\bra{v_j}\otimes \frac{\vI_{\vS_{\vec{\lambda}}}}{\tr[\vI_{\vS_{\vec{\lambda}}}]}\otimes \vO_{ij}\\
     \vA &=\bigoplus_{\vec{\lambda} \in \Lambda_{k,N}} \sum_{i,j} A_{ij}^{\vec{\lambda}}\ket{v_i}\bra{v_j}\otimes \vI_{\vS_{\vec{\lambda}}} \otimes \vO'_{ij}
\end{align}
were $\vO_{ij}, \vO'_{ij}$ acts on space of dimension $d_{\vec{\lambda}}$, which encompasses all possible operator $\vA$ acting on $(\BC^{N})^{\otimes k}\otimes \BC^{d_{\vec{\lambda}}}$ and inheriting the operator norm bounds.
\end{proof}
\begin{lem}[Rational polynomial for Haar channel]\label{lem:poly_Haar_rhoO}
In the setting of~\autoref{lem:extensions}, the Haar random channel evaluates on the same set of extensions $\vrho_{N'}, \vA_{N'}$
    \begin{align}
   f \L(\frac{1}{N'}\R) := \tr\L[ \vA_{N'}\CN_{ 2k, \vU}[\vrho_{N'}]\R] \quad \text{for each integer} \quad N' \ge 2k
\end{align}
is rational polynomial of $N$ at integers $-(k-1), \cdots, 0,1,2k -1$, each with multiplicity bounded by $4k+2$.
\end{lem}
\begin{proof}
The change in the argument of~\autoref{lem:extensions} is that
\begin{align}
    (\vO'_{\Pi_2}|\vec{\CN}|\vO'_{\Pi_1})&= \sum_{\Pi'_1,\Pi'_2\vdash 2k  } K^{-1*}_{\Pi_2\Pi'_2} K^{-1}_{\Pi_1\Pi'_1}(\vO_{\Pi'_2}|\vec{\CN}|\vO_{\Pi'_1})
\end{align}
And we directly invoke the Weingarten function for the unitary group~\cite[Theorem 4.3]{collins2022weingarten}, which has poles at $-(k-1),\cdots,0,\cdots (k-1)$ with multiplicity at most $k$.
\end{proof}
\section{Proof of main result}\label{sec:proof_main}
This section contains the formal proof of our main result (\autoref{thm:parallelPRU}), which we restate as follows.

\parallelPRU*

Since the proof requires several lemmas, we first discuss them locally and combine them in the end.
\subsection{Expanding one exponential into words}
We first expand the single exponential and collect the distinct words
\begin{align}
    \e^{\ri \theta\vA_m} &= \exp \L(\frac{\ri\theta}{\sqrt{2m}}\sum_{a=1}^m (\vZ_a+\vZ^{\dagger}_a) \R)\\
    &= \vI + \frac{\ri\theta}{\sqrt{2m}}\sum_{a=1}^m \L(\vZ_a+\vZ^{\dagger}_a\R) - \frac{\theta^2}{2}\frac{1}{2m} \L(\sum_{a=1}^m \vZ_a+\vZ^{\dagger}_a\R)^2 \cdots \\
    & = v(\theta) \vI + \sum_{i} w_i(m,\theta) \vW_i \quad \text{where}\quad \vW_i \in \CW(\{\vZ_a\}_{a=1}^m).
\end{align}

\begin{lem}[Normalized, bounded weights]\label{lem:property_weights}
\begin{align}
    \sum_i \labs{w_i}^2 &= 1\\
    \max_i{\labs{w_i}} &\le \frac{1}{\sqrt{m}}.
\end{align}
\end{lem}
\begin{proof}
Fix $m$ and taking the large-$N$ limit: 
\begin{align}
    1= \lim_{N\rightarrow \infty}\BE\btr[\e^{-\ri \theta\vA_m}\e^{\ri \theta\vA_m}]= \sum_i \labs{w_i}^2\tag{large-$N$ limit}.
\end{align}

The second inequality uses the fact that the normalization trace of a nontrivial word vanishes in the large-$N$ limit. The second claim follows from the symmetry of the paths: each path with nonzero length must map to $m$ other paths by cycling through the labels ($a =1,\cdots, m$). Thus, there cannot be a path contributing more than $1/\sqrt{m}$ due to the sum constraints $\sum_i \labs{w_i}^2 =1$.
\end{proof}

The coefficient $v(\theta)$ for the identity operator is a sum over weighted trivial words (i.e., loops on the Cayley graph of the free group of $m$ elements). For example, 
\begin{align}
     \frac{1}{2m} \L(\sum_{a=1}^m \vZ_a+\vZ^{\dagger}_a\R)^2 &= 1\cdot \vI + (\text{nontrival words})\\
     \frac{1}{4m^2} \L(\sum_{a=1}^m \vZ_a+\vZ^{\dagger}_a\R)^2 &= (1+ \frac{m-1}{2m})\cdot \vI + (\text{nontrival words}).
\end{align}
The sum over all paths at order $p$ is exactly the moments of Kesten-McKay distribution~\cite{Kesten_1959,McKay1981TheEE}. Thus, the exponentially weighted sum is then the characteristic function
\begin{align}
    v_m(\theta) = \int_{-\infty}^{\infty} \e^{\ri \theta x} p_m(x) \rd x.
\end{align}

\begin{defn}[Kesten-McKay distribution~\cite{Kesten_1959,McKay1981TheEE}]
    The Kesten-McKay distribution is defined by
    \begin{align}
        p_m(x) := \begin{cases} \frac{\sqrt{4\frac{m-1}{m}-x^2}}{2\pi(1-\frac{x^2}{m})}\quad &\text{if}\quad \labs{x} \le 2\sqrt{1-\frac{1}{m}}\\
        0 &\text{else}.
        \end{cases}
    \end{align}
\end{defn}
In the limit $m\rightarrow \infty$, the Kesten-McKay distribution converges to Wigner's semi-circle
\begin{align}
        p_m(x) := \begin{cases} \frac{\sqrt{4-x^2}}{2\pi}\quad &\text{if}\quad \labs{x} \le 2\\
        0 &\text{else},
        \end{cases}
\end{align}
and the above gives an exact formula for the correction due to $m$. Thus, we may eliminate the trivial word $\vI$ by choosing the zeros of this function:
\begin{align}
    \theta_m \quad \text{such that}\quad v_m(\theta_m) =0,
\end{align}
i.e., making the exponential traceless in the large-$N$ limit. In particular, the angle is always bounded by an absolute constant $\theta_m = \CO(1)$ for any $m$.

The nice consequence of choosing $v_m(\theta_m)=0$ is that there is \textit{no cancellation} due to taking the product of exponentials. For example, in the product of two independent copies
\begin{align}
\e^{\ri \theta\vA_m}\e^{\ri \theta\vA'_m}&= \sum_{i} w_i(m,\theta) \vW_i  \cdot  \sum_{i} w_i(m,\theta) \vW'_i\\
& =\sum_{i,j} w_i(m,\theta) w_j(m,\theta) \vW_i\vW'_j,
\end{align}
each of the product words is different so that
\begin{align}
\vW_i\vW'_j = \vW_{i'}\vW'_{j'} \quad \text{if and only if}\quad (i,j) = (i',j').
\end{align}
Thus, taking products yields an exponentially growing number of words that we can easily keep track of. Note that taking powers of the same $(\e^{\ri \theta_m \vA_{m}})^2$ could also work in principle with a more careful combinatorial analysis for the cancellation of words. 

\subsection{Proof of~\autoref{thm:parallelPRU}}
We begin with a proof sketch. For any integers $m,\ell$, the product of $\ell$ independent exponentials of $2m$-sparse random matrices can be expanded into distinct words
\begin{align}
        \prod_{j=1}^{\ell} \e^{\ri \theta\vA^{(j)}_m} &= \sum_{\vec{i}} w_{\vec{i}}(\theta)\vW_{\vec{i}}\tag*{(Distinct words: $\vW_{\vec{i}}\ne \vW_{\vec{j}}$ for each $\vec{i}\ne \vec{j}$.)}
\end{align}
with fairly spread-out weights $w_{\vec{i}}$ (calculated from the individual properties (\autoref{lem:property_weights}))
\begin{align}
    \sum_{\vec{i}}\labs{w_{\vec{i}}}^2 &= (\sum_{i}\labs{w_{i}}^2)^\ell = 1\notag\\
    \max_{\vec{i}} \labs{w_{\vec{i}}} &= \L(\max_{i} \labs{w_{i}}\R)^{\ell} \le \frac{1}{\sqrt{m}^{\ell}}\label{eq:w_vec_i}.
\end{align}

Intuitively, we expect the resulting product to be a very dense matrix (sparsity $\Omega(m)^{\ell}$) as we expect the distinct words of permutations to map the same bit string to distinct bit strings. 

\subsubsection{\texorpdfstring{The large-$N$ limit}{The large-N limit}}
Indeed, if the dimension $N$ is infinite (fixing the length of the words and the number of queries $k$), then there will almost surely be no \textit{collision} between the distinct words, and we expect them can be replaced with \textit{independent} permutations. This can be captured by the following equality of limits, with the precise order of limit and choice of norm stated in $\autoref{lem:freeness_random_P}$. 
\begin{align}
    \vZ_L\L(\sum_{\vec{i}} w'_{\vec{i}}(\theta)\vW'_{\vec{i}}\R) \vZ_R&\stackrel{N\rightarrow \infty}{=} \vZ_L\L(\sum_{\vec{i}} w'_{\vec{i}}(\theta)\vZ_{\vec{i}}\R)\vZ_R\tag*{(distinct words are free in the large-$N$ limit:\autoref{lem:freeness_random_P})}\\
    &\stackrel{m^{\ell}\rightarrow\infty}{=} \vG. \tag*{(Matrix central limit theorem in $k$-fold diamond norm)}
\end{align}
Observe that the summands match the second (mixed) moments of a Gaussian matrix
\begin{align}
    \BE[ \vZ^{(\dagger)}_i \otimes \vZ^{(\dagger)}_i] = \BE[ \vG^{(\dagger)}_i \otimes \vG^{(\dagger)}_i]. 
\end{align}
Thus, in the limit of many summands (with some requirement that the weights all be reasonably small), the sum converges to a Gaussian; see~\autoref{lem:lindeberg_largeN} for the nonasymptotic bounds.

The Gaussian is a simple random matrix that is left and right unitarily invariant. The Gaussian is nonetheless not the same as a Haar random unitary. In fact, it is not even a unitary matrix, as its singular values are distributed according to the Marchenko-Pastur distribution, which differs from the unitary case. Nevertheless, the Gaussian is effectively Haar in the parallel access model. 
\begin{align}
    \vG \stackrel{N\rightarrow \infty}{=}
    \vU_{Haar}\tag{in the diamond distance, \autoref{lem:Ginibre_is_Haar}}.
\end{align}
Intuitively, for \textit{arbitrary fixed} input state uncorrelated with $\vG$, it is nearly \textit{isometric} with \textit{high probability}, as can be seen from the expected value $\BE\vG^{\dagger}\vG = \vI$. Regardless of the interpretation, the presence of Gaussians $\vG$ is merely a proof artifact, and all that matters in this context is merely the $k$-fold CP map $\BE \CN_{2k,\vG}$.

The assertions above together imply that  
$\vZ_L \L(\prod_{j=1}^{\ell} \e^{\ri \theta\vA^{(j)}_m}\R) \vZ_R \approx \vU_{Haar}$ as $k$-fold channels in the limits of $N\rightarrow \infty, m^{\ell}\rightarrow \infty.$

\subsubsection{\texorpdfstring{Interpolating to finite $N$}{Interpolating to finite N}}
To complete the proof, the main technical argument is to handle finite-$N$ corrections. We isolate the finite-$N$ bounds as the following lemma. The crux of the argument is an interpolation argument for the variable $x = \frac{1}{N}$ from the large-$N$ limits. We begin with a sketch of the argument by invoking the partial results we already proved. 
Define the two channels of interest:
\begin{align}
    \CN_1: = \CN_{2k,\vV}\quad \text{and}\quad \CN_2: = \CN_{2k,\vU_{Haar}}.
\end{align}
By~\autoref{lem:small_ancilla}, the diamond distance can be controlled at ancilla size $\max_{\lambda}d_{\lambda}$
    \begin{align}
        \lnormp{\CN_{2k,\vS}\circ\bigg( \CN_{1} -\CN_{2} \bigg) \circ \CN_{2k,\vS}}{\diamond}  &= \sup_{\vrho,\vO} \tr\L[\vO\CN_1\vrho\R] - \tr\L[\vO\CN_2\vrho\R].
    \end{align}
    where the supremum is over symmetrized operators and states that may involve ancillas. (Even though we have in fact a stronger symmetry $\CN_{2k,\vS}$, it suffices to exploit the $\CN_{2k,\vS}$ symmetry.) We will show that the distinguishing probability is small for each input and output by constructing the function 
    \begin{align}
        f_{\vrho,\vO}(\frac{1}{N'}) &=  \tr\L[\vO^{(N')}\CN^{(N')}_1\vrho^{(N')}\R]-\tr\L[\vO^{(N')}\CN^{(N')}_2\vrho^{(N')}\R]\\
        & = f_1(N') - f_2(N')
    \end{align}
    as an extension from the original dimension $N'=N$ by defining a family of operators and superoperators
    \begin{align}
    &\L(\vO^{(N')},\CN_1^{(N')},\CN_2^{(N')}, \vrho^{(N')} \R) \quad \text{for each }\quad 2k \le  N' \le \infty\\ \text{such that} \quad &\L(\vO^{(N)},\CN_1^{(N)},\CN_2^{(N)}, \vrho^{(N)} \R)= \L(\vO, \CN_1,\CN_2, \vrho\R).
    \end{align}
    Of course, the operators $\vrho^{(N)},\vO^{(N)}$ at other dimensions $N'$ has to be ``similar'' to the original $\vrho, \vO$ at dimension $N$. In order to interpolate using Markov's inequality, we spell out the requirements for the extension:
\begin{itemize}
    \item Rational low-degree function (\autoref{lem:extensions}, \autoref{lem:poly_Haar_rhoO}): we extend to other dimensions by choosing the appropriate extensions of the inputs and outputs. Since the exponential is an infinite series, we have to truncate the Taylor expansion (\autoref{lem:exponential_expand}) to ensure that the product $\prod_{j=1}^{\ell} \e^{\ri \theta_m\vA^{(j)}_m}$ is approximated to error $\epsilon$ in the operator norm. See~\autoref{sec:proof_f_12_trun} for the proofs of the following lemma.
\begin{lem}[Low-degree approximation]\label{lem:f1_truncation}
For any $\epsilon,k$, truncating each exponentials $\e^{\ri \theta_m\vA^{(j)}_m}$ at degree $d_1 = \CO( \sqrt{m}\log(k\ell/\epsilon))$ suffices to ensure 
\begin{align}
    \labs{f_{\vrho,\vO,1}(\frac{1}{N}) - f^{(trun)}_{\vrho,\vO,1}(\frac{1}{N})} \le \epsilon.
\end{align}
\end{lem}

Meanwhile, truncated function $f^{(trun)}_{\vrho,\vO}(\frac{1}{N})$
is a rational polynomial of $N$ at integers 
\begin{align}
0,1,\cdots, 2k\ell d_1 -1\quad\text{each with multiplicity at most}\quad 2k\ell d_1+ 3k + 2
\end{align}
where $\ell \cdot d_1$ is the maximal length of words of permutations, as coming from the product of $\ell$-many degree $d_1$ words.
The case for $f_{\vrho,\vO,2}$ (\autoref{lem:poly_Haar_rhoO}) has fewer multiplicities in a smaller range of integers, and that of $f_{\vrho,\vO,1}$ will dominate.
The choices for the degrees $d_1$ will depend on $N,m,\ell, k$.
    
    \item A priori bounds: This merely came from the structural fact that the expression is a difference between probabilities, bounded between one and zero:
\begin{align}
    \labs{f_{\vrho,\vO}(\frac{1}{N'})} \le \lnorm{\CN_1^{(N')}-\CN^{(N')}_2}_{1-1} \norm{\vrho_{N'}}_1 \norm{\vO_{N'}}_{\infty} \le 2\quad \text{for each integer} \quad N'.\label{eq:apriori}
\end{align}
    \item Large-$N$ limits (\autoref{lem:freeness_random_P}, \autoref{lem:Ginibre_is_Haar}): The large-$N$ limit coincide with the ``free'' model: sum over i.i.d. permutations. Also, the Haar random Channel can be replaced with Ginibre ensembles 
\begin{align}
   f_{1,\vrho,\vO}(\frac{1}{N'})  &= \tr\L[\vO_{N'}\CN_1\vrho_{N'}\R] \stackrel{N'\rightarrow \infty}{=}\tr\L[\vO_{N'}\CN^{(free)}_1\vrho_{N'}\R] = f^{(free)}_{1,\vrho,\vO}(\frac{1}{N'})\notag\\
    f_{2,\vrho,\vO}(\frac{1}{N'}) &=\tr\L[\vO_{N'}\CN_2\vrho_{N'}\R] \stackrel{N'\rightarrow \infty}{=}\tr\L[\vO_{N'}\CN^{(Ginibre)}_2\vrho_{N'}\R] = f^{(Ginibre)}_{2,\vrho,\vO}(\frac{1}{N'})\label{eq:f2_free}.
\end{align}
    \item Comparing the ``nicer'' models (\autoref{lem:lindeberg_largeN}): We further take a Lindeberg principle to argue that these large-$N$ limits are very close to each other. While the most general Lindeberg argument works at any finite dimension $N$, the large-$N$ limits simplify the calculations.
\begin{align}
    \labs{f_{\vrho,\vO}(\frac{1}{\infty})} &= \labs{\lim_{N'\rightarrow \infty } f_{1,\vrho_{N'},\vO_{N'}} - f_{2,\vrho_{N'},\vO_{N'}}}\\
    &= \labs{ \lim_{N'\rightarrow \infty } f^{(free)}_{1,\vrho_{N'},\vO_{N'}} - f^{(Ginibre)}_{2,\vrho_{N'},\vO_{N'}}}\tag{\eqref{eq:f2_free}}\\
    &\le \frac{1}{8} \sum_{\vec{i}} \L( (4k\labs{w_{\vec{i}}})^4 + (4k\labs{w_{\vec{i}}})^{2k}\R) + \frac{1}{8} \sum_{{\vec{i}}} \L( (4k\labs{w'_{\vec{i}}})^4 + (4k\labs{w'_{\vec{i}}})^{2k}\R) \tag{Lindeberg argument for both:\autoref{lem:lindeberg_largeN}}\\
     &\le \frac{(4k)^2}{8} (\sum_{\vec{i}}\labs{w_{\vec{i}}}^2)\cdot  \max_{\vec{i}} \L( (4k\labs{w_{\vec{i}}})^2 + (4k\labs{w_{\vec{i}}})^{2k-2}\R) + (w_{\vec{i}} \rightarrow w'_{\vec{i}})\tag{Holder's}\\
    &\le \frac{(4k)^2}{8} 1 \cdot  \L( \frac{16k^2}{m^{\ell}} + (\frac{16k^2}{m^{\ell}})^{k-1}\R) + (m,\ell\rightarrow m',\ell')\\ \tag{Bounds on the weights $w_{\vec{i}}$:\eqref{eq:w_vec_i}}\\
    &\stackrel{\ell'=2, m' \rightarrow \infty}{\le} \frac{(4k)^2}{8}\L( \frac{16k^2}{m^{\ell}}+ (\frac{16k^2}{m^{\ell}})^{k-1}\R) \tag{Second term drops due to $m'\rightarrow \infty$}\\
    &\le \min\L( \frac{128k^4}{m^{\ell}},2\R).
\end{align}
The last inequality recalls the a priori bound $\labs{f_{\vrho,\vO}(\frac{1}{\infty})} \le 2$ to drop the higher order term since we must have $\frac{16k^2}{m^{\ell}} \le 1$ for the bound to be better than the a priori bounds.
\end{itemize}

\begin{proof}[Proof of~\autoref{thm:parallelPRU}]
We put the above together. For any $m,\ell, N$,  
\begin{align}
    \lnormp{\CN_{2k,\vS}\circ\bigg( \CN_{1} -\CN_{2} \bigg) \circ \CN_{2k,\vS}}{\diamond} &\le 
    \sup_{\vrho,\vO} \labs{ f_{\vrho,\vO}(\frac{1}{N})}.
\end{align}
We further introduce a truncation error (for the exponential in $f_1$)
\begin{multline}
    \labs{ f_{\vrho,\vO}(\frac{1}{N})} \le \labs{ f_{\vrho,\vO}(\frac{1}{N})-f_{\vrho,\vO}^{(trun)}(\frac{1}{N})} + \labs{ f^{(trun)}_{\vrho,\vO}(\frac{1}{N})-f_{\vrho,\vO}^{(trun)}(\frac{1}{\infty})} \\
    + \labs{ f^{(trun)}_{\vrho,\vO}(\frac{1}{\infty})-f_{\vrho,\vO}(\frac{1}{\infty})} + \labs{ f_{\vrho,\vO}(\frac{1}{\infty})}
\end{multline}
We instantiate the truncation error from $\CN_1$
\begin{align}
    \labs{ f_{\vrho,\vO}(\frac{1}{N})-f_{\vrho,\vO}^{(trun)}(\frac{1}{N})} &\le \epsilon\tag{By~\autoref{lem:f1_truncation}}\quad \text{for any $N$.}
\end{align}

The truncated function $f_1$ now has a total number of poles $2k\ell d_1(2k\ell d_1+3k+2)$ (the real part has a total number of poles $d = 4k\ell d_1(2k\ell d_1+3k+2)$) and bound on the pole locations $B \le 2k\ell d_1 -1$ (which is strictly worse than the Haar random channel (\autoref{lem:poly_Haar_rhoO})), so we can invoke Markov's inequality for $N_0 \ge d^2 +8bB-1$:
\begin{align}
     \labs{\text{Re} f^{(trun)}_{\vrho,\vO}(\frac{1}{N})-\text{Re} f_{\vrho,\vO}^{(trun)}(\frac{1}{\infty})} 
     &\le \frac{4d^2 (N_0 + 10dB)}{N} \sup_{N\ge N_0}\L( \labs{\text{Re} f^{(trun)}_{\vrho,\vO}(\frac{1}{N})} \R)\tag{By~\autoref{lem:large_N_interpolate}}\\
     &\le \frac{8d^2 (N_0 + 10dB)}{N} \sup_{N\ge N_0}\L( \labs{f^{(trun)}_{\vrho,\vO}(\frac{1}{N})} \R)\\
     &\le \CO\L(\frac{(k^2 \ell^2 d_1^2)^4}{N}\R), 
\end{align}
using that $\labs{\text{Re} f^{(trun)}_{\vrho,\vO}(\frac{1}{N_0})}\le \labs{\text{Re} f_{\vrho,\vO}(\frac{1}{N_0})}+\epsilon\le 2+\epsilon \le 3$ and similarly for the imaginary parts. To consolidate the error,
\begin{align}
        \lnormp{\CN_{2k,\vS}\circ\bigg( \CN_{1} -\CN_{2} \bigg) \circ \CN_{2k,\vS}}{\diamond} &\le \epsilon + \CO\L(\frac{k^8\ell^8(m^4\log^{8}(k\ell/\epsilon))}{N}\R) + \min\L( \frac{128k^4}{m^{\ell}},2\R)\\
        &\le \CO\L(\frac{k^8\ell^8m^4n^8}{N}+\frac{k^4}{m^{\ell}}\R)\tag{Setting $\epsilon = \CO(\frac{k\ell}{N})$ }.
\end{align} The last line drops the $\min(\cdot)$ since the quantity is a priori bounded by 2.
\end{proof}

\subsection{Truncation error for exponentials}\label{sec:proof_f_12_trun}
This section includes a standard truncation argument for Taylor expansion of exponentials.
\begin{lem}[Taylor expansion for exponentials]\label{lem:exponential_expand}
For every degree $d$ there exists a polynomial of $h_d$ such that
\begin{align}
    \norm{\e^{\ri \vH} - h_d(\vH)} \le \frac{\norm{\vH}^{d+1}}{(d+1)!}\quad \text{for any Hermitian matrix}\quad \vH. 
\end{align}
In particular, choosing 
    \begin{align}
    d = \CO\L(\frac{\norm{\vH}\log(1/\epsilon)}{\log(\log(1/\epsilon))}\R) \quad \text{ensures that}\quad \norm{\e^{\ri \vH} - h_d(\vH)} \le \epsilon.
    \end{align}
\end{lem}

\begin{proof}[Proof of~\autoref{lem:f1_truncation}]
We begin with Taylor expansion
    \begin{align}
        \norm{ \e^{\ri \theta_m \vA_m} - h_{d_1}(\vA_m)} \le \frac{\norm{\theta_m\vA_m}^{d_1+1}}{(d_1+1)!}.
    \end{align}
    Since the exponential appears $2k\cdot\ell$ times in $f_{\vrho,\vO,1}(\frac{1}{N})$, the error can accounted for by triangle inequality 
    \begin{align}
        \labs{f_{\vrho,\vO,1}(\frac{1}{N}) - f^{(trun)}_{\vrho,\vO,1}(\frac{1}{N})}&\le \L(1 +  \frac{\norm{\theta_m\vA_m}^{d_1+1}}{(d_1+1)!}\R)^{2k\ell} - 1\\
        &\le \exp\L( 2k\ell\frac{\norm{\theta_m\vA_m}^{d_1+1}}{(d_1+1)!}\R) - 1 \tag{Since $1+x\le \e^{x}$}\\
        &\le \epsilon \tag{Since $\theta_m\norm{\vA_m} \le \CO(\sqrt{m})$) and $\e^{x} \le 1+ 2 x$ for $x\le 1$.}.
    \end{align}
    That last line chooses the suitable $d_1 = \CO(\sqrt{m}\log(k\ell/\epsilon))$.
\end{proof}

\appendix
\section{Lindeberg arguments}
One step of our argument converts an independent sum to a matrix with Gaussian entries. This is not terribly surprising and appeared in~\cite{chen2023sparse,chen2024efficient} in a different form. We quickly derive a version suitable for our context.
\begin{lem}[A Lindeberg argument assuming large-$N$]\label{lem:lindeberg_largeN}
For any $\vO = \sum_{j} w_j\vZ_j$ such that $\sum_i \labs{w_i}^2 =1$, In the large-$N$ limit, 
\begin{align}
    \lim_{N\rightarrow \infty}\lnormp{ \CN_{2k,\vO} - \CN_{2k, \vG}}{\diamond} \le \frac{1}{8} \sum_{j=1} \L( (4k\labs{w_j})^4 + (4k\labs{w_j})^{2k}\R).
\end{align}
\end{lem}
This is an adaptation of the matrix Lindeberg argument~\cite{chen2023sparse} with simplification from the large-$N$ limit. The main observation is that the (complex conjugated) second moments are equal for random permutations and Gaussians:
\begin{align}
    \BE[\vZ\otimes \vZ^{\dagger}] &= \BE[\vG\otimes \vG^{\dagger}],
\end{align}
and all other (mixed) moments vanished up to $3$ copies,
\begin{align}
     \BE[\vG^{(\dagger)}] &= \BE[\vZ^{(\dagger)}] = 0\\
     \BE[\vG^{(\dagger)}\otimes \vG^{(\dagger)}\otimes \vG^{(\dagger)}] &= \BE[\vZ^{(\dagger)}\otimes \vZ^{(\dagger)} \otimes \vZ^{(\dagger)}] = 0.
\end{align}
Since we are comparing CP maps, the following polarization identity will be useful to turn any superoperator into a weighted sum over CP maps.
\begin{lem}[Polarization] \label{lem:polarize} For any square matrix $\vA, \vB$,
    $$\vA[\cdot]\vB^{\dagger} = \BE_{s = \pm 1, \pm \ri} s(\vA+s\vB)[\cdot] (\vA+s\vB)^{\dagger}.$$
\end{lem}
\begin{proof}[Proof of~\autoref{lem:lindeberg_largeN}]
Define the interpolant
\begin{align}
    \vO_j := \sum^j_{i=1} \tvA_i +  \sum^m_{j+1} \vA_i \quad \text{where}\quad \vA_j := w_j \vZ_j\quad \text{and}\quad \tvA_j := w_j \vG_j.
\end{align}
    Then, 
\begin{align}
     \lnormp{ \CN_{2k,\vO} - \CN_{2k, \vG}}{\diamond} &= \sup_{\rho,\vO} \tr[\vO(\CN_{2k,\vO} - \CN_{2k, \vG})\vrho].
\end{align}
By \autoref{lem:small_ancilla}, it suffices to take the state to be entangled with ancilla of dimension $\max_{\vec{\lambda}} d_{\vec{\lambda}}$. Nevertheless, this will not be used anywhere in the proof. 

We can expand the difference as a telescoping sum where we exchange one term at a time
\begin{align}
    \tr[\vO(\CN_{2k,\vO} - \CN_{2k, \vG})\vrho] = \sum_{j=1}^{m} (n_j-n_{j-1}) \quad \text{where}\quad
    n_j := \tr[\vO\CN_{2k,\vO_j}[\vrho]].
\end{align}
In order to analyze the difference, for each intermediate step $j$,  denote 
\begin{align}
    \vO_- := \vO_{j}-\tvA_{j} = \vO_{j-1}-\vA_{j}.
\end{align}
When expanding powers of a sum of matrices, keep in mind the scalar binomial expansion:
\begin{align}
    (x+y)^{2k} &= \sum_{q = 0}^{2k} \binom{2k}{q} x^{2k-q}y^q. 
\end{align}
For our matrices of interest, the expansion takes the form
\begin{align}
    (\vO_- + \vA_{j})^{\otimes k} [\cdot] (\vO_- + \vA_{j})^{\dagger \otimes k} & = \sum_{words} \vS_- \otimes \cdots \vA_j \otimes \vS_-\cdots \vA_j\vS_- [\cdot] \otimes (\vA_j \vS_-\cdots)^{\dagger}\\
    &=:\sum_{q=0}^{2k} \vM_q\tag{regroup by the occurrences $q$ of $\vA_j$}.
\end{align}

To control the norm of the operators $\vM_q$, we first polarize the error terms to turn the super-operator into a complex linear combination of completely positive maps (By~\autoref{lem:polarize}). There are different patterns of $\vS_-$ and $\vA_j$:
(1) If all $q$-many $\vA_j$ only paired with $\vS_-$, there are $q$ locations where we need to polarize 
\begin{align}
    &\labs{\tr[ \vO \BE \vS_- \otimes \cdots \vA_j \otimes \vS_-\cdots \vA_j\vS_- [\vrho] \otimes (\vA_j \vS_-\cdots)^{\dagger} ]}\\
    & = \labs{\BE_{s} s\tr[ \vO \CN_{2,S_-}\circ \CN_{2,(a_j\vA_j+ s s_-\vS_-)}\circ \cdots [\vrho] ]} \tag{By~\autoref{lem:polarize}}\\
    & \le \BE_{s} \lnormp{ \CN^{\dagger}_{2,S_-}\circ \CN^{\dagger}_{2,(a_j\vA_j+ s s_-\vS_-)}[\vI]}{\infty}\tag{Norm bounds for CP maps, $\labs{s} =1$}\\
    & \stackrel{N\rightarrow \infty}{=} \BE_{s} \btr[ \CN^{\dagger}_{2,S_-}\circ \CN^{\dagger}_{2,(a_j\vA_j+ s s_-\vS_-)}[\vI] ]\tag{By~\autoref{cor:Z+G}}\\
    & \stackrel{N\rightarrow \infty}{=}\BE\L[(\btr\L[\vS_-^{\dagger}\vS_-\R])^{k - q} \cdot \btr\L[ \labs{a_j}^2\vA_j^{\dagger}\vA_j + \labs{s_-}^2 \vS_-^{\dagger}\vS_-\R]^{q}\R]\\
    & = \BE_-(\btr\L[\vS_-^{\dagger}\vS_-\R])^{\frac{2k-q}{2}} \cdot \BE_j\btr\L[ \vA_j^{\dagger}\vA_j\R]^{q/2}\tag{Optimize parameters and exploit independence}\\
    & \stackrel{N\rightarrow \infty}{=} (1-\labs{w_j}^2)^{\frac{2k-q}{2}} \cdot  \labs{w_j}^{q} \le  \labs{w_j}^{q}\tag{By $\vA_j^{\dagger}\vA_j=\vI$ and \autoref{cor:Z+G}}.
    \end{align}
    The fourth inequality uses the permutation symmetry of both channels to use~\autoref{cor:Z+G}. The last equality minimizes the expression for suitable real random scalar variables
\begin{align}
 \labs{a_j}^2 = \sqrt{\frac{\btr[\vS_-^{\dagger}\vS_-]}{\btr[\vA_j^{\dagger}\vA_j]}} \quad \text{and}\quad \labs{s_-}^2 = \sqrt{\frac{\btr[\vA_j^{\dagger}\vA_j]}{\btr[\vS_-^{\dagger}\vS_-]}}.
\end{align}
The case where $\vA_j$ are paired up with $\vA_j^{\dagger}$ enjoys the same bound using that the normalization trace is deterministic in the large-$N$ limit for the Gaussian case
\begin{align}
    \lim_{N\rightarrow \infty}\btr[\vG^{\dagger}\vG] \stackrel{dist.}{=} 1. 
\end{align}

Likewise, $(\vO_- + \tvA_{j})^{\otimes k} [\cdot] (\vO_- + \tvA_{j})^{\dagger \otimes k} = \sum_{i=1}^m \tilde{\vM}_i$ with analogous bounds for the summands. 

To proceed, since the first and second (possibly conjugated) moments of the random matrices $\vA_i$ and $\tvA_i$ match for each index $i$, 
\begin{align}
    \Expect \vA^{(\dagger)}_i \otimes \vB &= \Expect \tvA^{(\dagger)}_i \otimes \vB; \\ 
    \Expect \vA_i^{(\dagger)} \otimes \vA^{(\dagger)}_i\otimes \vB  &=\Expect \tvA^{(\dagger)}_i \otimes \tvA^{(\dagger)}_i \otimes \vB \\
    \Expect \vA_i^{(\dagger)} \otimes \vA_i^{(\dagger)}\otimes \vA^{(\dagger)}_i\otimes \vB  &=\Expect \tvA^{(\dagger)}_i \otimes \tvA^{(\dagger)}_i \otimes \tvA_i^{(\dagger)}\otimes \vB \quad \text{for arbitrary $\vB$ independent from $\vA_i, \tvA_i$}.
\end{align}
Crucially, this implies that subtracting the expected moments completely \textit{cancels} the first to third order terms $\vM_1,\vM_2, \vM_3$ and  $\tilde{\vM}_1,\tilde{\vM}_2, \tilde{\vM}_3$
\begin{align}
    \labs{n_j-n_{j-1}} &=\labs{ \sum_{q=4}^{2k} \tr[\vO(\vM_q - \tilde{\vM}_q)[\vrho]]}\\
    &\stackrel{N\rightarrow \infty}{\le} \sum_{q=4}^{2k} \binom{2k}{q} \L(\btr\L[ \BE\vA_j^{\dagger}\vA_j\R]^{q/2} + \btr\L[\BE \tvA_j^{\dagger}\tvA_j\R]^{q/2}\R)\\
    &\le 2 \sum_{q=4}^{2k} \binom{2k}{q} \labs{w_j}^q\\
    &\le 2 \sum_{q=4}^{2k} (2k)^q\labs{w_j}^q \tag{Since $\binom{2k}{q}\le (2k)^q$}\\
    &\le \frac{(4k\labs{w_j})^4 + (4k\labs{w_j})^{2k}}{8}.
    \label{eq:update}
\end{align}
The last inequality bounds the geometric series using the elementary numerical inequality 
\begin{align}
    \sum_{q=4}^{2k} {x^q} = \sum_{q=4}^{2k} {2^{-q}(2x)^q} \le \frac{1}{8}( (2x)^4+(2x)^{2k})
    \quad\text{for $x \geq 0$.}
\end{align}
Sum over the index $j$ to conclude the proof.
\end{proof}

\section{Comparing the Gaussian CP map with a unitary channel}
\label{sec:comparingG_Haar}

\begin{lem}[Ginibre matrices are effectively Haar unitaries in the diamond distance]\label{lem:Ginibre_is_Haar}
\begin{align}
    \lim_{N\rightarrow \infty} \normp{\CN_{2k,\vU} - \CN_{2k,\vG}}{\diamond}=0
\end{align}    
\end{lem}
\begin{proof}
    Since the Ginibre ensemble is left and right unitarily invariant, its action maps between irreps of the commutant of the unitary group. In particular, this commutant algebra has bounded dimensions ($k!$ for fixed $k$) as the dimension $N$ goes to infinity. We express general normalized PSD symmetrized input in terms of diagrams
    \begin{align}
        \vrho = \L(\sum_{\vsigma}c_{\vsigma} \frac{\vO_{\vsigma}}{\normp{\vO_{\vsigma}}{2}}\R)\L(\sum_{\vsigma}c_{\vsigma} \frac{\vO_{\vsigma}}{\normp{\vO_{\vsigma}}{2}}\R)^{\dagger}\quad \text{such that}\quad\tr[\vrho]=1. 
    \end{align}
The normalization can be determined by 
    \begin{align}
        1=\tr\L[ \L(\sum_{\vsigma}c_{\vsigma} \frac{\vO_{\vsigma}}{\normp{\vO_{\vsigma}}{2}}\R)\L(\sum_{\vsigma}c_{\vsigma} \frac{\vO_{\vsigma}}{\normp{\vO_{\vsigma}}{2}}\R)^{\dagger}\R] &= \sum_{\vsigma} \labs{c_{\vsigma}}^2 + (\text{cross terms}.)\\
        &\stackrel{N\rightarrow\infty}{=} \sum_{\vsigma} \labs{c_{\vsigma}}^2.
    \end{align}
    The second line uses that, in the large $N$ limit, the cross terms vanish since different permutations are orthogonal in the limit
    \begin{align}
\tr[\vO_{\vsigma}\vO^{\dagger}_{\vsigma'}] = \CO\L(\frac{1}{N} \normp{\vO_{\vsigma}}{2}\normp{\vO_{\vsigma'}}{2}\R)\quad \text{if $\vsigma\ne \vsigma'$}. 
    \end{align}
    Now, we control the difference in 1-norm for any PSD input
\begin{align}
     \lnormp{ (\CN_{2k,\vG}-\CN_{2k,\vU})[\vrho]}{1} &\le \sum_{\vsigma,\vsigma'}\labs{c_{\vsigma}c_{\vsigma'}}\lnormp{ (\CN_{2k,\vG}-\CN_{2k,\vU})[\frac{\vO_{\vsigma}}{\normp{\vO_{\vsigma}}{2}}\frac{\vO^{\dagger}_{\vsigma'}}{\normp{\vO_{\vsigma'}}{2}}]}{1}\\
     &\le \L(\sum_{\vsigma,\vsigma'}\labs{c_{\vsigma}c_{\vsigma'}}\R) \sup_{\vsigma} \frac{\lnormp{ (\CN_{2k,\vG}-\CN_{2k,\vU})[\vO_{\vsigma}}{1}}{\normp{\vO_{\vsigma}}{1}}\tag{Since $\vO_{\vsigma}\vO^{\dagger}_{\vsigma'}=\vO_{\vsigma\vsigma'^{-1}}$}.
\end{align}
    
    Therefore, it suffices to compare the limiting behavior between two maps on each individual diagrams:
    \begin{align}
        \lnormp{ (\CN_{2k,\vG}-\CN_{2k,\vU})[\vO_{\vsigma}]}{1} &= \normp{ (\CN_{2k,\vG}-1)[\vO_{\vsigma}]}{1}\tag{Haar random channel acts identity}\\
        &=\lnormp{ (\CN_{2k,\vG}-1)[\vO'_{\vsigma}]}{1} + o(1)\norm{\vO_{\vsigma}}_1\tag{$\norm{\vO'_{\vsigma}-\vO_{\vsigma}}_1=o(1)\norm{\vO_{\vsigma}}_1$ is subleading in $N$}\\
        &=\lnormp{ (\CN_{2k,\vG,\vG',\cdots,}-1)[\vO'_{\vsigma}]}{1}+o(1)\norm{\vO_{\vsigma}}_1\tag{Decouple distinct blocks}\\
        &=\lnormp{ (\CN_{2k,\vG,\vG',\cdots,}-1)[\vO_{\vsigma}]}{1}+o(1)\norm{\vO_{\vsigma}}_1\tag{Restoring $\vO_{\vsigma}$}\\
        &=o(1)\norm{\vO_{\vsigma}}_1\tag{Since $\BE[\vG\cdot\vI\cdot\vG^{\dagger}] = \vI$}.
    \end{align}
    The second line makes the indices distinct at the $\frac{1}{N}$ cost in 1-norm 
    \begin{align}
      \vO_{\vsigma} = \sum_{\Pi\ge \vsigma} \vO'_{\Pi}\quad \text{such that}\quad \norm{ \vO_{\vsigma} -\vO'_{\vsigma}}_1 \le \frac{f(k)}{N} \norm{ \vO_{\vsigma}}_1
    \end{align}
    for some combinatorial factors $f(k)$.
    (Since each $\vO'_{\Pi}$ have all blocks propagating if $\Pi\ge \vsigma$, the 1-norm is $\norm{\vO'_{\Pi}}_1= M'_{\Pi}$, which is at least $\CO(\frac{1}{N})$ smaller than $M'_{\vsigma}$.) The third line is a decoupling argument: distinct rows of Ginibre matrices are independent. Therefore, the distinct blocks of $\vsigma$ (which are the $k$ wires connecting the bottom to the top) are effectively acted by \textit{independent} pairs of $\vG$s
    \begin{align}
        \CN_{2k,\vG}[\vO'_{\vsigma}] = \CN_{2k,\vG,\vG',\cdots,}[\vO'_{\vsigma}]
    \end{align}
where the non-CP map $\CN_{2k,\vG,\vG',\cdots,}$ depends on the blocks of $\vsigma$. The fourth line removes the distinctness constraint at error $\norm{\CN_{2k,\vG,\vG',\cdots,}[\vO_{\vsigma}-\vO'_{\vsigma}]}_{1}\le \norm{\vO_{\vsigma}-\vO'_{\vsigma}}_{1}$ (using asymptotic CPTP property of $\CN_{2k,\vG,\vG',\cdots,}$~\autoref{lem:freeness_random_P_NI}).
To conclude the proof for the diamond norm, consider arbitrary inputs with positive and negative parts, recall that the prefactor $
\sum_{\vsigma,\vsigma'}\labs{c_{\vsigma}c_{\vsigma'}}$ is finite in the large-$N$ limit, and uses that the ancilla dimension for the diamond norm is also independent of $N$ (similar to~\autoref{lem:small_ancilla}).
\end{proof}

\section{\texorpdfstring{Unitary $k$-designs from $k'$-wise permutations}{Unitary k-designs from k'-wise permutations}}\label{sec:optimal_k_design}

Our construction of random unitaries from random permutations can recover efficient unitary $k$-designs in the diamond distance by replacing the truly random phased permutations with $\tCO(k)$-wise independent phased permutations. The goal of this section is to prove this argument as the detailed version of~\autoref{cor:k_design_from_k_wise}.

\begin{cor}[Unitary $k$-designs from $k'$-wise independence]\label{cor:k_design_from_k_wise_full}
There is a constant $c$ such that for $k \le 2^{cn}$, an $\epsilon$-approximate quantum unitary $k$-design can be implemented
using 
\begin{align}
\CO(\log^2(k/\epsilon))&\quad \text{one and two-qubit gates and }\\
\CO(\log^2(k/\epsilon))&\quad \text{controlled-queries}
\end{align}
to $\CO(\log(k/\epsilon))$-many independent samples of $\CO(\epsilon N^{-4k'})$-approximate $k'$-wise random permutations and their inverses, and to the $\CO(\log(k/\epsilon))$ independent samples of exact $2k'$-wise random functions $[N]\rightarrow [2k'+1]$ for some almost linear
\begin{align}
    k' = \CO(k\log(k/\epsilon)).
\end{align}
\end{cor}

The first step of the argument is to replace the truly random phases in each term of the product with $2k'$-wise independent phases.
If the classical $2k'$-wise independent phased permutations are exact, then the substitution argument is immediate (even with inverses;~\eqref{eq:SST}). Indeed, this can be done exactly for the phases by replacing them with efficiently implementable, exact $2k'$-wise independent phases.
\begin{lem}[$2k'$-wise independent phases]
    For any integer $0<k'<2^{cn}$, consider the set of $ N$ evenly spread points on the unit circle
    \begin{align}
        S_{N} := \{ 1, \e^{\ri \frac{1}{N}2\pi},\e^{\ri \frac{2}{N}2\pi},\cdots\}.
    \end{align}
    Then, a $2k'$-wise independent function $f: [N] \rightarrow S_{N}$ acting on the diagonal matches exactly the $2k'$ (mixed) moments of diagonal random phases 
    \begin{align}
        \BE[ \vD_{f}^{(\dagger)2k'} ] = \BE[ \vD_{z}^{(\dagger)2k'} ].
    \end{align}
    Further, $\vD_f$ can be implemented using $\tCO(nk')$ gates.
\end{lem}
\begin{proof}
    First, we use $2k'$-wise independence to replace $f$ with truly random functions on the diagonal. Now that each of the entries is decoupled, and we merely require the discrete points on the unit circle to reproduce the moments,  
    \begin{align}
        \BE_{S_{N}}[z^{p}] =  \BE_{\vU(1)}[z^{p}] = 0 \quad \text{for each integer}\quad \labs{p} \le 2k'
    \end{align}
    which can be seen from the properties of the discrete Fourier transform. The implementation cost has two parts: first, the cost of classically evaluating a $k'$-wise independent function $f:[N] \rightarrow [N]$, which is merely $\tilde{O}(nk')$. This is obtained using the classic construction using a random degree $k'-1$ polynomial over the field $\mathbb{F}_{2^n}$ ~\cite{Joffe, wegman1981new,ALON1986fastsimple,Alon2012Almostkwise}).
    Second, is the cost of then applying these values in $[N]$ as phases in $S_N$, which can be done by the usual controlled tower of geometrically decreasing rotations at cost $O(n)$.
    Third, we then uncompute the function $f$. 
    The net cost is $\tilde{O}(nk')$.
\end{proof}

The second step of the argument is to replace the truly random permutations with (approximate) $2k'$-wise independent permutations.
While the independence of the phases can be exact, the best known efficient classical $k$-wise independent permutation for large values of $k$ is approximate. Thus, we have to demand the classical $k$-wise independent permutation to achieve very high precision at very low costs to compensate for the distance conversion rates against quantum queries.
\begin{defn}[Classical $\epsilon$-approximate $k$-wise independent permutations]\label{defn:classical_kwise}
    We say an ensemble of permutations gives a classical $\epsilon$-approximate $k$-wise independent permutations if
    \begin{align}
        \lnormp{\BE_{\vS'} \vS^{'k}\ket{\vec{i}} - \BE_{\vS} \vS^{k}\ket{\vec{i}}}{1} \le \epsilon \quad \text{for each}\quad \vec{i}\in [1,N]^k.\label{eq:Sk_eps}
    \end{align}
\end{defn}
In the above, we defined the vector $\ell_1$ norm
\begin{align}
        \lnormp{ \sum_{\vec{i}} c_{\vec{i}}\ket{\vec{i}}}{1} := \sum_{\vec{i}} \labs{c_{\vec{i}}}
\end{align}
which captures the statistical distance when applied to the difference of two probability distributions.

In our use cases, the above classical statistical notion of distance is not immediately applicable because (1) the quantum inputs can be entangled and (2) our implementation of the unitary makes adaptive queries to both $\vS$ and $\vS^{-1}$. Fortunately, the above complication can be circumvented entirely by a powerful black-box result that allows us to effectively use $2k$-wise independent permutations as \textit{exact} in other parts of the paper, provided the precision is high enough.
\begin{thm}[{\cite[Theorem 1.1]{Alon2012Almostkwise}}]\label{thm:approx_to_exact}
    Suppose there exists an $\epsilon$-approximate classical $k$-wise independent permutation $\vS_{approx}$. Then, there exists an exact classical $k$-wise independent permutation $\vS_{exact}$ that is $\CO(\epsilon\cdot N^{4k})$-statistically close on the distribution of permutations.
\end{thm}

This Theorem of \cite{Alon2012Almostkwise} is very powerful, as it allows one to transfer the error $\epsilon$ in the $k$-design to an error in the distinguishing probability of the algorithm. This is by a simple hybrid argument: it's easy to see that for any distributions $D,D'$ on inputs, the probability a quantum (or classical) algorithm accepts on a random input from $D$ vs $D'$ is upper bounded by twice their total variation distance\footnote{This is simply because the probability the quantum algorithm accepts is a bounded function in $[0,1]$ over the domain -- so a change of the input distribution by $\epsilon$ can change the value of the expectation by at most $2\epsilon$. }. So given a quantum algorithm $\vA$ trying to distinguish our ensemble (built from $\epsilon$-approximate $k$-wise independent permutations) from Haar, substituting in our ensemble with the same one built with $\vS_{exact}$ at most changes the acceptance probability\footnote{Here one must invoke data processing inequality as well to say the ``lifted'' ensembles on unitaries are also close in total variation distance.} (diamond norm value) by $\epsilon N^{4k}$. As we will discuss shortly, $\epsilon$ can be set to be exponentially small which makes this change negligible.

\begin{proof}[Proof of~\autoref{cor:k_design_from_k_wise}]
Recall that our construction
\begin{align}
 \vV:=\vZ_L \L( \prod_{j=1}^{\ell} \e^{\ri \theta_m \vA^{(j)}_m} \R)\vZ_R,
\end{align}
    can be implemented to precision $\CO(\epsilon/k)$ using $2+m\ell$-many random samples of $k'$-wise independent permutations and their inverses (\autoref{prop:cost_exponential}) with
\begin{align}
    k' = \CO\L( k\cdot \underset{\text{queries per $\e^{\ri \theta_m\vA_m}$}}{\underbrace{(\sqrt{m} +\log(k\ell/\epsilon))}} \R).
\end{align}
In particular, we may choose (assuming $k\le 2^{cn}$), 
\begin{align}
    m = 2,\quad \ell = \CO( \log(k/\epsilon)) \quad \text{such that}\quad \normp{\CN_{2k,\vV}-  \CN_{2k,\vU_{Haar}}}{\diamond}= \CO\L(\frac{k^8\ell^8(m^4n^8+ n^{8})}{N}+\frac{k^4}{m^{\ell}}\R)\le \CO(\epsilon).
\end{align}

If the classical $k$-wise independent permutations are exact, then the claim follows even in the presence of inverses $\vS^{-1}$ since the partial transpose is linear (and that $\vS^{-1}=\vS^{T}$)
\begin{align}
    \BE_{\vS}[\vS^{\otimes k}] &= \BE_{\vS'}[\vS^{'\otimes k}]\\
    \text{implies}\quad \BE_{\vS}[\vS^{\otimes k_1}\otimes (\vS^{T})^{\otimes (k- k_1)}] &= \BE_{\vS'}[\vS^{'\otimes k_1}\otimes (\vS^{'T})^{\otimes (k- k_1)}].\label{eq:SST}
\end{align}
Now, by~\autoref{thm:approx_to_exact}, we may simply replace the truly random permutation $\vS$ with the exact $k$-wise independent permutations while paying an additive error $\epsilon$ in the distinguishing probability between $\vS_{exact}$ and $\vS_{approx}$. 
\end{proof}
Lastly, we invoke the highly nontrivial construction of efficient $k$-wise independent permutations~\cite{Kassabov2007Sym,caprace2023tame} as a black box. Generators that give Cayley graph expanders for the symmetry group (and others) have been known~\cite{Kassabov2007Sym,caprace2023tame,breuillard2011suzuki}. These explicit constructions~\cite{Kassabov2007Sym,caprace2023tame} can also be efficiently evaluated (given $n$-bit inputs)~\cite{explicit_circuit2024}. For our purposes, the recent elegant construction (\cite[Theorem 1.5]{caprace2023tame})\footnote{We thank Martin Kassabov for this pointing us to this reference.} can be efficiently computed classically using finite field arithmetics on large primes, which in turn gives the following classical circuit complexity on $k$-wise independent permutations on $2^n$ many elements.

\begin{thm}[Efficient classical approximate $k$-wise independent permutations{~\cite{Kassabov2007Sym,caprace2023tame,explicit_circuit2024}}]\label{thm:efficient_k_wise}

There exist classical $\epsilon$-approximate $k$-wise independent permutations on $2^n$ elements with time complexity 
\begin{align}
    \CO\bigg(\poly(n) \L(nk+ \log(\frac{1}{\epsilon})\R)\bigg).
\end{align}
\end{thm}
Therefore, we can take a small error $\epsilon = 1/2^{\Omega(nk)}$ to feed into~\autoref{thm:approx_to_exact} to obtain nearly exact $k$-wise independent permutations at the overall cost $\CO( k\cdot\poly\log(n))$.

\subsection{Bootstrapping nonadaptive designs to adaptive designs}

Throughout the paper, we have mostly discussed nonadaptive designs in the sense of the diamond norm. This can be bootstrapped to adaptive security at high enough precision (for example, by iterating the same random unitary many times~\cite[Lemma XI.7]{chen2024efficient}.).

\begin{lem}[From nonadaptive $1$-norm to adaptive queries]
\label{lem:1norm_to_adaptive}
Consider two unitary ensembles. The distinguishing probability for any quantum algorithm using $k$ adaptive queries can be controlled by trace distance for parallel distinguishability with additional dimensional factors:
\begin{align}
    \labs{\BE\CA(\CN_{2k,\vU}) -\BE\CA(\CN_{2k,\vV})}\le 4^k N^{4k}\cdot \normp{ \BE\CN_{2k,\vU} - \BE\CN_{2k,\vV}}{1-1}.
\end{align}
The above continues to hold if the unitaries at each query are different with $\vU_1,\cdots, \vU_k$ and $\vV_1,\cdots, \vV_k$ on both RHS and LHS.
\end{lem}
\begin{proof}
    Consider the super operator entries of the channel
\begin{align}
   \CA(\CN_{2k,\vU}) -\CA(\CN_{2k,\vV}) = \sum_{ \vec{i'}\vec{i}\vec{j}\vec{i'}} \alpha_{\vec{i'}\vec{i}\vec{j}\vec{i'}}  \bra{\vec{i'}} \delta\CN[ \ket{\vec{i}}\bra{\vec{j}} ]\ket{\vec{j'}}. 
\end{align}
for $\vec{j} = (j_1,\cdots,j_k) \in [N]^{k}$. Expecting dimensional factors, we will content ourselves with bounding the entry-wise coefficients $\alpha_{\vec{i'}\vec{i}\vec{j}\vec{i'}}$. To do so,
consider a polarization argument (\autoref{lem:polarize})
for each system ($i,i',j,j'\in [1,N]$) that turns it into a weighted sum over CP maps
\begin{align}
\ket{i'}\bra{i}\otimes \ket{j}\bra{j'} = \BE_{z=\pm1,\pm \ri} \L[ z\bigg(\ket{i'}\bra{i} +z\ket{j'}\bra{j}\bigg)\otimes \bigg(\ket{i'}\bra{i} +z\ket{j'}\bra{j}\bigg)^{\dagger}\R] .
\end{align}
In particular, the Kraus operator is bounded as
\begin{align}
    \bigg(\ket{i'}\bra{i} +z\ket{j'}\bra{j}\bigg)^{\dagger}\bigg(\ket{i'}\bra{i} +z\ket{j'}\bra{j}\bigg) \le 4\vI.
\end{align}

Thus, the Kraus operator tensored across systems
\begin{align}
    \vO_{\vec{z}}:=\bigg(\ket{i_1'}\bra{i_1} +z_1\ket{j'_1}\bra{j_1}\bigg)\otimes \cdots \otimes \bigg(\ket{i_k'}\bra{i_k} +z_k\ket{j'_k}\bra{j_k}\bigg)
\end{align}
and the corresponding CP map (with suitable normalization) 
\begin{align}
    \CN_{\vec{z}}: = \frac{1}{4^k}\vO_{\vec{z}}[\cdot]\vO_{\vec{z}}^{\dagger}  
\end{align}
define a trace-non-increasing map since $\CN_{\vec{z}}^{\dagger}[\vI] \le \frac{1}{4^k} (4 \vI)^{\otimes k} \le \vI$,  and thus $\CA(\CN_{\vec{z}})$ is a ``valid experiment,'' with some possibility of failure. Thus $\labs{\CA(\CN_{\vec{z}})}\le 1$. 

Combining the above, we obtain the entry-wise bound
\begin{align}
    \labs{\alpha_{\vec{i'}\vec{i}\vec{j}\vec{i'}}} = \labs{ 4^{k}\BE_{\vec{z}} z_1\cdots z_{k}\CA(\vO_{\vec{z}})}
    \le 4^{k}\BE_{\vec{z}}\labs{ \CA(\vO_{\vec{z}})}\le 4^k.
\end{align}
Altogether,
\begin{align}
    \labs{\CA(\CN_{2k,\vU}) -\CA(\CN_{2k,\vV})} &= \labs{ \sum_{ \vec{i'}\vec{i}\vec{j}\vec{i'}} \alpha_{\vec{i'}\vec{i}\vec{j}\vec{i'}}  \bra{\vec{i'}} \delta\CN[ \ket{\vec{i}}\bra{\vec{j}} ]\ket{\vec{j'}}}\\
    &\le \sum_{ \vec{i'}\vec{i}\vec{j}\vec{i'}} \abs{\alpha_{\vec{i'}\vec{i}\vec{j}\vec{i'}}}\normp{\delta\CN}{1-1}\\
    &\le 4^k N^{4k} \normp{\BE\CN_{2k,\vU} - \BE\CN_{2k,\vV}}{1-1}
\end{align}
as advertised. In fact, the identical argument holds if the unitaries at each query are different with $\vU_1,\cdots, \vU_k$ and $\vV_1,\cdots, \vV_k$ on both RHS and LHS.
\end{proof}

\bibliographystyle{amsalpha}
\bibliography{ref}

\end{document}